\documentclass[acmsmall,screen,nonacm,review=false,timestamp=false]{acmart}
\acmArticle{1}
\usepackage{graphicx} % Required for inserting images
\usepackage{pifont}
\usepackage{algorithm}% http://ctan.org/pkg/algorithms
\usepackage[noend]{algpseudocode}% http://ctan.org/pkg/algorithmicx
\usepackage{xspace}
\usepackage{enumitem}
\usepackage{xcolor}
\usepackage{proof}
\usepackage{amsmath,amsthm}
\usepackage{hyperref}
\usepackage{bbding}

\usepackage{bm}
\usepackage{listings}
\usepackage{newunicodechar}
\usepackage{tikz}
\usetikzlibrary{automata,positioning}
\usepackage{scalerel}
\usepackage{url}
\usepackage[only,llbracket,rrbracket,llparenthesis,rrparenthesis]{stmaryrd}
\usepackage{accsupp} 
\usetikzlibrary{cd}
\usepackage{algorithm}
\usepackage{stmaryrd}
\usetikzlibrary{tikzmark}
\usepackage[normalem]{ulem} 
\usepackage{color}
\definecolor{keywordcolor}{rgb}{0.1, 0.5, 0.1}      % green 
\definecolor{tacticcolor}{rgb}{0.0, 0.1, 0.6}    % blue
\definecolor{commentcolor}{rgb}{0.4, 0.4, 0.4}   % grey
\definecolor{symbolcolor}{rgb}{0.0, 0.1, 0.6}    % blue
\definecolor{sortcolor}{rgb}{0.7, 0.1, 0.1}   % red
\definecolor{attributecolor}{rgb}{0.7, 0.1, 0.1} % red
\usepackage{setspace}

\usepackage{fancyhdr}
\AtBeginDocument{%
    \addtolength{\footskip}{2.0\baselineskip}%
    \fancyfoot[L]
}
\newunicodechar{⊆}{\ensuremath{\subseteq}}
\newunicodechar{∃}{\ensuremath{\exists}}
\newunicodechar{↦}{\ensuremath{\mapsto}}
\newunicodechar{∧}{\ensuremath{\land}}
\newunicodechar{⇐}{\ensuremath{\Leftarrow}}
\newunicodechar{⬝}{\ensuremath{\bullet}}
\newunicodechar{↼}{\ensuremath{\leftharpoonup}}
\newunicodechar{Ξ}{\ensuremath{\mathbb{F}}}
\newunicodechar{β}{\ensuremath{\mathrm{\beta}}}
\newunicodechar{α}{\ensuremath{\alpha}}
\newunicodechar{⟧}{\ensuremath{\rrbracket}}
\newunicodechar{⟦}{\ensuremath{\llbracket}}
\newunicodechar{Π}{\ensuremath{\mathrm{\Pi}}}
\newunicodechar{∀}{{{$\forall$}}}
\newunicodechar{λ}{\ensuremath{\lambda}}
\newunicodechar{Φ}{\ensuremath{\mathrm{\Phi}}}
\newunicodechar{⊨}{\ensuremath{\vDash}}
\newunicodechar{⦃}{\ensuremath{\lBrace}}
\newunicodechar{⦄}{\ensuremath{\rBrace}}
\newunicodechar{∅}{\ensuremath{\emptyset}}
\newunicodechar{∪}{\ensuremath{\cup}}
\newunicodechar{∩}{\ensuremath{\cap}}
\newunicodechar{Δ}{\ensuremath{\mathrm{\Delta}}}
\newunicodechar{δ}{\ensuremath{\mathrm{\delta}}}
\newunicodechar{γ}{\ensuremath{\gamma}}
\newunicodechar{↔}{\ensuremath{\leftrightarrow}}
\newunicodechar{∈}{\ensuremath{\in}}
\newunicodechar{φ}{\ensuremath{\phi}}
\newunicodechar{⇈}{\ensuremath{\Uparrow}}
\newunicodechar{⊥}{\ensuremath{\bot}}
\newcommand{\conf}{\ensuremath{\mathfrak{c}}\xspace}
\newcommand{\toolName}{\textit{MMINT-A}}
\newcommand{\toolNamePL}{\textit{MMINT-PL}}
\usepackage{xspace}

\newcommand{\nc}[1]{{{#1}}}

\newcommand{\elem}[1]{\ensuremath{{\tt #1}}}
\newcommand{\token}[1]{{\elem{#1}}}

\newcommand{\AC}{\ensuremath{{\tt AC}}}
\newcommand{\Goal}{\ensuremath{{\tt Goal}}}
\newcommand{\Evd}{\ensuremath{{\mathcal{E}}}}
\newcommand{\Supp}{\ensuremath{\mathsf{Supp}}}
\newcommand{\Var}[1]{\ensuremath{#1_{\textbf{PL}}}}
\newcommand{\pl}[1]{\ensuremath{\textbf{#1}}}
\newcommand{\Inv}[3]{\ensuremath{\mathsf{Inv}(#1,#2,#3)}}

\newcommand{\xmark}{\text{\ding{55}}}
\newcommand{\cmark}{\text{\ding{51}}}

\algnewcommand\algorithmicswitch{\textbf{switch}}
\algnewcommand\algorithmiccase{\textbf{case}}
\algnewcommand\algorithmicassert{\texttt{assert}}
\algnewcommand\Assert[1]{\State \algorithmicassert(#1)}%
\algdef{SE}[SWITCH]{Switch}{EndSwitch}[1]{\algorithmicswitch\ #1\ \algorithmicdo}{\algorithmicend\ \algorithmicswitch}%
\algdef{SE}[CASE]{Case}{EndCase}[1]{\algorithmiccase\ #1}{\algorithmicend\ \algorithmiccase}%
\algtext*{EndSwitch}%
\algtext*{EndCase}%
\algnewcommand{\IfThenElse}[3]{% \IfThenElse{<if>}{<then>}{<else>}
\algorithmicif\ #1\ \algorithmicthen\ #2\ \algorithmicelse\ #3}

\algnewcommand{\IfSingle}[2]{
\algorithmicif\ #1\ \algorithmicthen\ #2\
}
\algnewcommand{\ElsIfSingle}[2]{
 \algorithmicelse\ \algorithmicif\ #1 \algorithmicthen\ #2\
}
\newtheorem*{lemma*}{Lemma}
\newtheorem*{theorem*}{Theorem}
\newtheorem*{proposition*}{Proposition}

\newcommand{\ellipses}[0]{,\ldots,}

\newcounter{typofixcounter}

\newcommand{\typofix}[3][]{%
  \refstepcounter{typofixcounter}
  \if\relax\detokenize{#1}\relax
  \else
  \fi
  {#3}%
}

\newcommand{\weakNew}[1]{{#1}}
\newcommand{\weakDelete}[1]{}
\newcommand{\weakChange}[2]{#2}

\newcommand{\typoRef}[1]{\hyperref[#1]{\textbf{T\ref{#1}}}}

\newcounter{newTextCounter}

\newcommand{\newText}[2][]{%
  \refstepcounter{newTextCounter}%
 \if\relax\detokenize{#1}\relax
  \else
    \label{#1}%
  \fi
  {#2}%
}

\newcommand{\newTextRef}[1]{\hyperref[#1]{\textbf{N\ref{#1}}}}

\newcounter{changeCounter}

\newcommand{\change}[3][]{%
  \refstepcounter{changeCounter}
  \if\relax\detokenize{#1}\relax
  \else
    \label{#1}
  \fi
  {#3}%
}

\newcommand{\changeRef}[1]{\hyperref[#1]{\textbf{C\ref{#1}}}}

\newcounter{deleteCounter}

\newcommand{\delete}[2][]{%
  \refstepcounter{deleteCounter}
  \if\relax\detokenize{#1}\relax
  \else
    \label{#1}
  \fi
}
\newcommand{\deleteRef}[1]{\hyperref[#1]{\textbf{D\ref{#1}}}}
\newcommand{\ntn}[1]{{#1}}
\newcommand{\modlang}{\ensuremath{\mathcal{M}}}
\newcommand{\template}[2]{\ensuremath{T}_{#1,#2}}
\begin{document}

\title{Assurance Case Development for Evolving Software Product Lines: A Formal Approach}

\author{Logan Murphy}
% \orcid{XXXXXXXX}
\email{lmurphy@cs.toronto.edu}
\affiliation{%
  \institution{University of Toronto}
  \city{Toronto}
  \country{Canada}
}

\author{Torin Viger}
\email{tviger@cs.toronto.edu}
\affiliation{%
  \institution{University of Toronto}
  \city{Toronto}
  \country{Canada}
}
\author{Alessio Di Sandro}
\email{adisandro@cs.toronto.edu}
\affiliation{%
  \institution{University of Toronto}
  \city{Toronto}
  \country{Canada}
}
\author{Aren A. Babikian}
\email{babikian@cs.toronto.edu}
\affiliation{%
  \institution{University of Toronto}
  \city{Toronto}
  \country{Canada}
}
\author{Marsha Chechik}
\email{chechik@cs.toronto.edu}
\affiliation{%
  \institution{University of Toronto}
  \city{Toronto}
  \country{Canada}
}

\renewcommand{\shortauthors}{Murphy et al.}

\begin{abstract}
  {In critical software engineering, structured assurance cases (ACs) are used to demonstrate how key \weakNew{system} properties are supported by evidence (e.g., test results, proofs). Creating rigorous ACs is particularly challenging in the context of software product lines (SPLs), i.e, sets of software products with overlapping but distinct features and behaviours. Since SPLs can encompass very large numbers of products, developing a rigorous AC for each product individually is infeasible. Moreover, if the SPL evolves, e.g., by the modification or introduction of features, it can be infeasible to assess the impact of this change. Instead, the development and maintenance of ACs ought to be \emph{lifted} such that a single AC can be developed for the entire SPL simultaneously, and  be analyzed for regression in a variability-aware fashion. In this article, we describe a formal approach to lifted AC development and regression analysis. We formalize a language of variability-aware ACs for SPLs and study the lifting of template-based AC development. We also define a regression analysis to determine the effects of SPL evolutions on variability-aware ACs. We describe a model-based assurance management tool which implements these techniques, and illustrate our contributions by developing an AC for a product line of medical devices.
  
  }
\end{abstract}




% \input{coverletter}
% \newpage 
% \input{changes}
% \newpage 
\maketitle

% \vspace{-0.3in}
\section{Introduction}
\label{sec:intro}
{{In many software engineering contexts, stakeholders require \emph{assurance} that software products will operate as intended. Several industries (e.g., automotive), have developed safety standards (e.g., ISO 26262 ~\cite{iso26262}) requiring careful documentation of verification activities via \emph{assurance cases} (ACs)~\cite{rushby2015interpretation}. ACs use structured argumentation to refine system-level requirements into lower-level specifications that be supported directly by evidence artifacts (e.g., tests, proofs). Given their role in safety and reliability engineering, great care must be taken to verify the correctness of an AC prior to system deployment. While ACs are traditionally developed and reviewed manually, using natural language, portions of ACs can be sufficiently formalized that formal methods can be leveraged to verify their correctness~\cite{viger2023foremost,varadarajan2023clarissa}.}} Nonetheless, the development of rigorous ACs, and systematically analyzing their \emph{regression}, remains a challenging problem.

{Simultaneously, software systems are increasing in scale and complexity. In many cases, companies are not developing individual software products, but a {family} of related products, i.e., a {software product line} (SPL)~\cite{apel2016feature}.} An SPL encompasses a set of {products} with distinct, but often overlapping, functionalities and behaviours. An SPL can be implemented in any number of ways, for instance, by annotating lines of code \weakChange{by}{with} \emph{presence conditions} indicating in which products a given statement is included.  SPLs can become very large in size, representing, e.g., millions of distinct products. When developing SPLs in assurance-critical contexts, assurance must be obtained for each product in the SPL. However, \weakDelete{given the complexity of SPLs,} developing and verifying individual ACs for each product is infeasible.

{In SPL engineering (SPLE), \emph{lifting} is the process of redefining a product-based analysis (e.g., testing, model checking) so that it can be applied directly to product lines~\cite{MURPHY2025112280}. In the lifting process, the semantics of the analysis are made \emph{variability-aware}, such that the variability of the SPL (e.g., as indicated by presence conditions) is handled explicitly. Applying a lifted analysis produces a \emph{variability-aware} analysis result, which concisely encodes distinct analysis results for each product of the SPL. Lifted analyses generally offer dramatic speedups compared to product-based methods.}

\nc{The same principle of {lifting} can be applied to AC development and regression, as illustrated in Fig.~\ref{fig:liftvsbrute}. Traditional product-based methods are shown in the bottom half\weakDelete{, coloured in blue}. Given a software product, we can used existing AC development methods to obtain a product AC. Thereafter, given some evolution of the product, we can attempt to determine the \emph{regression} of its assurance, wherein parts of the AC \weakChange{are found to be}{may be identified as} out of date and requiring revision. By contrast, lifted AC development (shown in the top half\weakDelete{, in pink}), proceeds directly from the software product line, producing a \emph{product line of ACs}, one for each product in the SPL. Then, given an evolution of the product line, lifted regression analysis determines the impact of this evolution on the assurance of each product. The two layers are formally related through \emph{product derivation}, indicated with dashed arrows. For lifted development and regression to be correct, the diagram in Fig.~\ref{fig:liftvsbrute} must commute; that is, applying lifted techniques and deriving a result for some product $p$ must give the same outcome as deriving $p$ from the SPL and applying product-based techniques.}

In a recent paper~\cite{ifm2024}, we developed a formal framework for lifted AC development, the cornerstone of which is a language of variability-aware ACs. \typofix{Building on recent work in formal methods for AC development, leverages formal AC \emph{templates} to develop formally verifiable ACs}{Building on recent work in formal methods for AC development, our framework uses formal AC \emph{templates} to enable the development of verifiable ACs.}~\cite{viger2020just,viger2023foremost,Denney:2018}.  To support lifted AC development, we studied the \weakDelete{systematic} \emph{lifting} of AC templates, and in particular, the use of AC templates which integrate lifted software analyses (e.g., lifted model querying~\cite{di2023adding}). These contributions were formalized and verified using the proof assistant Lean~\cite{lean4}. We also described a model-based AC management tool supporting lifted AC development, and \weakChange{illustrated}{demonstrated} our methodology and tooling by developing an AC for a product line of medical devices. 

This article extends our prior work~\cite{ifm2024} in the following ways: (1) In ~\cite{ifm2024}, only a high-level sketch of the formalization was provided, with the complete details being available through the implementation in Lean. Here, we provide a complete formalization of lifted AC development, including the formal semantics of variational ACs and lifted AC templates; (2) We extend our formalization to consider the impact of SPL evolutions on variational ACs; (3) We describe additional functionalities which have been implemented in our AC management tool to perform regression analysis in response to product line evolutions; (4) We extend our case study to \weakChange{illustrate}{demonstrate} contributions (2) and (3) over evolutions \typofix{to}{of} the product line of medical devices.

\begin{figure}[t]
    \centering
\includegraphics[width=\textwidth]{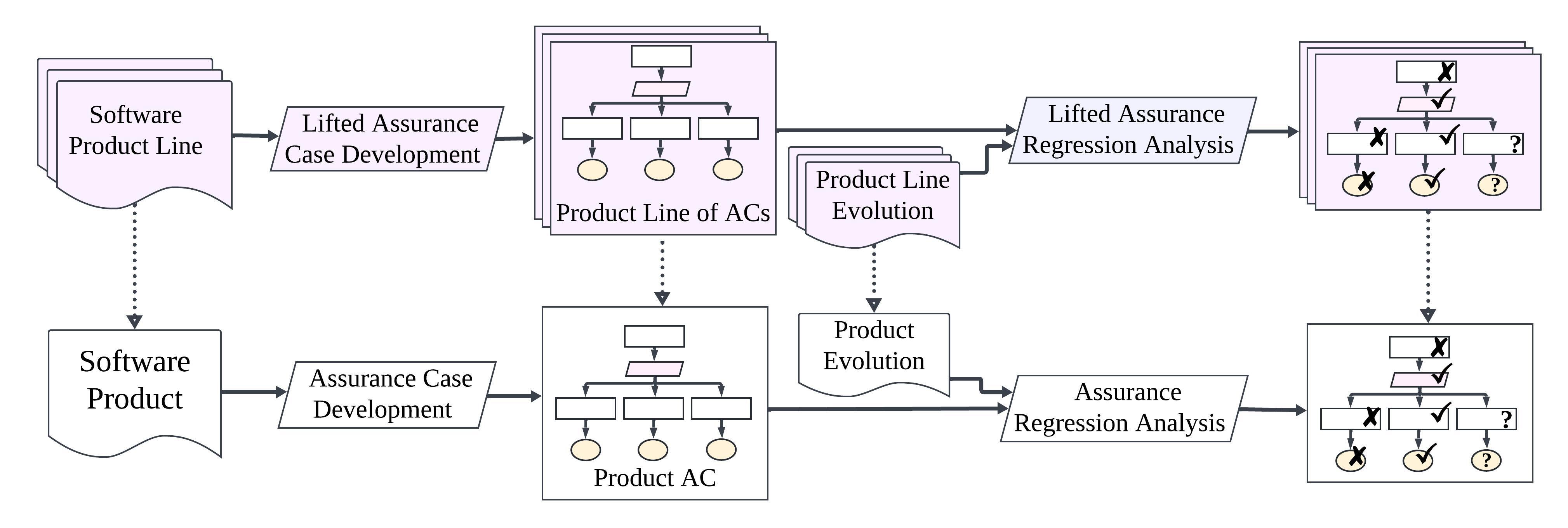}
    \vspace{-0.2in}
    \caption{\bf Lifted development and regression of assurance cases for evolving SPLs (pink), compared with product-based methods (blue). Derivation of products is indicated by dotted arrows. The labels $\{\cmark, \xmark, \textbf{?}\}$ indicate regression (or lack thereof) of assurance following an evolution.}
    % \vspace{-0.25in}
    \label{fig:liftvsbrute}
    \Description[]{}
\end{figure}
\nc{This article is structured as follows. We begin by providing the necessary background material on assurance cases and product lines (Sec.~\ref{sec:background}). We \weakNew{then} formalize a simple assurance case language, and define a formal framework for product-based AC development through verifiable templates (Sec.~\ref{sec:product_based_acs}). This framework is \weakNew{then} \emph{lifted} from the product level to the product line level (Sec.~\ref{sec:lifting}). We {\ntn{then}} extend our framework to consider the regression of assurance following evolutions of product line models. This is done in two stages: we first define a formal regression analysis for product-based ACs  (Sec.~\ref{sec:product_regression}), and subsequently lift this regression analysis for product lines of ACs (Sec.~\ref{sec:liftedRegression}). We describe tool support for lifted AC development and regression analysis, and demonstrate our contributions over a case study in which we develop and evolve an AC for a product line of medical devices (Sec.~\ref{sec:tool}). We provide an overview of related work (Sec.~\ref{sec:related}) before concluding with an outlook on future research directions (Sec.~\ref{sec:conclusion}). } 

% \vspace{-0.1in}
\section{Background}
\label{sec:background}
In this section, we provide background on structured assurance cases (Sec.~\ref{sec:background:model_based_acs}) and software product lines (Sec.~\ref{sec:productLines}).

\paragraph{Notation and Terminology.}
 A \emph{predicate} $P$ over $X$ is a function $P : X \to  \{\top,\bot\}$, where $\top$ and $\bot$ denote truth and falsity, respectively. \change[notation:pset]{For any set $S$, we use $[S]^*$ to denote the set of all \emph{finite} subsets of $S$}{For any set $U$, we use $\mathbb{P}(U)$ to denote the power set of $U$, and $\mathbb{F}(U)$ to denote the set of \emph{families} of subsets of $U$, i.e., $\mathbb{F}(U) = \mathbb{P}({\mathbb{P}(U)})$}. We use the\typofix{follow}{following} conventions to denote product line entities: boldface (e.g., $\pl{x}$) denotes a product line object; \delete[vartypenotation]{given a set $X$, $\Var{X}$ denotes the set of ``product lines of $X$'s'';} a superscript arrow on a function (e.g., $f^\uparrow$) indicates that this function is \emph{lifted}, i.e., is a modified version of function $f$ which can be applied to product lines. \delete[featexpr:fwdref]{Lowercase $\phi$ always denotes a feature expression, and $\psi$ always denotes an analysis specification (e.g., an LTL formula). Uppercase $\Phi$ always denotes a feature model.} Throughout this work, we  use the term ``variational'' to mean an object imbued with variability, generally coinciding with the term ``product line'', e.g., a variational assurance case is the same thing as a product line of assurance cases. A \emph{variability-aware} analysis is one which explicitly interprets the variability encoded in a product line\weakChange{. By contrast, a ``product-based'' language is one in which there is no explicit representation of variability, and}{, whereas} a product-based analysis \weakDelete{or function} is one which can be applied only to products (as opposed to product lines).
 
\subsection{Structured Assurance Cases}
\label{sec:background:model_based_acs}

\begin{figure}[t]
    \centering
\includegraphics[width=\textwidth]{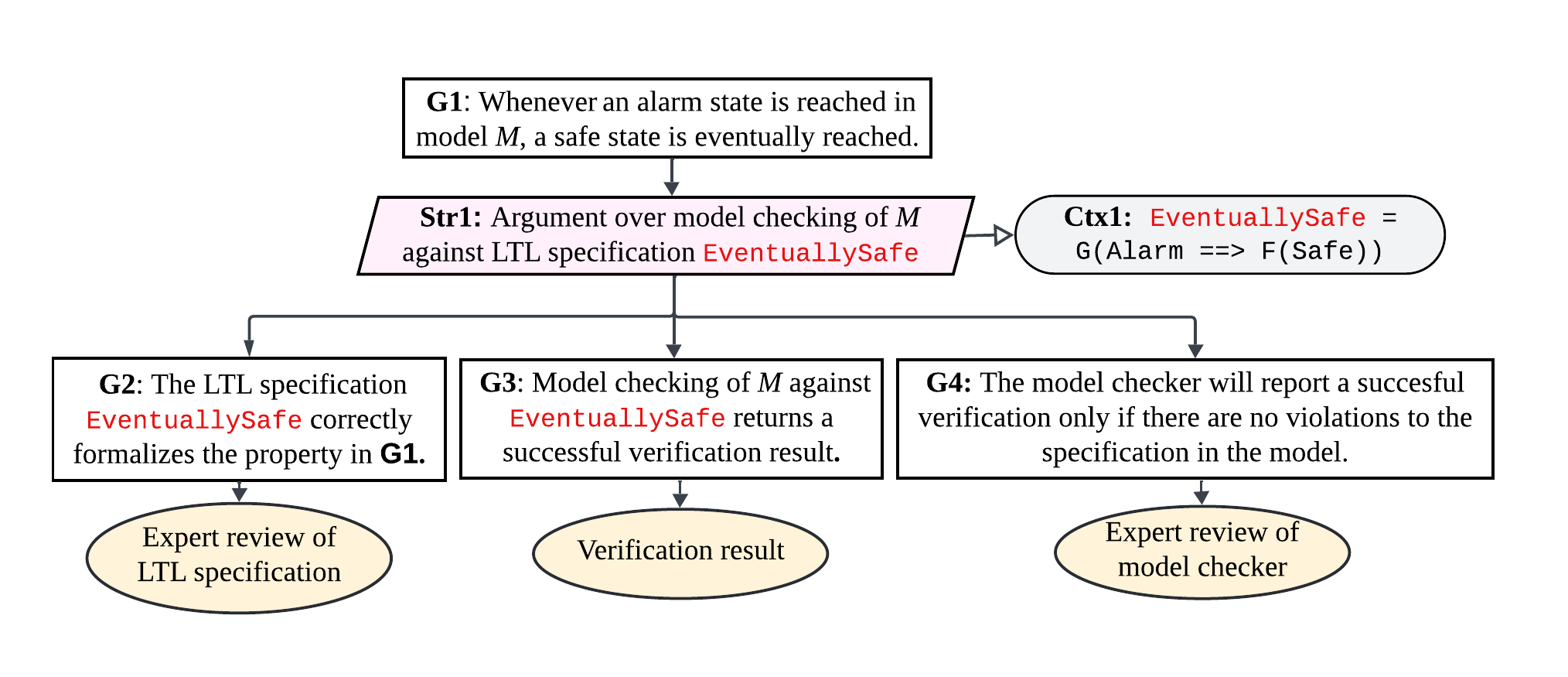}
% \vspace{-0.1in}
    \caption{\bf An AC fragment (in GSN) for an LTS model $M$. The parent goal $g_1$ is decomposed via a model checking strategy $st$ into three subgoals, each supported by evidence. A context node (gray) is used to show the formal specification used for model checking.}
    \label{fig:gsn-example}
    \vspace{-0.1in}
    \Description[]{}
\end{figure}

An \emph{assurance case} (AC) is a document which systematically demonstrates a system's suitability for operation by appealing to a set of evidence artifacts. 
{In this work, we focus on assurance cases represented using Goal Structuring Notation (GSN)~\cite{kelly1999arguing}. 
{In GSN, an AC is modeled as a rooted tree, whose nodes can contain \emph{goals} (i.e., claims to be supported), \emph{evidence} artifacts, assumptions, contextual information (e.g., references to  documents or artifacts), and \emph{strategies}. Strategies are used to decompose goals into a finite set of subgoals, with the intended interpretation being a logical refinement: if each of the subgoals hold, then the parent goal should hold.

\begin{example}\label{ex:gsn}
    \nc{An AC fragment for a labelled transition system (LTS) model $M$ is shown in Fig.~\ref{fig:gsn-example}. The root goal ${\ntn{\textbf{G1}}}$ asserts that whenever an alarm is reached in the model, a safe state is eventually reached. This goal can be supported through model checking~\cite{baier2008principles} against the LTL specification ${\tt G(Alarm \implies F({\tt Safe}))}$, i.e., ``whenever  (${\tt G}$) an ${\tt Alarm}$ state is reached, a ${\tt Safe}$ state is eventually (${\tt F}$) reached.'' The corresponding argument over model checking is given by strategy ${\ntn{\textbf{Str1}}}$. This strategy decomposes ${\ntn{\textbf{G1}}}$ into three subgoals: ${\ntn{\textbf{G2}}}$, asserting that the specification correctly formalizes the property described in the parent goal; ${\ntn{\textbf{G3}}}$, {\ntn{asserting}} that model checking did not reveal any violations; and ${\ntn{\textbf{G4}}}$,  {\ntn{asserting}} that a successful verification result can be used to conclude absence of violations to the specification. Each of these subgoals can then be supported by evidence: in the case of \typofix{$g_2$}{{\ntn{\textbf{G3}}}}, the verification result itself provides the evidence, with the remaining goals being supported by expert opinion. {\ntn{The  specification used for verification is shown in the \emph{context} node \textbf{Ctx1.}}}}
\end{example}

\subsubsection{Assurance Case Templates}
\label{sec:background:templates}
One of the limitations of informal assurance case development is that the correctness of ACs needs to be validated manually, which is difficult, expensive and error-prone. If we instead impose formalization as part of AC development, we can use formal methods to verify the correctness of ACs. This has been an active research area in assurance engineering, with various formal methods being leveraged to improve the rigor of ACs, including theorem provers~\cite{denney2012advocate}, model checking~\cite{sljivo2018tool} and proof assistants~\cite{foster2021integration,viger2023foremost}. {While there are variations in the research literature, many proposals involve the formalization of AC fragments as \emph{templates} whose correctness can be verified. An AC template is effectively a procedure for producing a (partial) AC, often requiring some kind of \emph{instantiation} data. }

\begin{example} \label{ex:domainDecomp1}
    {Consider  the \emph{domain decomposition} template given by Viger\typofix{et. al}{et al.}~\cite{viger2020just} to decompose an {invariance} claim, i.e., a goal of the form $\forall x \in S,\: P(x)$ where $S$ is some subset of universe ${U}$. The template is instantiated for $S$ by giving a finite family $\mathcal{F} = \{X_1,...,X_n\}$ of subsets of ${U}$, creating $n$ subgoals $\{g_1,...,g_n\}$ such that goal $g_i$ asserts $\forall x \in X_i,\: P(x)$. An instantiation of this template yields a sound argument if family $\mathcal{F}$ is \emph{complete} with respect to $S$, i.e., $S \subseteq \bigcup_i X_i $.}
\end{example}

\subsection{Software Product Lines}
\label{sec:productLines}
%\vspace{-8mm}

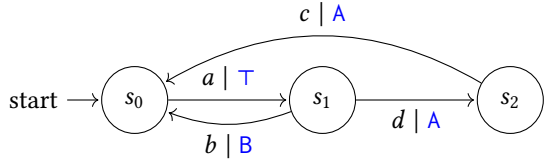
\begin{figure}[t]
\vspace{-0.2in}
    \centering
    \[\begin{tikzpicture}[shorten >=1pt,node distance=2.5cm,on grid,auto] 
   \node[state,initial] (q_0)   {$s_0$}; 
   \node[state] (q_1) [right=of q_0] {$s_1$}; 
   \node[state] (q_2) [right=of q_1] {$s_2$}; 
    \path[->] 
    (q_0) edge  node {$a \mid \textcolor{blue}{{\top}}$} (q_1)
    (q_2) edge [bend right=30, above] node {$c \mid \textcolor{blue}{{\tt A}}$} (q_0)
    (q_1) edge [below] node  {$d \mid \textcolor{blue}{{\tt A}}$}(q_2)
          edge [bend left=20, below] node {$b \mid \textcolor{blue}{{\tt B}}$} (q_0);
\end{tikzpicture}\]
\vspace{-9mm}
    \caption{\bf A Featured Transition System (FTS) over features ${F} = \{{\tt A}, {\tt B}\}$ and feature model $\Phi = {\tt A} ~{\tt xor}~ {\tt B}$ (adapted from ~\cite{beek2019static}).}
    \label{fig:FTS}
    % \vspace{-0.25in}
    \Description[]{}
\end{figure}

{{A \emph{software product line} (SPL) is a family of software artifacts ({products}) with distinct (but often overlapping) structure and\typofix{behaviours}{behaviour}~\cite{apel2016feature}. The variability of an SPL is defined in terms of a set $F$ of \emph{features}, each of which can be either present or absent in a given product. {A product is obtained from the SPL  by choosing a \emph{configuration} $\conf \subseteq {F}$ of features.} 
In this work, we focus on \emph{annotative} product lines, in which elements of software products -- lines of code, model elements, etc. -- are \emph{annotated} to indicate the products for which they are used. 

Annotative product line engineering makes extensive use of \emph{feature expressions}, i.e., propositional expressions whose atomic propositions are features. Given a set of features $F$, we let ${\tt Prop}(F)$ denote the set of feature expressions over $F$. When used to annotate a product element, a feature expression is called a \emph{presence condition}. There is a natural entailment relation between configurations and feature expressions; given $\conf \subseteq F$ and $\phi \in {\tt Prop}(F)$, we say that $\conf$ satisfies $\phi$ (denoted 
by $\conf \vDash \phi$) iff $\phi$ evaluates to $\top$ when all features in $\conf$ are replaced by $\top$, and all other features replaced by $\bot$. Given a feature expression \typofix{$\conf$}{$\phi$} over $F$, we let $\llbracket \phi \rrbracket$ denote the set of configurations satisfying $\phi$. In a given product line, not all possible combinations of features may be desired. A \emph{feature model} $\Phi$ defines the set of configurations which are considered ``valid'' in a given product line. The feature model can be\typofix{defied}{defined} \weakNew{in any number of manners:} graphically, i.e., via a feature diagram~\cite{schobbens2006feature}\weakNew{; textually, i.e., via feature modeling DSLs, such as the Universal Variability Language~\cite{benavides2025uvl};} or as a feature expression\weakChange{;}{.} In this work, we \weakChange{adopt the latter convention}{define feature models as feature expressions}.

{{Any modeling language for software products can be extended to represent product lines. For example, a product line of LTSs can be represented as a Featured Transition System (FTS)~\cite{classen2012featured}.}}

\begin{example}
    {Fig.~\ref{fig:FTS} illustrates an FTS  (adapted from \cite{beek2019static}) over the feature set $\{{\tt A}, {\tt B}\}$ with feature model $\Phi = A~{\tt xor}~B$. The presence conditions of transitions are \weakChange{written in blue font next to}{indicated to the right of} the transition labels. There are two valid configurations of this SPL, as $\llbracket\Phi\rrbracket = \{\{{\tt A}\},\{{\tt B}\}\}$.  In Fig.~\ref{fig:FTS}, the transition labeled by $a$ is present under both configurations, since it is annotated by $\top$, while the transition labeled by $b$ is present only in configuration $\{{\tt B}\}$.} 
\end{example}

While there are a multitude of ways in which a language can be extended to model product lines, we adopt the following as a default representation of an annotative product line model (adapted from \textrm{~\cite{shahin2020automatic}}). Given any modeling language $\mathcal{M}$ and model $M \in \mathcal{M}$, we let $\ntn{Elem}(M)$ denote the set of elements of $M$.

\begin{definition}[Product Line Model]\label{def:plModel}
    Let $\mathcal{M}$ be a modeling language. \weakNew{A} \emph{product line model} over $\mathcal{M}$ \weakDelete{-- also called a product line of $\mathcal{M}$-models --} is a tuple $\pl{M} = \langle F, \Phi, M, \ell \rangle$, where  $F$ is a set of features,  $\Phi$ is a feature model, \weakChange{$M$ is an $\mathcal{M}$-model}{$M \in \modlang$}, and $\ell : \ntn{Elem}(M) \to {\tt Prop}(F)$ is a \emph{labeling function} on elements of $M$, which annotates each element $e$ with a presence condition $\ell(e)$.
    We let $\Var{\mathcal{M}}$ denote the set of all possible \weakChange{product lines of $\mathcal{M}$-models}{product line models over $\modlang$}.
\end{definition} 

A fundamental operation on product line models is \emph{derivation}, i.e., the extraction of a product given a specific configuration. \newText[derivation:explain]{Given a product line model $\pl{M} = \langle F, \Phi, M, \ell\rangle$, we derive the product model induced by $\conf \subseteq F$ by removing from $M$ all elements $e$ such that $\conf \nvDash \ell(e)$. We refer to the resulting product model as $\pl{M}|_\conf$. Depending on the modeling language being used, there may be some elements of the model which ``depend'' on the element $e$ being removed. For example, removing a transition from a state machine model may make some states unreachable, in which case they should also be removed during derivation. As another example, if a certain class $C$ is removed from a class diagram, any subclasses of $C$ may also need to be removed. Generally speaking, defining a product line model consistent with Def.~\ref{def:plModel} requires defining a particular derivation operator for the chosen modeling language.}

    \delete[derivation:def]{\emph{Definition.} Given a product line model $\pl{M} = (F,\Phi,M,\ell)$, and a configuration $\conf \subseteq F$, we can obtain the \emph{product model} under $\conf$ by removing from $M$ all elements $e$ such that $\ell(e)$ is unsatisfiable. All elements which become ``unreachable'' are also removed.}

\begin{figure}[t]
    \centering
    \[\begin{tikzpicture}[shorten >=1pt,node distance=2.5cm,on grid,auto] 
   \node[state,initial] (q_0)   {$s_0$}; 
   \node[state] (q_1) [right=of q_0] {$s_1$}; 
    \path[->] 
    (q_0) edge  node {$a $} (q_1)
    (q_1) edge [bend left=20, below] node {$b$} (q_0);
\end{tikzpicture}\]
\vspace{-9mm}
    \caption{\bf The LTS product derived from the FTS in Fig.~\ref{fig:FTS} under configuration $\{\token{B}\}$.}
    \label{fig:LTS}
    % \vspace{-0.25in}
    \Description[]{}
\end{figure}
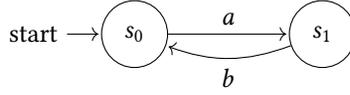

\begin{example}
    \nc{\weakNew{Derivation of an LTS from an FTS under a configuration $\conf$ is achieved by removing all edges whose presence conditions are not satisfied by $\conf$, and subsequently removing any states and transitions which become unreachable.} Fig.~\ref{fig:LTS} shows the LTS product derived from the FTS in Fig.~\ref{fig:FTS} under configuration $\conf = \{\token{B}\}$, i.e., the product in which feature $\token{B}$ is present and feature $\token{A}$ is absent. The transition labeled $a$ is kept since it is annotated by $\top$, which is satisfied by all configurations. The transition labeled $b$ is kept since it is guarded by \weakChange{feature expression}{presence condition} $\token{B}$, and $\{\token{B}\} \vDash \token{B}$. The transitions labeled by $c$ and $d$ are removed since they are both guarded by \weakChange{feature expression}{presence condition} $\token{A}$, and $\{\token{B}\} \nvDash \token{A}$. The state $s_2$ then becomes unreachable, and is removed.}
\end{example}

\change[explicit:pl:aside]{While Def.~\ref{def:plModel} is sufficiently generic for most applications, it is not the only possible definition of an annotative product line. Indeed, in this paper we consider some instances of ``product lines'' that use a more specialized representation -- for instance, product lines of assurance cases. In general, we consider a ``product line extension'' of any language $X$ to be a language $\Var{X}$ which is equipped with a derivation operator $d : \Var{X} \times 2^F \to X$, i.e., a way of obtaining a product in $X$ given a configuration $\conf \subseteq F$.
We use the same generic notation to denote derivation across all product line languages: given $\pl{x} \in \Var{X}$ and $\conf \subseteq F$, we let $\pl{x}|_\conf$ denote the product derived from $x$ under $\conf$.}{While Def.~\ref{def:plModel} provides one way of encoding an annotative product line, there are alternatives. For instance, given some set $X$, one can encode a ``product line of elements of $X$'' explicitly as a set of pairs $\{\langle x_1, \phi_1\rangle \ellipses \langle x_1, \phi_n\rangle \}$, where each $x_i$ is the ``product'' corresponding to configurations satisfying feature expression $\phi_i$~\cite{walkingshaw2014variational}. For the derivation of a product to be well-defined, the feature expressions must be pairwise disjoint (i.e., for all distinct $i$ and $j$, $\llbracket \phi_i \land \phi_j \rrbracket = \emptyset$), and every valid configuration must satisfy some $\phi_i$, (i.e., $\llbracket \bigvee_i \phi_i\rrbracket = \llbracket\Phi\rrbracket$). If both these properties hold, we can derive a product under any valid configuration $\conf$ by choosing the unique $x_i$ such that $\conf \vDash \phi_i$. For example, given features $F = \{\token{A}, \token{B}\}$ and feature model $\Phi = \top$, we could define an instance of a product line of natural numbers as $\pl{n} = \{(1, \token{A} \wedge \token{B}), (2,\token{A} \wedge \neg B), (17, \neg \token{A}) \}$, from which we could derive the products $\pl{n}|_{\{\token{A},\token{B}\}} = 1$ and ${\pl{n}}|_\emptyset = 17$. As another example, we could define a product line of Boolean values over the same features and feature model as $\pl{b} = \{(\top, \token{A}), (\bot, \neg \token{A})\}$, such that $\pl{b}|_{\{\token{A},\token{B}\}} = \top$ and $\pl{b}|_{\{\token{B}\}} = \bot$. 

While this ``explicit'' encoding can be used to represent product lines of  {arbitrary} data, there are many cases in which the encoding described in Def.~\ref{def:plModel} is more efficient. For instance, consider a product line of two LTSs $M_1$ and $M_2$, each of which contains thousands of states, and which differ  only in that $M_1$ has a state $s$ which is not present in $M_2$. An ``explicit'' encoding of this product line would require storing full copies of both $M_1$ and $M_2$, even though they are nearly identical models. Representing the product line using the encoding described in Def.~\ref{def:plModel} (e.g., as an FTS) is more efficient, as we only need to store one copy of all model elements appearing in either $M_1$ or $M_2$. Nevertheless, the explicit encoding of annotative product lines is useful as it allows us, in principle, to reason about product lines of arbitrary data. In what follows, we will extend the notation $\Var{X}$ to describe the set of product lines of elements of $X$, regardless of the encoding used. We also use the same notation $\pl{x}|_\conf$ to denote the derivation of a product from product line $\pl{x} \in \Var{X}$ under configuration $\conf$.}

\newText[pl:prod:clarification]{Finally, given any languages of product lines $\Var{X}$ and $\Var{Y}$, we can define the language $\Var{(X \times Y)}$ as $\Var{X} \times \Var{Y}$. That is, we can form a product line in which each product is a pair $\langle x,y \rangle \in X \times Y$ by taking a pair of product lines $\langle \pl{x}, \pl{y}\rangle \in \Var{X} \times \Var{Y}$. The derivation of the product pair under configuration $\conf$ is naturally $\langle \pl{x},\pl{y}\rangle|_\conf = \langle \pl{x}|_\conf, \pl{y}|_\conf\rangle$.}

    \delete[nat:example]{An example of a ``product line language'' which is \emph{not} an instance of Def.~\ref{def:plModel} is the language $\Var{\mathbb{N}}$ of variational natural numbers~\cite{shahin2020automatic}. Given a feature model $\Phi$, a value $\pl{v} \in \Var{\mathbb{N}}$ represents a ``product line'' of natural numbers of the form 
    $\pl{v} = \{\langle x_1,\phi_1 \rangle,...,\langle x_n, \phi_n\rangle  \}$,
    where each $x_i \in \mathbb{N}$ and each $\phi_i$ is a feature expression such that for any $\conf \in \llbracket \Phi \rrbracket$, there is exactly one $\phi_i$ such that $\conf \vDash \phi_i$. The existence and uniqueness of $\phi_i$ induces a derivation function, i.e., for each $\conf$, we choose $\pl{v}|_\conf = x_i$ where $\conf \vDash \phi_i$.}

Annotative product lines provide a concise representation of (possibly very many) individual models. The downside of product line modeling is that existing (product-based) analyses cannot be applied directly to product lines. \nc{If we want to analyze a product line model with a product-based analysis, we can only do so in a brute-force fashion, deriving the set of valid products and analyzing each of them independently. This is infeasible, as a product line with $n$ features will encompass up to $2^n$ products. It is also often \weakChange{redundant}{wasteful}, as different products may overlap significantly with one another, so performing each analysis ``from scratch'' is typically not necessary. An alternative strategy is to \emph{lift} the analysis, i.e., redefine it so that it can be applied to the entire product line, and produce equivalent results to brute-forced analysis.}

\begin{definition}[Lifting~{\rm \cite{Murphy2023ReusingYF}}]\label{def:lift} \rm
    Let $f : X \to Y$ be a function. Fixing a feature model $\Phi$, a function $f^\uparrow : \Var{X} \to \Var{Y}$ is a \emph{lift} of $f$ iff for all $\pl{x} \in \Var{X}$ and all $\conf \in \llbracket \Phi \rrbracket$, we have 
    $f^\uparrow(\pl{x})|_\conf = f(\pl{x}|_\conf)$.
\end{definition}
\nc{The above definition can be adapted for functions which take additional inputs. For instance, if $f$ is an analysis which also takes some specification $\psi$ (e.g.,\typofix{and}{an} LTL formula), then the correctness \weakChange{criterion}{specification} becomes $f^\uparrow(\pl{x},\psi)|_\conf = f(\pl{x}|_\conf,\psi)$. Many forms of software analyses have been lifted, such as dataflow analysis~\cite{bodden2013spllift}, model checking~\cite{classen2012featured}, and testing~\cite{kastner2012toward}\typofix{ }{}; see ~\cite{MURPHY2025112280} for a recent survey.}

% \vspace{-0.1in}
\section{Formalizing Product-Based AC Development}
\label{sec:product_based_acs}

In this section, we formalize a framework for product-based AC development, which will form the foundation of \emph{lifted} AC development in Sec.~\ref{sec:lifting}. More precisely, we define a formal AC language based on\change[gsn:ref]{GSN}{Goal Structuring Notation (GSN), introduced in Sec.~\ref{sec:background:model_based_acs}}, and formalize the correctness of ACs in terms of evidence and argumentation (Sec.~\ref{sec:productACs}). We then formalize assurance case \emph{templates} and their correctness (Sec.~\ref{sec:acTemplates}). We conclude by introducing a special class of AC templates, called \emph{analytic templates}, which incorporate the semantics and results of software analysis directly into an AC (Sec.~\ref{sec:analyticTemplates}). Unless noted otherwise, all definitions and theorems in this section are adapted from~\cite{ifm2024}. Proofs of selected theorems not given here are provided in Appendix~\ref{app:proofs:sec3}.
 
\subsection{Formalizing an Assurance Case Language}
\label{sec:productACs}
An assurance case serves to make and justify claims about a system and its suitability for operation. To enable assurance case formalization, we consider \emph{model-based} assurance cases, i.e., those in which the primary subjects of assurance case claims are system models. In this work, we adopt the model-driven engineering philosophy in which ``everything is a model''{~\cite{bezivin2005unification}. Hence, the} word ``model'' can refer to traditional models, such as state machines and activity diagrams, as well as to objects such as programs, strings, and sets.
\weakDelete{We assume the existence of \emph{metamodels} $\modlang, \mathcal{N}...$, which define languages (sets) of models (class diagrams, state machines, etc.). Given a metamodel $\modlang$, we refer to any $M \in \modlang$ as an \emph{$\modlang$-model}, or simply as a \emph{model} when the intended metamodel is clear.}

In general, an assurance case goal can represent any proposition. However, in formalizing assurance cases -- and especially in formalizing assurance case \emph{templates} -- it is \weakDelete{also} useful to model goals \emph{predicatively}, i.e., as the application of some predicate $P$ to some subject (e.g., a system model). This leads us to the following formal notion of assurance case goals.

\begin{definition}[Goal]\label{def:goal}
    An assurance  {{\ntn{case}}} goal $g$ is given as either:
    \begin{enumerate}[label=\roman*)]
        \item a proposition $p$ in first-order logic (called a \emph{propositional goal}), or
        \item a pair $\langle M,P \rangle $ where $M\in \modlang$ is a model and $P$ is a predicate over $\modlang$ (called a \emph{predicative goal}).
    \end{enumerate}
    \newText[goaldef]{We let $\Goal$ denote the set of all such goals.}
\end{definition}

Obviously, a goal $g$ can be interpreted directly as the proposition it \weakChange{represents}{asserts}: in the case of a propositional goal, it is already defined by some proposition \weakNew{$p$}; for a predicative goal $\langle M,P \rangle$, we interpret the goal as \weakNew{asserting} $P(M)$. With this in mind, we often treat goals like propositions for the sake of presentation, e.g., we write $g \implies g^\prime$ with the\change[nat:imp]{natural interpretation}{interpretation that the proposition asserted by $g$ logically implies the proposition asserted by $g^\prime$}. It would be straightforward to modify Def.~\ref{def:goal} to allow goals which assert $n$-ary relations between multiple models. Instead, to simplify presentation, we adopt the convention that a tuple of models $\langle M_1,\dots,M_n\rangle$ is itself considered a model.

Our formalization of assurance cases includes two other entities inspired by GSN: \emph{evidence} and \emph{strategies}. Evidence can take many forms: testing results, verification results, formal proofs, system documentation, expert opinion, and so on. To simplify our formalization, we assume a set $\mathcal{E}$ of evidence artifacts (broadly conceived), making no assumptions about the internal structure of any $e \in \mathcal{E}$. However, we assume that \weakNew{determining} \emph{evidence adequacy}, i.e., \weakDelete{determining} whether a piece of evidence actually \emph{is} evidence of a given goal, is decidable. More precisely, we can always (manually or automatically) decide whether evidence $e$ is adequate to support proposition $P$, in which case we write \change[prfnotation]{$e : P$}{$e \vdash P$}. \weakChange{Our formalization takes a deeper view of argumentation, organized into \emph{strategies}. As}{We next turn to \emph{strategies}, which, as} shown in Fig.~\ref{fig:gsn-example},\change[strategyclar]{strategies relate a parent goal (the conclusion of the argument) to its child goals (the premises of the argument)}{clarify the logical relationship between a parent goal (the conclusion of the argument) and its subgoals (the premises of the argument)}. Unlike goals or evidence, the strategy node itself adds no semantic content to the assurance case; it instead provides \emph{explanatory} content to the reader, so that they can understand what argument is being made. Accordingly, our formalization models strategies simply as opaque labels. We let $\mathsf{Str}$ denote the set of all strategies.

We can now \weakDelete{proceed to} define a formal language of GSN-like assurance cases.

\begin{definition}[Assurance Case]\label{def:AC} The language $\AC$ of \emph{assurance cases} is generated by the following grammar:
\begin{align*}
    \AC := &\; {\tt Und}(g) \tag{$g \in \Goal$}\\
    \mid &\;{\tt Evd}(g,e) \tag{$g \in \Goal$, $e \in \Evd$} \\ 
    \mid &\;{\tt Decomp}(g,st, \{A_1\ellipses A_n\}) \tag{ $g \in \Goal$, $st \in \mathsf{Str}$, $A_i \in \AC$}
\end{align*}
\end{definition}
The first constructor defines an undeveloped goal $g$, i.e., one which has yet to be supported by evidence or argumentation; the second defines a goal $g$ being supported by evidence $e$; and the third constructor defines a goal $g$ being decomposed into sub-ACs $\{A_1\ellipses A_n\}$ via decomposition strategy $st$, with $n \geq 1$. \weakChange{The above}{This} definition is purely syntactic, making no assumptions on the relations between goals, evidence, or strategies. Before we define these semantic qualities, we introduce the \emph{root} function ${\tt Rt} : \AC \to \Goal$, which returns the root goal of a given AC, i.e.,
\begin{align*}
    {\tt Rt}({\tt Und}(g)) = g \qquad 
    {\tt Rt}({\tt Evd}(g,e)) = g \qquad 
    {\tt Rt}({\tt Decomp}(g, st, \{A_1\ellipses A_n\})) = g
\end{align*}

The key semantic property of assurance case {argumentation} is \emph{goal refinement}~\cite{viger2023foremost}, which determines whether a given strategy represents a sound inference. 
\begin{definition}[Goal Refinement]\label{def:goalRefinement}
    We say that a set of assurance cases $\mathcal{A} = \{A_1,...A_n\}$ {\newText[guard]{($n \geq 1$)}} \emph{refines} goal $g$, denoted $\mathcal{A} \prec g$, iff \change[forall:rw]{$(\forall i \leq n, g_i) \implies g$}{$\left(\bigwedge_{i=1}^ng_i\right) \implies g$},
    where $g_i = {\tt Rt}(A_i)$.
\end{definition}
In other words, a strategy induces a sound goal refinement if, whenever all of the \weakNew{root} goals in the child ACs hold, the parent goal holds. \nc{In what follows, we often ignore the distinction between a goal $g$ and the AC fragment ${\tt Und}(g)$. For instance, we may indicate a goal refinement $\{g_1\ellipses g_n\} \prec g$, with the understanding that each $g_i$ is interpreted as ${\tt Und}(g_i)$ to maintain consistency with Def.~\ref{def:goalRefinement}}.

\begin{example}\label{ex:precRefinement}
   \begin{figure}
         \centering \includegraphics[width=0.7\linewidth]{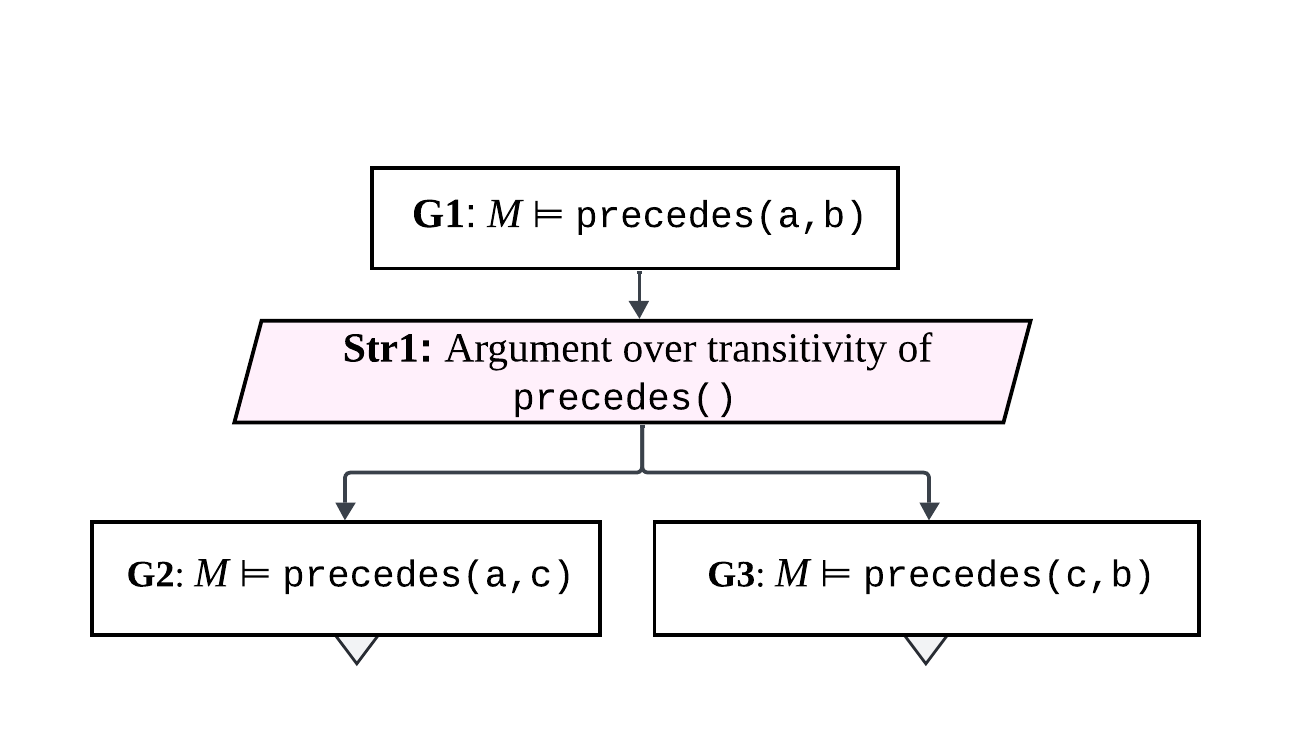}
         \caption{\bf AC fragment formalized in Example~\ref{ex:precRefinement}.}
         \label{fig:prec_example_fig}
         \Description[]{}
     \end{figure}
    We can formalize properties of an LTS $M$ through the satisfaction of LTL formulae, which we denote by $M \vDash \psi$, with $\psi$ being an LTL formula. Consider the ``$\token{precedes}$'' relation between events in an LTS, such that $\token{precedes}(a,b)$ means that whenever event $b$ occurs, event $a$ must have occurred first (in LTL, $\neg b \:W\:a$, with $W$ denoting weak-until). Fix the notation $P_{a,b}$ for predicates over LTSs as $P_{a,b}(M) = \typofix{}{}$``$M \vDash \token{precedes}(a,b)$'' for any model $M$ and events $a$ and $b$.

    Fix some $M \in {\tt LTS}$, for which we want to produce assurance that $M \vDash \token{precedes}(a,b)$. We define a root goal $g = \langle M, P_{a,b}\rangle$. Given some third event $c$, can then define two subgoals $g_{1} = \langle M,P_{a,c} \rangle $ and $g_{2} = \langle M,P_{c,b}\rangle$. Leaving these goals as undeveloped, we form the decomposition $A = {\tt Decomp}(g, st, \{g_1, g_2\})$, where strategy $st$ indicates an ``argument by transitivity of \token{precedes}''. This assurance case is visualized in Fig.~\ref{fig:prec_example_fig}, with undeveloped goals being indicated by a triangle token. We can also verify that the refinement $\{g_{1},  g_{2}\} \prec g$ holds, since $M \vDash \token{precedes}(a,c) \land M \vDash \token{precedes}(c,b) \implies M \vDash \token{precedes}(a,b)$.

 \end{example}

Goal refinement and evidence adequacy collectively define the correctness of an assurance case. Ideally, one wants to be able to infer with confidence that the root goal of an AC is true; one can infer this if (a) every evidence artifact in the AC is adequate \weakDelete{evidence} for the goal to which it is assigned, and (b) if every strategy in the AC provides a sound goal refinement. We say that an AC for which (a) and (b) hold is \emph{supported}.

\begin{definition}[Supported Assurance Case]\label{def:supp}
For any assurance case $A \in \AC$, the \emph{supported} predicate $\Supp(A)$ can be inferred via the following rules:
\[  \infer[{\tt [Supp\text{-}1]}]{\Supp({\tt Evd}(g,e))}{
    &g = p
    & \ntn{e \vdash p}
    }\qquad
     \infer[{\tt [Supp\text{-}2]}]{\Supp({\tt Evd}(g,e))}{
    &g = \langle M, P\rangle
    & \ntn{e \vdash P(M)}
    }
\]
% \vspace{0.5mm}
\[   \infer[{\tt [Supp\text{-}3]}]{\Supp({\tt Decomp}(g,st,\{A_1\ellipses A_n\}))}{
    & \forall i \leq n, \Supp(A_i)
    & \{A_1\ellipses A_n\} \prec g
    }
\]
\end{definition}
One can verify by induction that these rules correctly model the intended interpretation of an AC:
\begin{lemma}\label{def:suppCorrect}
Let $A$ be an assurance case \weakDelete{for system $S$} such that\; ${\tt Rt}(A) = g$. If $A$ is supported, then we can infer $g$.
\end{lemma}

\subsection{Formal Assurance Case Templates}
\label{sec:acTemplates}

\weakChange{As per}{Per} Def.~\ref{def:supp}, verifying goal refinements is a key component of ensuring the correctness of an AC. As mentioned in  Sec.~\ref{sec:background:templates}, \emph{templates} can be used to systematize the formalization and verification of \typofix{decompositions}{decomposition} strategies\weakDelete{in ACs}. \newText[templatePref]{In general, a template $T$ is designed to support goals of the form $\langle M,P\rangle$ for any model $M \in \modlang$, given some fixed predicate $P$ over $\modlang$. \emph{Instantiating} a template generally requires providing data from some domain $D$. This is formalized in the following definition.}

\begin{definition}[AC Template]\label{def:template}
    \change[templateDef]{An}{Given a modeling language $\modlang$ and a domain $D$, an} AC \emph{template} {\ntn{$\template{\modlang}{D}$}} is a tuple $\langle P, I, C \rangle$, where 
    \begin{itemize}
        \item $P : \modlang \to \{\top,\bot\}$ is called the \emph{parent predicate},
        \item $I : D \to {\ntn{\mathbb{P}({\Goal})}}$ is the \emph{instantiation function}, \weakChange{where $D$ is some auxiliary domain of data used in the instantiation,}{such that {for all $x \in D, I(x)$ is finite}}, and
        \item $C : D \times \modlang \to \{\top, \bot\}$ is the \emph{correctness criterion}, which is assumed to be decidable.
    \end{itemize}
    \nc{Intuitively, $P$ \weakChange{formalizes the parent goal which will be applied to some model $M$;}{defines the class of goals which $\template{\modlang}{D}$ is designed to decompose,} $I$ defines how the (\weakNew{finitely many}) subgoals of the \weakChange{strategy}{decomposition} are constructed, and $C$ defines the conditions under which the resulting strategy produces a sound argument}. The template is said to be \emph{valid} if for every $x \in D$ and $M \in \modlang$,  $C(x,M)$ implies $I(x,M) \prec$ \weakChange{$P(M)$}{$\langle M, P\rangle$}. 
\end{definition}

\begin{example} \label{ex:domainDecomp2}
    Recall the domain decomposition template described in Example~\ref{ex:domainDecomp1}. Fix a universe $U$ and a predicate $P$ over $U$. We can formalize the domain decomposition template with respect to $P$ as $\template{\mathbb{P}(U)}{\mathbb{F}(U)} = \langle $\weakChange{$\mathcal{P}$}{$\mathsf{Forall}_P$}$, I, C\rangle$, where 
    \begin{itemize}
        \item $\ntn{\mathsf{Forall}_P}$ is a predicate over $\ntn{\mathbb{P}(U)}$ such that\weakNew{, for any subset $S \subseteq U$, }  $\mathsf{Forall}_P(S) = $ ``$\forall x \in S, P(x)$'',
        \item the instantiation function $I :{\ntn{\mathbb{F}(U)} \to\ntn{\mathbb{P}(\Goal)}}$ {is defined by}  $I(\{X_1\ellipses X_n\}) = \{g_1\ellipses g_n\}$, 
        with each $g_i$ = $\langle X_i, \ntn{\mathsf{Forall}_P} \rangle$, and
        
        \item the correctness criterion is $C(\{X_1\ellipses X_n\}, S) = $ ``$S \subseteq \bigcup_iX_i$''. 
    \end{itemize}
        \weakNew{Note that $I$ is defined \emph{only} for finite families of subsets of $U$.}
\end{example}
\begin{proposition}\label{prop:domaindecompPrf}
    The domain decomposition template defined in Example~\ref{ex:domainDecomp2} is valid.
\end{proposition}
\begin{proof}
    \nc{Fix a universe $U$, a predicate $P$ over $U$, a set $S \subseteq U$, and a family $\mathcal{F} = \{X_1\ellipses X_n\}$ of subsets of $U$. \weakNew{Suppose we want to decompose the goal $\langle S, \mathsf{Forall}_P\rangle$ using the domain decomposition template instantiated with $\mathcal{F}$.} \weakChange{The domain decomposition induced by \typofix{$F$}{$\mathcal{F}$}}{The resulting decomposition} is given as ${\tt Decomp}(g,st, \{g_1\ellipses g_n\})$, where each $g_i$ \weakChange{$\equiv$}{asserts} $\forall x \in X_i, P(x)$. Assume the correctness criterion holds, i.e., $S \subseteq \bigcup_i X_i$, and assume that each $g_i$ is true. Given an arbitrary $x \in S$, there must be at least one $X_k$ for which $x \in X_k$, by the correctness criterion. Fixing this $X_k$, we can specialize $g_k$ with $x$ to obtain $P(x)$. By the generalization of $x$, we have $\forall x \in S, P(x)$, which is the parent goal $g$. Thus, $\{g_1,...g_n\} \prec g$, and so the template is valid.}
\end{proof}

\weakNew{We can also define an alternative template which can be applied to parent goals of the form $\mathsf{Forall}_P(S)$ when the set $S$ happens to be finite.}

\begin{example}[Enumeration Template] \label{ex:enumDecomp}
    Fix a universe $U$ and predicate $P$ over $U$\weakNew{, and suppose (as with domain decomposition) we wish to support a goal of the form $\mathsf{Forall}_P(S)$ for some finite set $S \subseteq U$}. Define the \emph{enumeration template} $\ntn{\template{\mathbb{P}(U)}{\mathbb{P}(U)}} = \langle {\ntn{\mathsf{Forall}}_{P}}, I, C \rangle $, where \weakDelete{${\ntn{\mathsf{Forall}}_{P}}$ is the same as in the domain decomposition template (Example~\ref{ex:domainDecomp2}), and} \typofix{$I : {2^U} \to [\AC]^*$}{$I : \ntn{\mathbb{P}(U)} \to \mathbb{P}(\Goal)$} is defined (for finite sets) as 
    $I(\{x_1,...x_n\}) = \{\langle x_1, {P} \rangle, ... , \langle x_n, P \rangle\}$
    \weakChange{With the constraint that the set $S$ {used to form the parent goal of the template} is the same set provided to the instantiation function.}{and $C : \mathbb{P}(U) \times \mathbb{P}(U) \to \{\top,\bot\}$ is defined as $C(T,S) = ``S = T"$. That is, {correctness merely requires} that the set $S$ which is the subject of the {parent} goal being decomposed is the same set $T$ provided to $I$. Intuitively,} \weakDelete{In other words,} the argument \weakDelete{simply} demonstrates that $P$ holds for every element of $S$ \weakNew{simply} by showing that $P$ holds for each element of $S$ individually.
\end{example}
A simplified version of the proof of Prop.~\ref{prop:domaindecompPrf} gives:
\begin{proposition}
    The enumeration template is \weakChange{correct-by-construction.}{valid.}
\end{proposition}

\weakDelete{Not every formal template may require a correctness criterion $C$.  Some templates may always produce sound goal refinements, without needing to verify any relation between the model in the parent goal and the instantiation data. We say that such templates are \emph{correct-by-construction}. We can define {such} a template simply in terms of parent predicate $P$ and instantiation function $I$, as long as we prove that it is correct-by-construction. \weakDelete{We give two such examples.}}

\weakDelete{\emph{Example (Transitivity Template)} Recall the argument over the transitivity of $\token{precedes}$ formalized in Example~\ref{ex:precRefinement}. Having fixed events $a$ and $b$, this argument can be abstracted into a formal template $T_{a,b}$ over LTSs as follows. Naturally, the parent predicate is $P_{a,b}(M) = $ \typofix{}{ }``$M \vDash \token{precedes}(a,b)$''. The instantiation of the template requires specifying the third event $c$, and produces the subgoals $\{g_1,g_2\}$, where
    $g_1 = \langle M, P_{a,c} \rangle$ and $g_2 = \langle M,P_{c,b}\rangle$. This template is correct-by-construction, since the refinement $\{g_1,g_2\} \prec \langle M,P_{a,b} \rangle $ does not depend on the choice of the intermediate event $c$, nor any actual property of the model $M$; it is sound simply because  
    $\token{precedes}(a,c) \land \token{precedes}(c,b)$ $ \implies \token{precedes}(a,b)$
    is an LTL-tautology.
}

\subsection{Analytic Argumentation}
\label{sec:analyticTemplates}

Templates form an important component of rigorous AC development, since they enable verification of goal refinement. But another essential component is to actually \emph{analyze} the system under assurance -- to determine facts about it, and to obtain assurance evidence. Rather than viewing argumentation and analysis as disjoint activities, we propose \emph{analytic} argumentation in AC development, wherein the semantics of analyses used in producing assurance are articulated explicitly in the AC. This serves two purposes. First, by making explicit  \emph{how} and \emph{why} a given analysis is used, the AC itself becomes more rigorous. Second, when we eventually move to \emph{lifting} AC development to product lines, the capability for lifted analyses to identify variation points in the system can be leveraged to automatically embed variability into argumentation.

{To this end, we formalize a special class of argument templates, which we refer to as \emph{analytic templates}, each of which is defined with respect to a fixed analysis $f$. In this work, we consider analyses to be functions of the form\change[analysisSig]{$f : \modlang \times \mathcal{L} \to Y$}{$f : X \to Y$}, where\change[analysisExpl]{$\modlang$ is some modeling language and $\mathcal{L}$ is some specification language (e.g., a property to verify in the model)}{$X$ is the set of possible inputs to the analysis (its domain) and $Y$ is the set of possible outputs of the analysis (its codomain)}.
Intuitively, an analytic argument makes assertions about the \change[specClar]{analysis specification}{input to the analysis}, the output of the analysis, and \weakChange{$f$}{the analysis} itself.

\begin{definition}[Analytic Template]\label{def:analyticTemplate}
Let ${\ntn{f : X \to Y}}$ be some analysis \change[analDef]{ over models in $\modlang$ and specifications in $\mathcal{L}$}{with domain $X$ and codomain $Y$, and let $\modlang$ be some modeling language}. An \emph{analytic template} over $f$ is a template ${\ntn{\template{\modlang}{X}}} = \langle P, I,C\rangle$, where
\begin{itemize}
    \item \weakChange{$P : \modlang \to \{\top,\bot\}$ is the parent predicate, and} {\weakNew{the parent predicate $P : \modlang \to \{\top,\bot\}$ and correctness criterion $C : X \times \modlang \to \{\top,\bot\}$ have the same meaning as in Def.~\ref{def:template}, and}}
    \item \weakNew{the instantiation function $I : X \to \mathbb{P}(\Goal)$ is defined for any $x \in X$ as $I(x) = \{g_X, g_{Y}, g_f\}$, where}
    \[\weakNew{g_X = \langle x, P_X\rangle, \qquad g_{Y} = \langle f(x), P_Y\rangle,\qquad g_f = \forall x \in X, P_f(x,f(x))}\]
    \weakNew{for some predicates $P_X$ over $X$, $P_Y$ over $Y$, and $P_f$ over $X \times Y$.}    
   
   \weakDelete{$P_{spec} : \mathcal{L} \to \{\top,\bot\}$ is called the {\emph{specification}} predicate,}

    \weakDelete{$P_{out} : Y \to \{\top,\bot\}$ is called the \emph{output} predicate, and}

    \weakDelete{$P_f : \modlang \times \mathcal{L} \to \{\top,\bot\}$ is called the \emph{functional specification} of the analysis.}
\end{itemize}
\weakChange{Intuitively, $P$ formalizes the property of the model we want to demonstrate via $f$; $P_{spec}$ formalizes the intended semantics of the analysis specification; $P_{out}$ formalizes some assertion we want to make about the outcome of the analysis; and $P_f$ formalizes what the analysis actually does. Engineers need to determine appropriate choices for $P_{spec}$, $P_{out}${,} and $P_f$ based on the given $f$ and $P$.}{Generally speaking,  to support the parent goal $\langle M,P\rangle$ for some model $M \in \modlang$, $P_X(x)$ should assert that $x$ is an appropriate input to the analysis, $P_Y(f(x))$ should assert some useful fact about the analysis result, and $P_f$ should assert that the analysis itself is ``correct''. The precise definitions of $P_X$, $P_Y$, and $P_f$ depend naturally on the parent predicate $P$ and the semantics of the analysis $f$.}

\weakDelete{The \emph{instantiation} of an analytic template is as follows: given model $M$ to form the parent goal and a specification  $\psi \in \mathcal{L}$, the instantiation function $I_f(M,\psi)$ produces subgoals $\{g_{spec},g_{out},g_f\}$, where}

\weakDelete{$g_{spec} = \langle \psi, P_{spec}\rangle \quad g_{out} = \langle f(M,\psi), P_{out}\rangle \quad g_f = \forall ~ M \in \modlang, \forall \psi \in \mathcal{L}, P_f(M,\psi)$}

\weakDelete{We refer to $g_{spec}$, $g_{out}$, and $g_{f}$ as the \emph{specification}, \emph{output}, and \emph{functional} subgoals, respectively.}
\end{definition}

\weakDelete{Note that the functional subgoal $g_f$ is propositional, describing the functional correctness of the analysis $f$ over \emph{any} choice of $M$ and $\psi$. This means that, once this goal is supported, this assurance about $f$ can potentially be reused over multiple analytic templates relying on $f$.}

\begin{example}\label{ex:mc_template} 
    \nc{\weakNew{Consider a hypothetical LTL model checker $V$, which can be modeled as an analysis $V : {\tt LTS} \times {\tt LTL}  \to \token{Result}$, with any $r \in \token{Result}$ being either a successful verification result $(ok)$ or a counterexample of the property being verified. To simplify presentation, let $X$ denote ${\token{LTS} \times \token{LTL}}$ and let $Y$ denote ${\token{Result}}$.} Consider the GSN fragment shown in Fig.~\ref{fig:gsn-example}. \weakDelete{The argument over model checking represented by strategy $st$ can be abstracted as an analytic template over the model checker $V : {\tt LTS} \times {\tt LTL} \to \{\top\} \cup Cex$, where $\top$ indicates a successful verification, and $Cex$ denotes the set of potential counterexamples the model checker could return.} The proposition in parent goal \textbf{G1} can be abstracted formally\footnote{In fact, it is not strictly necessary for the predicate to be abstracted formally -- one could simply define $P(M)$ in terms of the natural language expression given in ${\ntn{\textbf{G1}}}$ of Fig.~\ref{fig:gsn-example}. Of course, if the abstraction is left informal, then we cannot have the same degree of assurance in the \weakChange{specification subgoal}{soundness of the argument}. The question of whether or not to formalize a predicate is left to the discretion of the assurance engineer.} as a predicate ${\ntn{\mathsf{AlarmResponse}}}$ over ${\tt LTS}$ with} \begin{align*}
    {\ntn{\mathsf{Alarm}}}&{\ntn{\mathsf{Response}}}(M) =\\[0.2em] \textrm{``}\forall \sigma \in Exec(M),\: \forall i \in \mathbb{N},\: \sigma[i] \vDash{\tt Alarm} &\implies \exists\: j \in \mathbb{N}.\: j > i \land \sigma[{\ntn{j}}] \vDash \textrm{\typofix{{\tt State}~}{\tt Safe}}\textrm{"},
    \end{align*}

    \nc{where $Exec(M)$ denotes the set of all possible executions of model $M$. \weakChange{The assertion in $g_2$ that the specification used for model checking correctly formalizes the parent goal can be defined in terms of a predicate $P_{\token{LTS} \times \token{LTL}}$}{The assertion in ${\ntn{\textbf{G2}}}$ that the specification used for model checking correctly formalizes the parent goal can be abstracted as a predicate $P_X$, such that for any LTS $M$ and LTL specification $\psi$, the language of traces satisfying $\psi$ is contained within the language of ``safe'' traces described by the predicate $\mathsf{AlarmSafe}$. Formally:  }} 
    \begin{align*}
     &\weakNew{{P_X(M,\psi)}= ``\ntn{{Words}(\psi) \subseteq \{\sigma \mid \forall i \in \mathbb{N},\: \sigma[i] \vDash{\tt Alarm}} \ntn{\implies \exists\: j \in \mathbb{N}.\: j > i \land \sigma[j] \vDash {\tt Safe}\}}",}
    \end{align*}
    
    \weakDelete{${P_{spec}(\psi) = \textrm{``}{ Words}(\psi) \subseteq \{\sigma \mid \forall i \in \mathbb{N},\: \sigma[i] \vDash{\tt Safe} \implies \exists\: j \in \mathbb{N}.\: j > i \land \sigma[i+j] \vDash {\tt Safe}\}\textrm{"},}$}

    \nc{where $Words(\psi)$ denotes the set of all traces satisfying $\psi$. \weakNew{Note that in this \emph{particular} case, the model $M$ to which $P_X$ is applied is not constrained -- only the specification $\psi$.}     \weakChange{The output predicate $P_{out} : \{\top\} \cup Cex \to \{\top,\bot\}$ simply asserts}{The goal $\textbf{G3}$ asserting that the verification was successful can similarly be abstracted as a predicate} \weakDelete{ $\textrm{\typofix{$P(r)$}{$P_{out}(r)$}} =$} \weakNew{$P_Y(r) = ``r=ok$}'' \weakDelete{i.e., that $r$ is a successful verification result}. Finally, the \weakChange{functional specification predicate $P_V$ asserts the soundness of the model checker,}{goal $\textbf{G4}$ asserting the soundness of the model checker can be formalized in terms of a predicate $P_V$ over $X \times Y$}, i.e.,} 
    $$P_V(M,\psi,r) = \textrm{``}r= ok \implies Exec(M)\subseteq Words(\psi) {"}$$
    \nc{That is, the model checker succeeds only if every execution of $M$ satisfies specification $\psi$.}

    \weakNew{Finally, the correctness criterion $C$ for this template simply requires that the LTS model $M$ which is the subject of the parent goal is the same as the model $M'$ which is used as input to the verifier, i.e., $C(M',\psi,M) = ``M' = M"$. Using these components, we can construct an analytic template for model checking $\template{\token{LTS}}{X}$ according to Def.~\ref{def:analyticTemplate}.}
\end{example}
\nc{Note that the template defined above cannot be applied to \emph{any} model checking scenario; it is fixed to support the particular predicate \weakChange{$P$ describing the response pattern between alarm and safe states. However, it is straightforward to generalize the construction. For instance, we could abstract the template over the response pattern, such that the parent predicate asserts that whenever a state satisfying $X$ is reached, a state satisfying $Y$ is eventually reached, and instantiating the template requires specifying values for $X$ and $Y$.}{$\mathsf{AlarmResponse}$.}}

\begin{proposition}
    The analytic template over LTL model checking described in Example~\ref{ex:mc_template} is \weakChange{correct-by-construction}{valid}.
\end{proposition}
\begin{proof}
    \weakChange{Fix $M \in {\tt LTS}$, $\psi \in {\tt LTL}$ and predicate $P$ as in Example~\ref{ex:mc_template}, and let $g_{spec}$, $g_{out}$, {and}  $g_{V}$ be the subgoals produced by the instantiation of this template with $M$ and $\psi$; assume each of them hold. By $g_{out}$, we have $V(M,\psi) = \top$. By $g_V$, we can infer from this $Exec(M) \subseteq Words(\psi)$. From this and $g_{spec}$ we have $P(M)$.}{ Let $X,Y, P_X, P_Y, P_V$ and $\template{\token{LTS}}{X}$ be defined as in Example~\ref{ex:mc_template}. 
    Suppose we are using the template $\template{\token{LTS}}{X}$ to support the goal $\langle M, \mathsf{AlarmResponse}\rangle$ for some LTS model $M$. To satisfy the correctness criterion, we instantiate the template by providing $M$ as the input to the model checker alongside an LTL specification $\psi$. Per Def.~\ref{def:analyticTemplate}, we obtain the subgoals $\{g_X, g_Y, g_V\}$ such that 
    \begin{align*}
        g_X &= \langle \langle M,\psi\rangle,  P_X\rangle\\  g_Y &= \langle V(M,\psi), P_Y\rangle\\ g_V &= \forall x  \in X, P_V(x, V(x)) 
    \end{align*}
    Suppose that all three of these subgoals hold. From $g_Y$, we have $V(M,\psi) = ok$. By $g_V$, we can then conclude that $Exec(M) \subseteq {Words}(\psi)$. From $g_X,$ we have $Words(\psi) \subseteq \{\sigma \mid \forall i \in \mathbb{N},\: \sigma[i] \vDash{\tt Alarm} \implies \exists\: j \in \mathbb{N}.\: j > i \land \sigma[i+j] \vDash {\tt Safe}\}$. The parent goal ${\mathsf{AlarmResponse}}(M)$ then immediately follows.
        }
\end{proof}

\weakDelete{Note that Def.~\ref{def:analyticTemplate} does not specify a correctness criterion for analytic templates. For the sake of simplicity, in what follows we assume that for all analytic templates $T = \langle P,P_{spec}, P_{out}, P_f\rangle $, the four predicates are chosen such that $T$ is correct-by-construction, i.e., it produces a sound argument for each model $M$ and specification $\psi$. It is straightforward to generalize our results to account for arguments which require some precondition on the analysis input.}

\nc{The instance of analytic argumentation shown in Fig.~\ref{fig:gsn-example}, and formalized in the preceding example, serves to produce \emph{evidence} for the AC (in this case, model checking evidence). However, we can also consider analytic templates whose primary purpose is to drive \emph{argumentation}. Consider a scenario in which we want to verify that some property $P$ holds for every element ${\ntn{e}}$ of model $M$ which satisfies some criterion $C$ ({\ntn{in which case we write $e \vDash C$}}). We would need to do two things: first, compute the set $S = \{e \in \ntn{Elem}(M) \mid e \vDash C\}$, and second, show that $P$ holds for each $x \in S$. In an AC, the former can be achieved through an analytic template, e.g., by running a model query which collects elements satisfying $C$, and the latter through argumentation, e.g., domain decomposition. We can refer to this kind of analytic argumentation as being \emph{argument-producing}, since the output subgoal ${\ntn{g_Y}}$ {\ntn{produced by the model query}} would be supported by further argumentation (in this case, domain decomposition). By contrast, we can refer to the form of analytic argumentation described in Example~\ref{ex:mc_template} as being \emph{evidence-producing}, since the\weakDelete{output} subgoal ${\ntn{g_Y}}$ is supported directly by evidence \weakChange{as given by the analysis output}{(namely, the verification result itself)}.}

\nc{\weakChange{To make this illustration more concrete}{To illustrate this distinction}, we formalize the following \emph{querying template}:}
\begin{example}[Querying Template]\label{ex:queryTemplate}
    \nc{Fix an LTS $M$ which has multiple states which are labeled as ``alarm states'', meaning that when any of these states is reached, some alarm is \weakChange{issued}{raised}. Let ${\tt Alarms}(M)$ denote the set of all such alarm states. \weakChange{Furthermore, f}{F}or any state $s$ and execution $\sigma \in Exec(M)$, let \weakChange{${\tt FinallySafe}$}{$\mathsf{AlarmStateResponse}$}$(s,\sigma) = $``$\forall i \in \mathbb{N},\sigma[i] = s \implies \exists j \in \mathbb{N}, \sigma[i+j] \vDash {\tt Safe}$''\weakDelete{, where ${\tt Safe}$ is the label of a ``safe'' state}. \weakNew{Finally, we define a \emph{query result} to be a subset of $Elem(M)$ where $M$ is the model being queried. Denote the set of all possible \emph{query results} as $\mathsf{QR}$, where $M$ is inferred from the context.}} 

    \nc{Fix a \emph{query engine} \weakChange{$Q : \modlang \times Query \to [\ntn{Elem}(M)]^*$}{$Q : \token{LTS} \times \token{Query} \to \mathsf{QR}$}, i.e., a function which takes \weakNew{any LTS} model $M$ and query ${\ntn{q}} \in \token{Query}$ and returns  \weakNew{the set of elements of $M$ which} satisfy ${\ntn{q}}$. \weakNew{For simplicity, let $X$ denote $\token{LTS} \times \token{Query}$ and let ${Y}$ denote $\mathsf{QR}$.}
    With these components, we define a \emph{querying template} ${\ntn{\template{\token{LTS}}{X}}} = \langle \mathsf{AllAlarmsSafe}, I, C\rangle $\weakChange{$ = \langle P, P_{spec}, P_{out}, P_Q \rangle$}  ~as follows:} 
    \begin{itemize}
        \item \weakChange{$P(M)$}{$\mathsf{AllAlarmsSafe}(M)$} $= \textrm{``}\forall s \in {\tt Alarms}(M), \forall \sigma \in Exec(M), {\sf {\ntn{AlarmStateResponse}}}(s,\sigma)\textrm{''}$
        \item\typofix{$P_{sepc}$}{} \weakNew{The instantiation function $I$ uses the following predicates:}
        \begin{itemize}
            \item {\ntn{{$P_X$}$(M,q) = \textrm{``}\{e \in \ntn{Elem}(M) \mid e \vDash q\} = {\tt Alarms}(M)\textrm{''}$}}
            \item {\ntn{$P_Y(S) = \textrm{``}\forall s \in S, \:\forall \sigma \in Exec(M), \mathsf{AlarmStateResponse}(s,\sigma)\textrm{''}$}}
            \item {\ntn{$P_Q(M,q,S) = \textrm{``}S = \{e \in \ntn{Elem}(M) \mid e \vDash q\}\textrm{''}$}}
        \item {\ntn{The correctness criterion $C$ over $X \times Y$} is defined as $C(M',q,M) = ``M' = M"$}
        \end{itemize}
        
        % \item $P_{out}(S) = $
        % \item $P_Q(M,\psi) = \textrm{``}Q(M,\psi) = \{e \in \ntn{Elem}(M) \mid e \vDash \psi\}\textrm{''}$
    \end{itemize}
    \nc{Intuitively, this template simply supports the proposition \weakChange{$P(M)$}{$\mathsf{AllAlarmsSafe}{(M)}$} by executing the query on the model, producing a concrete set of model elements, and using this to replace the abstract set ${\tt Alarms}({\ntn{M}})$ in the parent goal.} \weakNew{As in the model checking template, the correctness criterion for the querying template merely requires that the model which is the subject of the parent goal is the same model provided to the query engine.}
\end{example}

\begin{example}\label{ex:fullRE}
    \begin{figure}[t]
        \centering
        \includegraphics[width=\linewidth]{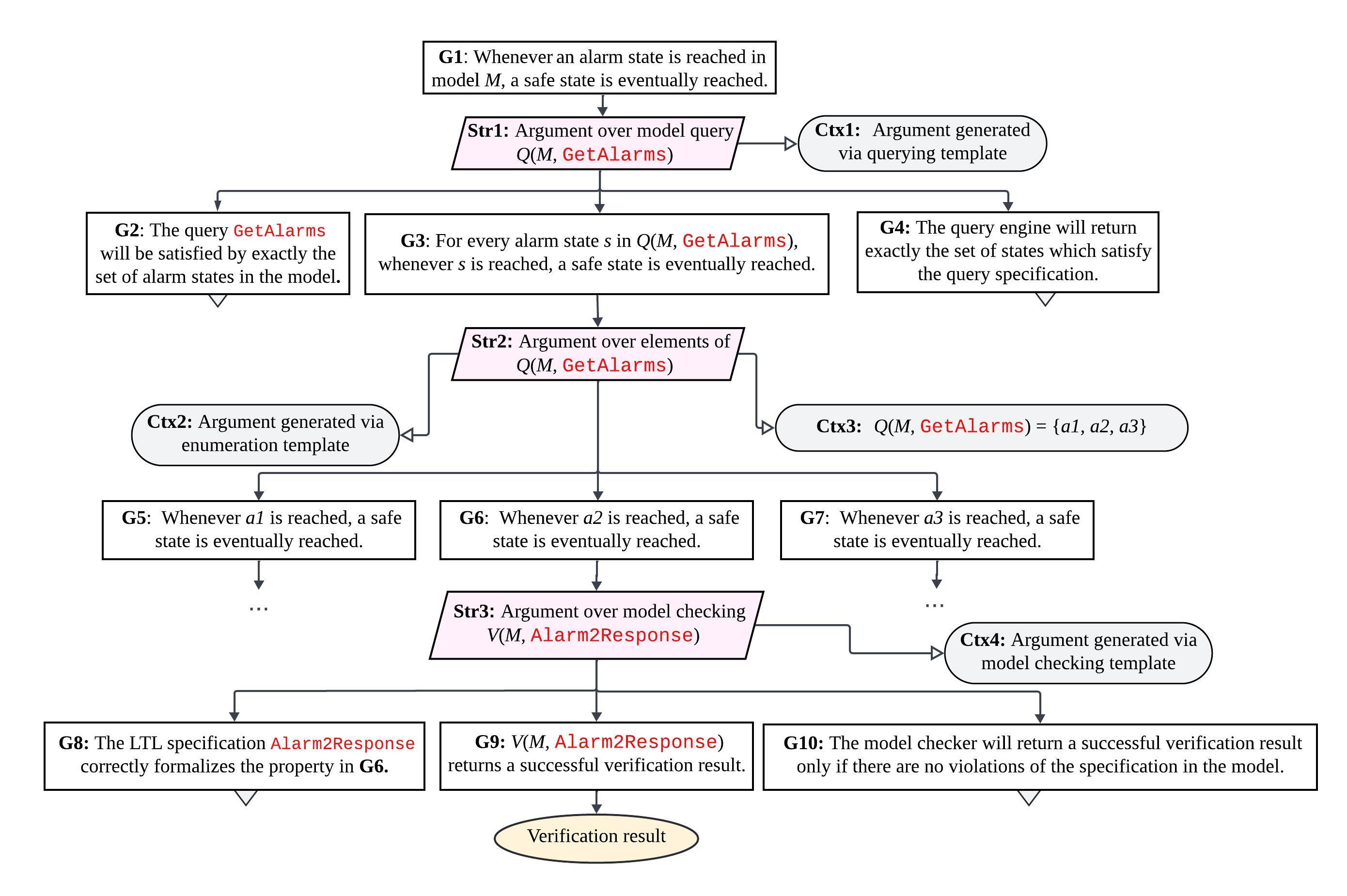}
        \caption{\bf The AC described in Example~\ref{ex:fullRE} showcasing the use of both argument-producing and evidence-producing analytic templates, connected via an enumeration template.}
        \label{fig:running_example_full_1}
        \Description[]{}
    \end{figure}
    \nc{The AC shown in Fig.~\ref{fig:running_example_full_1} demonstrates both the evidence-producing and argument-producing forms of analytic argumentation. We begin with an instance of the querying template given in {\ntn{Example}}.~\ref{ex:queryTemplate}, where ${\ntn{{\tt GetAlarms}}}$ is the query specification intended to match alarm states. \weakDelete{As in Example~\ref{ex:gsn}, the specification and functional subgoals are supported by expert review.} \weakChange{The output subgoal $g_{out}$}{The subgoal \textbf{G3} of \textbf{Str1} (corresponding to $g_Y$ of the querying template)} is further decomposed via the enumeration template, applied to the set of query results $Q(M,{\tt GetAlarms})$. In this example, there are three alarm states $\{a_1,a_2,a_3\}$, each \weakChange{with their own subgoal $g_1,g_2,g_3$  asserting the same property as was described in the root goal $g$ for their respective alarm}{of which induces its own subgoal via the enumeration template. In this way, we say that the use of the querying template is \emph{argument}-producing.} \weakChange{In this illustration,}{For brevity,} we only show the \weakChange{completion of the AC for}{decomposition of} the subgoal $\textbf{G6}$, \weakNew{asserting that the response property is satisfied with respect to alarm $a_2$.} \weakChange{which is given through the model checking template with LTL specification ${\tt alarm\_spec\_2} = {\tt G}(a_2 \implies {\tt F}({\tt Safe}))$.}{This goal is decomposed via the model checking template (Example~\ref{ex:mc_template}), given an LTL specification ${\tt Alarm2Response}$ which is specific to the alarm state $a_2$. This use of the model checking template is \emph{evidence}-producing, since \textbf{G9} (corresponding to $g_Y$ in the model checking template) is supported directly by the verification result itself. Following the use of the three templates to drive the argumentation, the AC developer is left to justify the query specification, LTL specification(s), the query engine, and the model checker (cf. goals \textbf{G2}, \textbf{G4}, \textbf{G8} and \textbf{G10}). Of course, we would expect that support for \textbf{G4} and \textbf{G10} need only be produced \emph{once}, and can then be re-used for each instantiation of their respective templates. By contrast, goals $\textbf{G2}$ and $\textbf{G8}$ are dependent on the query and LTL specification used in these particular instantiations of their respective templates.}}
\end{example}

\nc{In this section, we formalized a language of GSN-like assurance cases, and studied their development through \weakDelete{rigorous} formal templates. We also formalized the integration of software analyses as part of assurance argumentation through the concept of analytic templates. Our aim in the next section will be to \emph{lift} this framework from the level of software products to software product lines.}
\section{Lifting Assurance Case Development for Product Line Models}
\label{sec:lifting}

In the preceding section, we defined a language of assurance cases in the style of GSN, and formalized the development of ACs through the use of templates, including analytic templates. In this section, we aim to \emph{lift} this view of AC development to the level of product lines. \weakChange{We begin by defining}{After summarizing the differences between product-based and lifted AC development (Sec.~\ref{sec:lifting:overview}), we formalize} a suitable language for ``product lines of ACs'' and its  semantics (Sec.~\ref{sec:plACs}). Subsequently, we lift the notions of AC templates (Sec.~\ref{sec:liftTemplates}) and analytic templates (Sec.~\ref{sec:liftAnalytic}) to target this variability-aware language. As in the preceding section, all definitions and theorems, unless noted otherwise, are adapted from ~\cite{ifm2024}. Proofs of selected theorems not given here are provided in Appendix~\ref{app:proofs:sec4}.

\subsection{Overview of Lifted AC Development}
\label{sec:lifting:overview}

\begin{figure}[t]
    \centering
    \includegraphics[width=0.9\textwidth]{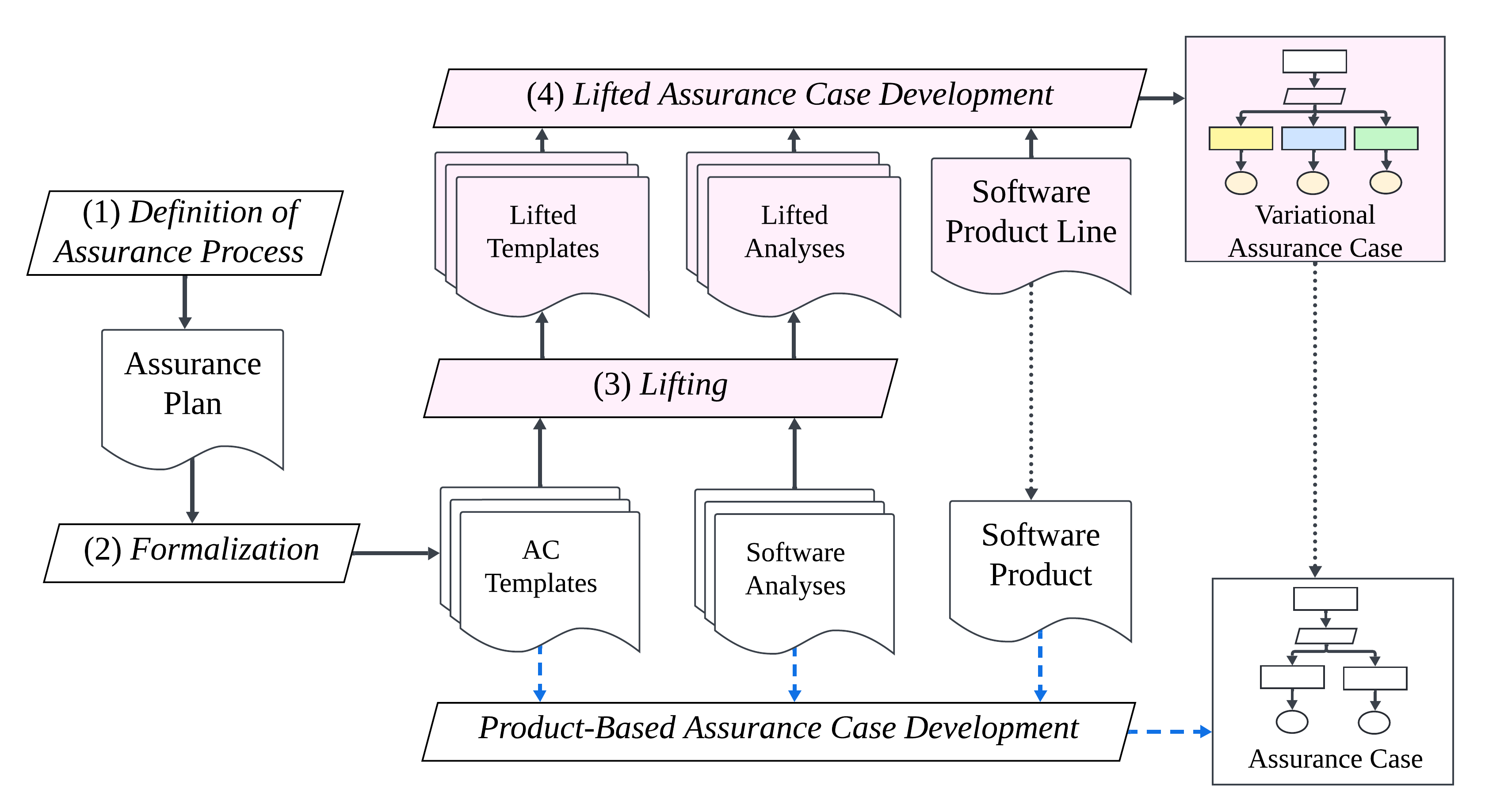}
    % \vspace{-0.2in}
    \caption{\bf  Lifting AC development from products to product lines.  Product-based activities and artifacts shown in \weakChange{blue}{white}, variational activities and artifacts shown in pink. Core activities for supporting lifted AC development are numbered 1-4. Blue dashed arrows indicate the traditional product-based AC development process; dotted arrows indicate product derivation.}
    \label{fig:workflows}
    % \vspace{-0.2in}
    \Description[]{}
\end{figure}
\nc{Fig.\typofix{}{~}\ref{fig:workflows} illustrates the necessary activities and artifacts required to support lifted AC development. The overall methodology consists of four key stages: the first two \weakChange{being instances of existing product-based assurance activities, the latter two being}{are shared with product-based AC development, while the latter two are}  specific to product line assurance.  The first stage is the \emph{definition of the assurance process} (Fig.~\ref{fig:workflows}-(1)). In this stage, a group of assurance stakeholders (e.g., safety engineers, verification engineers, managers) collectively determine an \emph{assurance} plan prescribing how assurance for the system should be produced. For instance, the assurance plan may outline which system requirements can be supported by testing, and which should be supported by formal verification. The  assurance plan can be informed by multiple factors, such as regulatory or industrial safety standards, or the resources of the assurance team.}

\nc{Once the assurance plan has been defined, aspects of it can be \emph{formalized} as AC templates (Fig.~\ref{fig:workflows}-(2)). For instance, the prescribed uses of particular analyses in the assurance plan can be used as a basis for defining corresponding analytic templates (Sec.~\ref{sec:analyticTemplates}). Other AC templates, such as domain decomposition (Example~\ref{ex:domainDecomp2}), can be identified and formalized if necessary. If we were only interested in assuring an individual software product, we could begin AC development at this stage, combining the prescribed analyses and formalized templates to produce a rigorous AC. However, if we are interested in assuring a \emph{product line}, it is infeasible to produce an AC for each of its products using only product-based methods.}

\nc{Instead, we enter the \emph{lifting} stage (Fig.~\ref{fig:workflows}-(3)). Here, we lift both the prescribed analyses and the formal AC templates which would otherwise be used in AC development. As mentioned in Sec.~\ref{sec:intro}, the lifting of software analyses is a very well-studied area of SPLE~\cite{thum2014classification,MURPHY2025112280}. By contrast, to the best of our knowledge, the lifting of arbitrary AC templates has not been rigorously studied (prior to ~\cite{ifm2024}).\footnote{An AC development process given by Nešić et al.~\cite{nevsic2021product} for component-based product lines, is a form of variability-aware assurance case development. However, it does not constitute \emph{lifted} AC development in the strict, formal sense, as the outcome of the process is a single AC which covers every product of the product line, rather than a ``product line of ACs''. Of course, it is relatively straightforward to reinterpret this process as an equivalent lifted one.} One factor preventing the rigorous study of lifting for AC templates is the lack of an appropriate language for variability-aware ACs. As per Def.~\ref{def:lift}, lifting some process (whether an analysis or a template) requires defining a variability-aware extension of its input and output spaces. Viewing a template as a function which produces an assurance case (fragment) as its output, lifting a template requires first determining an appropriate\typofix{formalism for variability-aware for ACs to target}{representation for variational ACs}.}

\nc{Assuming that such a \weakChange{variational language of}{language of variational} ACs has been given, and that product-based AC templates can be soundly lifted, we can begin \emph{lifted AC development} for the product line under assurance (Fig.~\ref{fig:workflows}-(4)). The outcome of this process is a {variational AC} which reflects the assurance-relevant variation points in the product line \weakDelete{(represented visually by coloured goals)}. Once again, the correctness of lifted AC development is defined such that deriving a \emph{product} AC from this variational AC should yield the same result as applying product-based AC development for the same product of the product line.}

\nc{\weakChange{To summarize: in order }{In summary,} to support lifted AC development, we need to do two things\weakChange{. First,}{:} we need to define an appropriate formal language for variational ACs\weakChange{; and second,}{, and} we need to formalize the \weakChange{process of lifting}{lifting of} formal AC templates to target this variational language.}

\subsection{Formalizing a Language of Variational Assurance Cases}
\label{sec:plACs}
The semantics of traditional assurance cases are essentially\typofix{boolean}{Boolean}, with a goal being either supported or unsupported. The central idea underlying variational assurance cases is to generalize these semantics to the level of \emph{sets} -- that is, sets of products in the product line. 

\weakChange{The first step of defining a language of variational assurance cases is}{We begin by} defining \emph{variational goals}\weakDelete{, i.e., goals with variability}. \weakChange{This variability}{The variability inherent in an assurance goal} can be both \emph{structural} and \emph{semantic}. Structural variability refers to the fact that a goal may be relevant only for some products and not others. Semantic variability refers to the fact that the \emph{interpretation} of a goal depends on the choice of a feature configuration. As one would expect, we can implement structural variability via presence conditions; however, the formalization of semantic variability is a bit more subtle.

\change[varSemExpl]{A variational goal is, of course, a goal associated with a product line model. But it is worth considering what kind of \emph{predicate} we want to associate with the goal. If we want to create a variational goal for a product line $\pl{x} \in \Var{X}$, we might assume that we would choose a predicate $P : \Var{X} \to \{\top,\bot\}$. However, the interpretation of this goal as $P(\pl{x})$ would \emph{not} be variational; this interpretation remains {boolean}. In fact, we do \emph{not} want to associate variational goals with predicates over product lines, but over \emph{products}. Given a product line $\pl{x} \in \Var{X}$ and predicate $P$ over $X$, $P$ cannot be interpreted with respect to $\pl{x}$ directly; but only given a choice of configuration $\conf \subseteq F$. Put another way, the interpretation of $P$ relative to $\pl{x}$ is effectively a \emph{function}  $P_{\pl{x}} : 2^F \to \{\top,\bot\}$.}{Given that the interpretation of a ``traditional'' assurance goal is a Boolean proposition, in order to properly ``lift'' the interpretation of goals for product lines, we require a \emph{variational} goal to be interpreted as a  ``product line'' of propositions -- one for each product. In the traditional setting, we define a (predicative) goal in terms of the model $M \in \modlang$ which is the subject of the goal, and a predicate $P$ over $\modlang$. What if instead we want to define a variational goal whose subject is a product line model  $\pl{M} \in \Var{\modlang}$? One may propose to apply a predicate $P$ over $\Var{\modlang}$ to $\pl{M}$, but this would not achieve the desired interpretation, since $P(\pl{M})$ would still be an ordinary proposition, rather than a product line of propositions. In fact, the desired interpretation can be achieved by associating $\pl{M}$ with a predicate over $\modlang$, i.e., a predicate over \emph{product} models. The key observation is that such a predicate $P$ over $\modlang$ cannot be applied to $\pl{M} \in \Var{\modlang}$ \emph{directly}, but only once a particular configuration $\conf$ has been fixed. If we consider $\langle \pl{M}, P\rangle$ as constituting a variational goal $\pl{g}$, we may thus interpret $\pl{g}$ as a  \emph{function} from configurations to propositions. Alternatively, we may consider the interpretation of $\pl{g} = \langle \pl{M},P\rangle$ to be the set of propositions}
$$\{P(\pl{M}|_\conf) \mid \conf \in \llbracket \Phi \rrbracket\}$$
\weakNew{with $\Phi$ being the feature model of the product line. Of course, in practice, the variational goal may not be relevant for every valid configuration under $\Phi$, so such a goal should also be annotated by a presence condition $\phi$ indicating the sets of configurations for which it is relevant.}

\weakChange{Note, however,}{Finally, note} that we do not require that \emph{every} goal in a product line AC be strictly variational in this sense -- some goals may simply be \typofix{boolean}{Boolean} propositions, as was the case in Def.~\ref{def:goal}. For instance, if a goal does not refer to a product line at all, but refers to some other aspect of system development, there may be no need for semantic variability. For the sake of simplicity, we include such goals under the umbrella of ``variational goals'' as defined below. \newText[StructDef]{We emphasize that the following definition is structurally congruent to the definition of ``traditional'' assurance goals (Def.~\ref{def:goal}), with the essential differences being (1) the use of product line models in the place of product models, and (2) the association of every variational goal to a presence condition.}

\begin{definition}[Variational Goal]
 Fix a set of features $F$ and a feature model $\Phi$. Then a \emph{variational goal} is either
 \begin{enumerate}[label=\roman*)]
     \item a pair $\langle p, \phi\rangle$ where $p$ is a proposition in first-order logic and $\phi \in {\tt Prop}(\textrm{\typofix{$f$}{$F$}})$ \newText[pcref]{is the presence condition of the goal}, or
     \item a tuple $\langle \pl{M}, P, \phi\rangle$, where $\pl{M} = \langle F, \Phi, M, \ell\rangle$ is a product line of $\mathcal{M}$-models \weakNew{(i.e., $\pl{M} \in \Var{\modlang}$)}, $P$ is a predicate over $\mathcal{M}$, and $\phi \in {\tt Prop}(F)$ is the presence condition of the goal.
 \end{enumerate}
 As in Def.~\ref{def:goal}, we refer to the two types of goals as \emph{propositional} and \emph{predicative}, respectively. \newText[pcg]{Given any variational goal $\pl{g}$, we let $\token{pc}(\pl{g})$ denote the presence condition of $\pl{g}$.} 
\end{definition}

Note that a variational goal \weakChange{$\pl{g} = \langle \pl{M},P,\phi\rangle$}{$\pl{g}$} can be interpreted as a \emph{product line of goals} over the restricted feature model \weakChange{$\Phi \land \phi$}{$\Phi \land \token{pc}(\pl{g})$}, with the \typofix{requisite}{} derivation operator defined as follows:

\begin{definition}[Derivation of Variational Goals]\label{def:vgoal}
    Fix a set of features $F$ and a feature model $\Phi$. \weakChange{Given a variational goal $g \in \Var{\Goal}$ annotated by $\phi$ and a configuration {$\conf \in \llbracket \Phi \land \phi \rrbracket$}}{Given a variational goal $\pl{g}$ and configuration $\conf \in \llbracket \Phi \land \token{pc}(\pl{g}) \rrbracket$}, the \emph{derivation} of $\pl{g}|_\conf \in \Goal$ under $\conf$ is defined as 
    \begin{align*}
         \langle p, \phi\rangle |_\conf &= p\\
        \langle \pl{M}, P, \phi\rangle|_\conf &= \langle \pl{M}|_\conf, P\rangle
    \end{align*}
\end{definition}

Accordingly, we let $\Var{\Goal}$ denote the set of all possible variational goals, when $F$ and $\Phi$ can be inferred from the context.
We can now define a product line assurance case in a manner structurally analogous to traditional assurance cases \weakNew{(Def.~\ref{def:AC})}, but with goals replaced by variational goals.

\begin{definition}[Variational Assurance Case]\label{def:placs}
Fixing a set of features $F$, the language $\Var{\AC}$ of \emph{variational assurance cases} -- also called product line assurance cases -- over $F$ is generated by the following grammar:
\begin{align*}
    \Var{\AC} := &\;{\tt Und}(\pl{g}) \tag{$g \in \Var{\Goal}$}\\
    \mid &\;{\tt Evd}(\pl{g},e) \tag{$g \in \Var{\Goal}$, $e \in {\Evd}$} \\ 
    \mid &\;{\tt Decomp}(\pl{g},st, \{\pl{A}_1\ellipses \pl{A}_n\}) \tag{$g \in \Var{\Goal}$, $st \in {\tt Str}$, \typofix{$A_i$}{$\pl{A}_i$}~$\in \Var{\AC}$}
\end{align*}
\end{definition}

Like with product assurance cases, we let ${\tt Rt}(\pl{A})$ denote the (variational) root goal of $\pl{A} \in \Var{\AC}$. We also let ${\tt pc}(\pl{A})$ denote the presence condition of the root goal of $\pl{A}$. Recalling that we have defined a derivation operator $\pl{g}|_\conf$ for variational goals, the derivation operator for $\Var{\AC}$ effectively maps this operator across the entire AC, pruning out any subtrees which are not relevant for the given configuration (as expressed through presence conditions).

\begin{definition}[Product Assurance Case Derivation]\label{def:varACDerivation} Given a variational assurance case $\pl{A}$ with features \typofix{in $f$}{$F$} and a configuration $\conf \subseteq F$, the derivation of the product assurance case for $\conf$ is defined recursively as:
\begin{align*}
    {\tt Und}(\pl{g})|_\conf &= {\tt Und}(\pl{g}|_\conf)\\
    {\tt Evd}(\pl{g},e)|_\conf &= {\tt Evd}(\pl{g}|_\conf, e)
    \\    {\tt Decomp}(\pl{g}, st, \{\pl{A}_1\ellipses  \pl{A}_n\})|_\conf &= {\tt Und}(\pl{g}|_\conf) \tag{If there is no $\pl{A}_i$ such that $\conf \vDash {\tt pc}(\pl{A}_i)$}\\
    {\tt Decomp}(\pl{g}, st, \{\pl{A}_1\ellipses \pl{A}_n\})|_\conf &= {\tt Decomp}(\pl{g}|_\conf, st,  \{\pl{A}_i|_\conf \mid\conf \vDash {\tt pc}(\pl{A}_i) \}) \tag{otherwise}
\end{align*}
 That is, whenever we encounter a goal decomposition, we ignore all sub-ACs whose root goals have presence conditions not satisfied by $\conf$, and \weakChange{recurse}{proceed recursively} through the remaining \weakChange{sub-ACs}{children}. If there are no \weakChange{sub-ACs}{children} whose root goals have satisfied presence conditions, the goal is marked as undeveloped.
\end{definition}

\begin{example}
\label{ex:derivation}
  \begin{figure}
        \centering
        \includegraphics[width=0.9\linewidth]{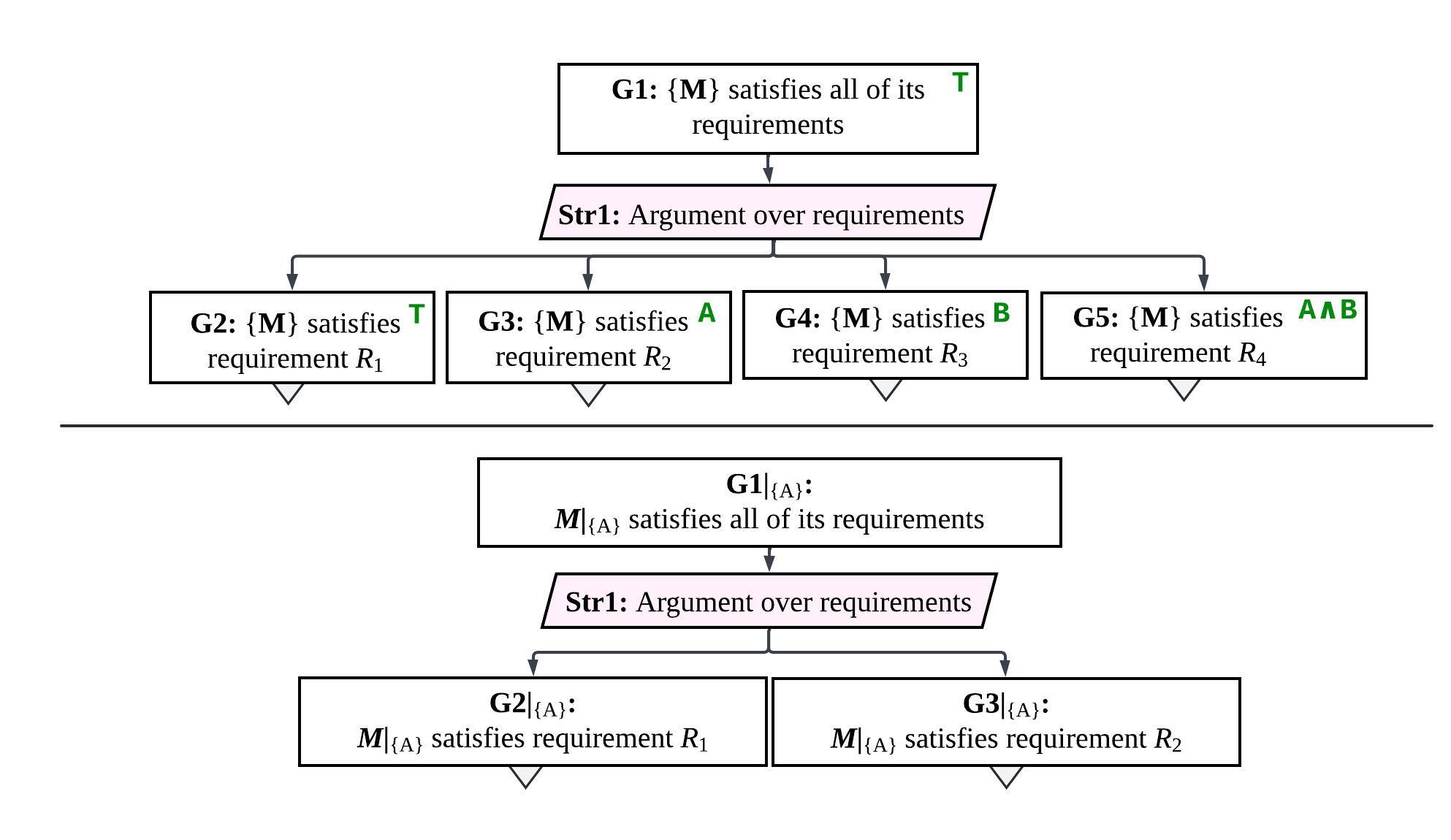}
        \caption{\bf A variational assurance case (top), and the product AC derived under $\conf = \{\token{A}\}$ (bottom), as described in Example~\ref{ex:derivation}.}
        \label{fig:derivation_example}
        \Description[]{}
    \end{figure}
Fig.~\ref{fig:derivation_example} (top) displays a variational AC fragment \weakDelete{$\pl{A}$} for a product line model $\pl{M}$ over features $F = \{\token{A}, \token{B}\}$ and feature model $\Phi = \top$. \change[abstrClar]{The variational goals are \emph{abstracted} with respect to a configuration of $\pl{M}$ (as denoted by curly brackets)}{The product line model $\pl{M}$, which is the subject of each of these variational goals, is enclosed in braces to emphasize that the goals cannot be interpreted as Boolean propositions until a configuration of $\pl{M}$ has been fixed}. The parent goal {\ntn{\textbf{G1}}} asserts that (a product of) $\pl{M}$ satisfies all of its requirements. For the sake of illustration, let us assume that four requirements $R_1,R_2,R_3, R_4$ have been defined for $\pl{M}$, with $R_1$ being a requirement for all products, $R_2$ being a requirement for products with feature $\token{A}$, $R_3$ being a requirement for \weakDelete{all} products with feature $\token{B}$, and $R_4$ being a requirement for \weakChange{any product}{products} with \weakChange{either $\token{A}$ or $\token{B}$}{both  $\token{A}$ and $\token{B}$}. Each of the variational subgoals \weakDelete{in $\pl{A}$} \typofix{assert}{asserts} that (a product of) $\pl{M}$ satisfies one of these requirements, and is annotated by the associated presence condition. The outcome of deriving the product AC associated with configuration $\{\token{A}\}$ is shown in Fig.~\ref{fig:derivation_example} (bottom). All goals with presence conditions not satisfied by $\{A\}$ (namely, ${\ntn{\textbf{G4}}}$ and ${\ntn{\textbf{G5}}}$) have been pruned, and each of the remaining goals \weakChange{has been instantiated $\left(\pl{M}|_{\{A\}}\right)$}{being instantiated with subject model $\pl{M}|_{\{\token{A}\}}$}.
\end{example}

Note that Def.~\ref{def:placs} does not place any constraints on the presence conditions placed on variational goals. Consider a situation in which a parent goal $\pl{g}$, annotated by $\phi$, \weakDelete{is decomposed by some strategy and} has a subgoal $\pl{g}^\prime$, annotated by $\phi^\prime$ such that $\llbracket \phi^\prime\rrbracket \not\subseteq \llbracket \phi \rrbracket$. This would mean that there are \weakChange{\emph{fewer}}{some} configurations for which we \weakNew{do not} need to support $\pl{g}$, \weakChange{than}{but} for which we need to support $\pl{g}^\prime$. Conceptually, this runs contrary to the intention that subgoals \change{should be easier to support than}{are used to support} their parent goals. Operationally, this also reveals a limitation of the derivation operator specified in Def.~\ref{def:varACDerivation}, since we would never be able to derive $\pl{g}^\prime$ under configurations in $\llbracket \phi^\prime \land \neg \phi \rrbracket$. Rather than defining a more complex derivation process to handle such situations, we will instead enforce the following \emph{well-formedness} condition on presence conditions in variational \typofix{ACS}{ACs}:

\begin{definition}[Well-formed Product \typofix{line}{Line} Assurance Case]\label{def:wellformed}
    A decomposition ${\tt Decomp}(\pl{g}, st, \{\pl{A}_1\ellipses  \pl{A}_n\})$ is said to be \emph{well-formed} if, given that $\pl{g}$ has presence condition $\phi$, for each $\pl{A}_i$ with ${\tt pc}(\pl{A}_i) = \phi_i$, we have $\phi_i \implies \phi$. A product line assurance case \pl{A} is well-formed if each of its\change[arg]{arguments}{decompositions}\typofix{are}{is} well-formed.
    \end{definition}

\change[recallSupp]{We now formalize the support semantics for variational assurance cases -- effectively, a lifting of Def.~\ref{def:supp}. As mentioned above, the fundamental idea is to lift the semantics of the AC to the level of \emph{sets} of products}{We now formalize what it means for a variational assurance case to be ``supported''. Recall from Def.~\ref{def:supp} that a \emph{product} assurance case is supported iff (1) each of its leaf-level goals are supported by adequate evidence, and (2) each goal decomposition in the AC is a sound goal refinement. Our aim now is to \emph{lift} this notion for product lines of ACs.} \weakDelete{This lifting needs to be done for evidence as well as for {argumentation}.} \newText[motInv]{Suppose we wish to determine whether a variational assurance case $\pl{A}$ is supported, in the sense that every \emph{product} assurance case derivable from $\pl{A}$ under any configuration of the product line is supported. This in turn would require (by Def.~\ref{def:supp}) that (1) for every configuration $\conf$ of the product line, and for every leaf-level goal $\pl{g}$ of $\pl{A}$ which is present under $\conf$, there is adequate evidence for $\pl{g}|_\conf$; and (2) for every configuration $\conf$ of the product line, and for every internal goal $\pl{g}$ which is present under $\conf$, the decomposition supporting $\pl{g}|_\conf$ derived under $\conf$ provides a sound goal refinement. We note that both (1) and (2) are asserting \emph{invariants} over the set of configurations of the product line. This similarity between the two criteria motivates the following definition of \emph{configuration invariance}:} 
\weakDelete{To make this notion sufficiently generic and precise, we define \emph{configuration invariance} as follows:}
\begin{definition}[Configuration Invariance]\label{def:inv}
Let $P$ be a predicate over $X$, and let $\pl{x} \in \Var{X}$ have features $F$ and feature model $\Phi$. Given any feature expression $\phi \in {\tt Prop(F)}$, the predicate $P$ is \emph{invariant in } $\pl{x}$ \emph{with respect to $\phi$}, denoted $\Inv{P}{\pl{x}}{\phi}$, if and only if $\forall \conf \in \llbracket\Phi\rrbracket, \conf \vDash \phi \implies P(\pl{x}|_\conf)$.
\end{definition}
\newText[explInv]{That is, given a product line $\pl{x}$ with feature model $\Phi$, $\Inv{P}{\pl{x}}{\phi}$ states that $P(\pl{x}|_\conf)$ holds for every product of $\pl{x}$ derivable under any valid configuration $\conf$ satisfying $\phi$.} Configuration invariance can then be used to to define \emph{variational} support of an assurance case.

\begin{definition}
\label{def:vsupp}
[Variationally Supported Assurance Case]
The \emph{variational support} predicate {\ntn{$\Supp^\uparrow$}} over $\Var{\AC}$ is defined inductively from the following inference rules:
\[  \infer[{\tt [Supp^\uparrow\text{-}1]}]{{\Supp^\uparrow}({\tt Evd}(\pl{g},e))}{
    &\pl{g} = \langle  p, \phi\rangle
    & \ntn{e \vdash p}
    }\qquad
    \infer[{\tt [Supp^\uparrow\text{-}2]}]{{\Supp^\uparrow}({\tt Evd}(\pl{g},e))}{
    &\pl{g} = \langle \pl{M}, P, \phi\rangle
    & \ntn{e \vdash}\weakNew{ ~\Inv{P}{\pl{M}}{\phi}}
    }
\]
\[ 
     \infer[{\tt [\Supp^\uparrow\text{-}3]}]{\Supp^\uparrow({\tt Decomp}(\pl{g},\typofix{}{st,}\mathcal{A}))}{
     & \pl{g} = \langle \pl{M}, P, \phi\rangle 
    & \mathcal{A} = \{\pl{A}_1,...\pl{A}_n\}
    & \forall \pl{A}_i \in \mathcal{A}, \Supp^\uparrow(A_i)
    & \weakNew{\Inv{\prec}{\langle \mathcal{A}, \pl{g}\rangle}{\phi}}
    }
\]
Note that the invariance condition in rule $[\Supp^\uparrow$-$3]$ is interpreted as $$\forall \conf \in \llbracket \Phi \rrbracket, \conf \vDash \phi \implies \mathcal{A}|_\conf \prec \pl{g}|_\conf$$ with $\mathcal{A}|_\conf = \{\pl{A}_i|_\conf \mid\conf \vDash {\tt pc}(\pl{A}_i) \}$. \weakNew{That is, for every configuration $\conf$, having derived $\mathcal{A}|_\conf$ as the set of child ACs whose root goals are present under $\conf$, we require that $\mathcal{A}|_c \prec \pl{g}|_\conf$.} 
\end{definition}
\newText[suppExpl]{Let us quickly summarize the intuitive meaning of these inference rules:
\begin{itemize}[itemsep=0.3em]
    \item { } [$\Supp^\uparrow$-1] asserts that for a propositional variational goal $\pl{g} = \langle p, \phi\rangle$ to be supported by evidence, we require the evidence $e$ to be adequate to support proposition $p$. This is exactly the same as rule [$\Supp$-1] in Def.~\ref{def:supp}, which reflects the fact that propositional goals -- lacking variability by definition -- have the same semantics in variational ACs as in traditional ACs.
    \item { } [$\Supp^\uparrow$-2] asserts that for a predicative variational goal $\pl{g} = \langle \pl{M}, P,\phi\rangle$ to be supported by evidence, we require the evidence $e$ to be adequate to support the proposition $\Inv{P}{\pl{M}}{\phi}$. That is, $e$ must provide adequate evidence for the proposition $P(\pl{M}|_\conf)$ for every configuration $\conf$ under which the goal is present.
    \item {} [$\Supp^\uparrow$-3] asserts that (1) for every internal variational goal $\pl{g}$ which is present under some configuration $\conf$, the \emph{product-level} decomposition obtained by derivation under $\conf$ is a sound goal refinement, and (2) each of the sub-ACs under $\pl{g}$ are variationally supported. In particular, note that the definition of $\prec$ (Def.~\ref{def:goalRefinement}) excludes the possibility of $\pl{g}|_\conf$ being soundly refined by an empty set of subgoals. Consequentially,  $[\Supp^\uparrow$-3] can never be applied if there is some configuration $\conf$ under which $\pl{g}$ is present, but no subgoal of $\pl{g}$ is present under $\conf$.
\end{itemize}
}
We can \weakDelete{quickly} verify by induction that $\Supp^\uparrow$ is indeed a correct lift of $\Supp$.
\begin{theorem}\label{thm:supplift}
    If {\rm $\pl{A}$} is well-formed, \newText[weakNew]{then} $\Supp^\uparrow(${\rm $\pl{A}$}) \weakChange{$\equiv$}{if and only if}  $\;\mathsf{Inv}(\Supp,$ {\rm $\pl{A}$}$,\phi$), where $\phi = \token{pc}(${\rm \pl{A}}$)$.
\end{theorem}

\subsection{Lifting Assurance Case Templates}
\label{sec:liftTemplates}
\weakChange{As p}{P}er Thm.~\ref{thm:supplift}, the correctness of a variational assurance case depends on the ability to verify that a strategy is invariantly a sound goal refinement. This is obviously a more involved task than verifying a \emph{single} goal refinement; as such, the use of rigorous and verifiable templates is even more \weakChange{crucial}{essential} in developing variational ACs. Furthermore, we do not want to begin designing templates for variational ACs from scratch: developing and verifying \weakChange{templates even for traditional ACs can}{product-based templates can already} be a complex and expensive process. Instead, we want to \weakChange{\emph{reuse} this effort as much as possible -- hence, if we lift a valid template, we want to be sure that this process preserves validity}{\emph{lift} product-based templates in a manner which preserves their validity (at the product line level)}.

\begin{definition}[Lifting Templates]\label{def:liftTemplate}
    Let ${\ntn{\template{\modlang}{D}}} = \langle P, I,C \rangle$ be a template \weakChange{with $D$ being the domain of data used for}{for supporting predicate $P$ over $\modlang$ with} instantiation function $I$ and correctness criterion $C$.  Then {lifting} ${{\template{\modlang}{D}}}$ requires producing a \emph{lifted} instantiation function $${I}^\uparrow : \Var{D} \to \mathbb{P}(\Var{\Goal})$$ 
    \weakNew{That is, the lifted template is instantiated by providing a \emph{product line} of data $\pl{x} \in \Var{D}$. The requirement that $I^\uparrow$ is a \emph{lift} of $I$ is interpreted as follows: whenever $I^\uparrow$ is used to support a variational goal $\langle \pl{M},P,\phi\rangle$ with $\pl{M} \in \Var{\modlang}$, for any $\pl{x} \in \Var{D}$ and configuration $\conf \in \llbracket \phi\rrbracket$, we require $I^\uparrow(\pl{x})|_\conf = I(\pl{x}|_\conf)$. }

    \weakDelete{i.e., for all $\pl{x} \in \Var{D}$, and each  configuration $\conf \in \llbracket \phi \rrbracket$, we have
${I}^\uparrow(\pl{x},\phi)|_\conf = \textrm{{$I(\pl{x}|_\conf)$}}$.} 

\weakDelete{The \emph{variational correctness criterion} for the lifted template is}{}
\weakDelete{{${C}^\uparrow : \Var{D}  \times {\Var{\mathcal{M}}} \times {\tt Prop}(F) \to \{\top,\bot\}$,} \weakDelete{defined by ${C}^\uparrow(x,\pl{M},\phi) = {\Inv{\langle \pl{x},\pl{M}\rangle}{C}{\phi}}$}}

{Given that $C : D \times \modlang \to \{\top,\bot\}$ was the correctness criterion for the original template, we define the \emph{variational} correctness criterion for the lifted template as $\Inv{C}{\langle \pl{x},\pl{M}\rangle}{\phi\rangle}$. That is, the argument produced by the lifted template should be provably sound (in the variational sense) if for the given $\pl{x} \in \Var{D}$ and $\pl{M} \in \Var{\modlang}$, for all $\conf \in \llbracket \phi \rrbracket$, we have $C(\pl{x}|_\conf, \pl{M}|_\conf)$. }

\end{definition}
We denote the lifted template produced by this process as \weakNew{$\template{\modlang}{D}^\uparrow = \langle P, I^\uparrow, C\rangle$ to emphasize that the sole requirement for lifting the template is the lifting of the instantiation function. The following theorem establishes that any template which has been lifted via the preceding construction preserves validity of the product-based template across the entire product line.}

\begin{theorem} \label{thm:liftTemplateCorrect} \ntn{Let $\pl{g} = \langle \pl{M}, P,\phi\rangle$ be a variational goal with $\pl{M} \in \Var{\modlang}$ for some language $\modlang$ and with feature model $\Phi$. If}
\begin{enumerate}
    \item {\ntn{$\template{\modlang}{D}= \langle P, I, C\rangle$ is a valid template, 
    \item $\template{\modlang}{D}^\uparrow = \langle P, I^\uparrow, C\rangle$ is a lifted template conforming to Def.~\ref{def:liftTemplate},}} 
    \item $\pl{x} \in \Var{D}$ such that $\Inv{C}{\langle \pl{x},\pl{M}\rangle}{\phi}$ holds,
    \item $\pl{G} = \{\pl{g}_1\ellipses \pl{g}_n\} = I^\uparrow(\pl{x})$,
\end{enumerate} then $\Inv{\prec}{\langle \pl{G},\pl{g}\rangle}{\phi}$ holds. That is, for every $\conf \in \llbracket \Phi \land \phi \rrbracket$, given $\pl{G}|_\conf = \{\pl{g}_i|_\conf \mid \conf \vDash \token{pc}(\pl{g}_i)\}$, we have $\pl{G}|_\conf \prec \pl{g}|_\conf$.

\weakDelete{$\forall \pl{x} \in \Var{D}, \forall \pl{M} \in \Var{\mathcal{M}}, \forall \phi \in {\tt Prop}(F),\; {C}^\uparrow(\pl{x}, \pl{M},\phi) \implies$}

\weakDelete{$ {\Inv{\prec}{\langle \pl{G},\pl{g}\rangle}{\phi}}$ where $\pl{g} = \langle P, \pl{M}, \phi\rangle$ and $\pl{G} = {I}^\uparrow(\pl{x},\phi)$.}
\end{theorem}
\begin{proof} 
    \nc{\ntn{Assume that $\template{\modlang}{D}$ is a valid template. Suppose that $\template{\modlang}{D}$ is used to decompose a goal $\pl{g} = \langle \pl{M},P,\phi\rangle$, and is instantiated with $\pl{x} \in \Var{D}$ such that $\Inv{C}{\langle \pl{x}, \pl{M}\rangle}{\phi}$ holds. Let $\pl{G} =  \{\pl{g}_1\ellipses  \pl{g}_n\}$ be the resulting variational goals.}} Fix some configuration $\conf \in \llbracket \typofix{\conf}{\phi} \rrbracket$ for which we need to show that $\{\pl{g}_1\ellipses  \pl{g}_n\}|_\conf \prec \pl{g}|_\conf$. Let $\{\pl{g}_j\ellipses  \pl{g}_k\}$ be the subset of variational goals whose presence conditions are satisfied by $\conf$. Since ${\ntn{I^\uparrow}}$ is a lift of $I$, we know that $I(\pl{x}|_\conf) = \{\pl{g}_j|_\conf\ellipses  \pl{g}_k|_\conf\}$. Since \weakNew{$\template{\modlang}{D}$} is valid, and $C(\pl{x}|_\conf,\pl{M}|_\conf)$ holds, we know that $I(\pl{x}|_\conf) \prec \pl{g}|_\conf$. By generalization of $\conf$, we have \typofix{$\Inv{\langle \pl{G},\pl{g}\rangle}{\prec}{\phi}$}{$\Inv{\prec}{\langle \pl{G},\pl{g}\rangle}{\phi}$}.
\end{proof}

% \vspace{-2mm}
% \begin{lstlisting}[mathescape=true]
% theorem lift_sound (T : Template A D) (T' : vTemplate α γ) ($\phi$ : PC $\mathbb{F}$) :
%     valid T ∧ isLift T.inst T'.inst →
%     ∀ x d, [Φ| $\phi$] T.prec (x,d) → deductive (T'.inst (x, d)) 
% \end{lstlisting}
% \vspace{-2mm}
% \old{where $\token{x} \in \alpha$, $\token{d} \in \gamma$ and $\phi$ is the presence condition over which the template is instantiated.  Intuitively, the above theorem provides a structured approach to ``lifting'' the correctness of product-level templates, which is established during step C2 of \process. We thus reduce lifting and verifying PL AC templates to (i) lifting instantiation functions and (ii) creating variational proofs. }
% \vspace{-0.1in}
\begin{example}[Variational Domain Decomposition] \label{ex:vdomdecomp} 
We can lift the domain decomposition template defined in Example~\ref{ex:domainDecomp2} as follows. Fixing a universe $U$ and predicate $P$ over $U$, the product-based template is used to refine a parent goal of the form $\langle S, \ntn{\mathsf{Forall}_P}\rangle$, with $S \in {\ntn{\mathbb{P}(U)}}$ and ${\ntn{\mathsf{Forall}_P}(S) = \textrm{``}\forall x \in X, P(x)"}$. The \emph{lifted} template will apply this predicate $\ntn{\mathsf{Forall}_P}$ to \emph{variational sets} $\pl{S} \in \ntn{\Var{\mathbb{P}(U)}}$, i.e., sets whose elements are annotated by presence conditions~\cite{shahin2021towards}.\footnote{Formally, $\ntn{\Var{\mathbb{P}(U)} = \mathbb{P}({U \times {\tt Prop}(F)})}$, and derivation of a set under $\conf$ is defined as $\pl{S}|_\conf = \{x \mid (x, \phi) \in \pl{S} \land \conf \vDash \phi\}$.} 

Consider the original \weakDelete{template} instantiation function {\ntn{$I : \mathbb{F}(U) \to \mathbb{P}(\Goal)$}, defined for finite families of subsets of $U$}. In order to lift $I$, we need to decide how to represent a product line of (finite) families over $U$. A natural choice is to represent them as (finite) families of annotated sets, i.e., $\ntn{\Var{\mathbb{F}(U)} \subset \mathbb{P}(\Var{\mathbb{P}(U)})}$, with derivation defined as $\{\pl{X}_1\ellipses \pl{X}_n\}|_\conf = \{\pl{X}_1|_\conf\ellipses \pl{X}_n|_\conf\}$

\weakChange{Then the lifted instantiation is taken as ${I^\uparrow(\{\pl{X}_1\ellipses \pl{X}_n\}) = \{\pl{g}_1\ellipses  \pl{g}_n\}}$, with $\pl{g}_i = \langle \pl{X}_i, \mathcal{P}, \phi \rangle$}{Then, to instantiate the lifted template to support $\pl{g} = \langle \pl{S}, \mathsf{Forall}_P,\phi\rangle$, we take $I^\uparrow(\{\pl{X}_1\ellipses\pl{X}_n\}) = \{\pl{g}_1\ellipses \pl{g}_n\}$, where each $\pl{g}_i = \langle \pl{X}_i, \mathsf{Forall}_P, \phi\rangle$}.
\end{example}
\begin{proposition}\label{prop:enumLiftCorrect} The instantiation function of the variational domain decomposition (Example~\ref{ex:vdomdecomp}) correctly lifts the instantiation function defined in Example~\ref{ex:domainDecomp2}.
\end{proposition}

By combining the above proposition with Thm.~\ref{thm:liftTemplateCorrect} {\ntn{and the validity of traditional domain decomposition (Prop.~\ref{prop:domaindecompPrf})}}, we obtain:
\begin{proposition}
    The variational domain decomposition template defined in Example~\ref{ex:vdomdecomp} is valid.
\end{proposition}

\change[varExpl]{It should be clear, however, that this lifting of domain decomposition only produces AC fragments with \emph{semantic} variability, and not with syntactic variability, since the same presence condition $\phi$ is applied to all goals. We will shortly consider a modified form of domain decomposition which \emph{does} introduce syntactic variability in the AC.}{Consider the kind of variability that is present in the strategies produced via this lifted form of domain decomposition. As mentioned in Sec.~\ref{sec:plACs}, a product line assurance case can exhibit two forms of variability: \emph{structural}, in which different parts of the assurance case (e.g., different goals) are only relevant for certain products, and \emph{semantic}, in which the interpretation of an assurance goal depends on the choice of feature configuration. Strategies produced by lifted domain decomposition, as defined above, exhibit semantic variability, since each subgoal $\pl{g}_i = \langle \pl{X}_i, \mathsf{Forall}_P, \phi\rangle$ refers to a product line $\pl{X}_i$ of subsets of $U$, and interpreting $\mathsf{Forall}_P$ with respect to $\pl{X}_i$ requires fixing a configuration $\conf$ and deriving the set $\pl{X}_i|_\conf$. However, the arguments produced using this template do not introduce any \emph{additional} structural variability, since each subgoal inherits the same presence condition $\phi$ from the parent goal which is being refined. That is, if the parent goal $\pl{g}$ is present under configuration $\conf$, then so are all of the subgoals produced by the template. Conversely, if $\pl{g}$ is not present under $\conf$, then neither are any of the subgoals produced by the template. We will shortly consider a modified form of domain decomposition which \emph{does} introduce syntactic variability in the AC.}

\subsection{Lifting Analytic Templates}
\label{sec:liftAnalytic}
\weakChange{As mentioned in the preceding section,}{Just as} the complexity of manually verifying variational goal refinements makes the use of lifted templates strongly compelling\weakChange{. T}{, t}he complexity of analyzing product lines makes the use of lifted analyses \weakDelete{as a means of producing variational evidence} similarly compelling. Given the integration of analysis and argumentation through analytic templates (Sec.~\ref{sec:analyticTemplates}), it is worth providing special attention to \weakNew{lifting} this particular class of templates. \weakChange{The general idea is to modify Def.~\ref{def:analyticTemplate} }{To this end, we can lift analytic templates (Def.~\ref{def:analyticTemplate})}   such that\weakDelete{(i) the analytic template is instantiated using a product line model, rather than a product model, and (ii)} the instantiation of the template executes a \emph{lifted} analysis, rather than a product-based analysis.

\begin{definition}[Lifted Analytic Template]\label{def:analyticTemplateLift}
    \weakChange{Fixing a product-based analysis $f : \mathcal{M} \times \mathcal{L}\to Y$, let $T_f = \langle P, P_{spec}, P_{out}, P_f\rangle $ be an analytic template for $f$, with parent predicate $P : \mathcal{M} \to \{\top,\bot\}$. Let $f^\uparrow : \Var{\mathcal{M}} \times \mathcal{L} \to \Var{Y}$. Then for any product line model $\pl{M}$, specification $\psi \in \mathcal{L}$ and presence condition $\phi$, the lifted analytic template induced by $f^\uparrow$ uses the instantiation $I_f^\uparrow(\pl{M},\psi,\phi) = \{\pl{g}_{spec},\pl{g}_{out}, \pl{g}_f,\pl{g}_{f^\uparrow}\}$, where}{Let $f : X \to Y$ be a product-based analysis with domain $X$ and codomain $Y$, and let $\template{\modlang}{X}$ be an analytic template for $f$, i.e., $\template{\modlang}{X}$ is used to decompose goals of the form $\langle M, P\rangle$ with $M \in \modlang$. Suppose that the instantiation function $I$ of $\template{\modlang}{X}$ is defined using predicates $P_X$ over $X$, $P_Y$ over $Y$, and $P_f$ over $X \times Y$, as per Def.~\ref{def:analyticTemplate}. Suppose that $f^\uparrow : \Var{X} \to \Var{Y}$ is a lift of $f$. Then the \emph{lifted analytic template} $\template{\modlang}{X}^\uparrow$ \emph{induced by $f^\uparrow$} is defined as $T^\uparrow = \langle P, I^\uparrow, C\rangle$, with the lifted instantiation $I^\uparrow : \Var{X} \to \mathbb{P}(\Goal)$ being defined for any $\pl{x} \in \Var{X}$ as }
        \[\ntn{I^\uparrow(\pl{x}) = \{\pl{g}_X, \pl{g}_Y, \pl{g}_f, \pl{g}_{\mathsf{Lift}}\}}\]
    \ntn{where}
    \begin{align*}
     \pl{g}_X &= \langle \pl{x}, P_X, \phi\rangle \\    
     \pl{g}_Y &= \langle f^\uparrow(\pl{x}), P_Y,\phi \rangle \\
     \pl{g}_f &= \langle \forall x \in X, P_f(x,f(x)), \phi \rangle \\ 
     \pl{g}_{\mathsf{Lift}} &= \langle \forall \pl{x} \in \Var{x}, \forall \conf \in \llbracket \Phi \rrbracket, f^\uparrow(\pl{x})|_\conf = f(\pl{x}|_\conf), \phi\rangle 
    \end{align*}
    
    \weakDelete{$\pl{g}_{spec} = \langle \psi, P_{spec} , \phi\rangle, \quad  \pl{g}_{out} = \langle f^\uparrow(\pl{M},\psi), P_{out}, \phi\rangle\quad \pl{g}_f = \langle g_f, \phi \rangle \quad \pl{g}_{f^\uparrow} = \langle \mathsf{Lifts}(f,f^\uparrow), \phi\rangle$} 
    % \ntn{}
    
\weakChange{where $g_f = \textrm{``}\forall M \in \mathcal{M}, \forall \psi \in \mathcal{L}, P_f(M,\psi)$ and  $\mathsf{Lifts}(f,f^\uparrow)$ asserts that $f^\uparrow$ is a correct lift of $f$. }{and where $\phi$ is the presence condition of the parent goal being decomposed and $\Phi$ is the feature model of the product line being analyzed. Note that goals $g_f$ and $g_{\mathsf{Lift}}$ are propositional.}
\end{definition}

\change[liftAnExpl]{We make a few observations about Def.~\ref{def:analyticTemplateLift}. First, unlike the lifting of general templates (Def.~\ref{def:liftTemplate}), the lifting of analytic templates does not require any work on the part of the AC developer, besides producing the lifted analysis. Second, the predicates $P_{spec}$, $P_{out}$, and the propositional goal $g_f$ are reused analogously to the original template; moreover, existing assurance of ${g}_f$ {can be} reused as-is for instantiations of the lifted template. Third, the output predicate $P_{out}$ is now ``applied'' to the lifted analysis result $f^\uparrow(\pl{x})$, as we have now replaced the analysis $f$ by the (likely much more efficient) analysis $f^\uparrow$. Fourth, the only new subgoal asserts the correctness of the lift $f^\uparrow$. If the correctness of the lift is established, then we can conclude that the entire argument over the lifted analysis is just as sound as the original analytic template.}{Let us  elaborate on what is being described in Def.~\ref{def:analyticTemplateLift} and make some observations. First, as with product-based analytic templates, instantiating a lifted analytic template requires executing the analysis over which the template is defined, with the key difference being that we are now executing the \emph{lifted} analysis $f^\uparrow$ on a variational input $\pl{x} \in \Var{X}$, hence the use of $f^\uparrow(\pl{x})$ as the subject of $\pl{g}_Y$. Second, the goal $\pl{g}_f$ asserts exactly the same proposition as $g_f$ in the product-based template, i.e., the correctness of the product-based analysis $f$ with respect to the predicate $P_f$. Third, the goal $\pl{g}_{\mathsf{Lift}}$ simply asserts that $f^\uparrow$ is a correct lift of $f$. Fourth, since $\pl{g}_f$ and $\pl{g}_{\mathsf{Lift}}$ are both propositional, supporting these goals does not (generally) depend on the particular product line model under assurance: the goal $\pl{g}_f$ refers solely to the product-based analysis $f$, and although the goal $\pl{g}_{\mathsf{Lift}}$ does refer a particular feature model $\Phi$, in practice, lifted analyses are designed to satisfy this criterion for \emph{any} product line with \emph{any} feature model~\cite{MURPHY2025112280}. Thus, adequate support for goals $\pl{g}_f$ and $\pl{g}_{\mathsf{Lift}}$ can generally be produced \emph{once} and then re-used for every instantiation of the lifted analytic template. Finally, we emphasize that the \emph{only} requirement for lifting the template is the production of the lifted analysis $f^\uparrow$, since the predicates $P_X$, $P_Y$, and $P_f$ are reused \emph{verbatim} from the original analytic template.}

\weakNew{The following theorem establishes that the lifting of any analytic template as described in Def.~\ref{def:analyticTemplateLift} preserves validity.}

\begin{theorem}\label{thm:liftAnalyticCorrect}
    \ntn{Let $\template{\modlang}{X}$ be an analytic template over  $f : X \to Y$, and let $f^\uparrow : \Var{X} \to \Var{Y}$. Let $\template{\modlang}{X}^\uparrow$ be the result of lifting $T$ with respect to $f^\uparrow$ as defined in Def.~\ref{def:analyticTemplateLift}. If ~$\template{\modlang}{X}$ is valid, then $\template{\modlang}{X}^\uparrow$ is valid.}
\end{theorem}
\begin{proof}
    {Let {\ntn{$\template{\modlang}{X}  = \langle P,I,C\rangle$ be an analytic template for $f : X \to Y$ and let $f^\uparrow : \Var{X} \to \Var{Y}$}. Let $\template{\modlang}{X}^\uparrow$ be the lifted template induced by $f^\uparrow$.} Let the instantiation function $I$ of $\template{\modlang}{X}$ be defined using predicates $P_X$ over $X$, $P_Y$ over $Y$, and $P_f$ over $X \times Y$. Suppose that $\template{\modlang}{X}^\uparrow$ is used to decompose a variational goal $\pl{g} = \langle \pl{M}, P, \phi\rangle$ by executing $f^\uparrow$ on $\pl{x} \in \Var{X}$, \weakNew{with $\pl{x}$ satisfying $\mathsf{Inv}(C,\langle \pl{x}, \pl{M}\rangle,\phi)$,} producing the variational subgoals $\{\pl{g}_X, \pl{g}_Y, \pl{g}_f, \pl{g}_{\mathsf{Lift}}\}$ as per Def.~\ref{def:analyticTemplateLift}.  We need to show that for all configurations $\conf \in \llbracket{\phi}\rrbracket$, we have $$\{\pl{g}_{X}, \pl{g}_{Y}, \pl{g}_f, \pl{g}_{\mathsf{Lift}}\}|_\conf \prec \pl{g}|_\conf$$ Since all goals have the same presence condition $\phi$, we have $$\{\pl{g}_{X}, \pl{g}_{Y}, \pl{g}_f, \pl{g}_{\mathsf{Lift}}\}|_\conf = \{\pl{g}_{X}|_\conf, \pl{g}_{Y}|_\conf, \pl{g}_f|_\conf, \pl{g}_{\mathsf{Lift}}|_\conf\}$$ Thus, to prove the refinement, we need to verify }
    \begin{align*}
        &\{\pl{g}_{X}|_\conf \land \pl{g}_{Y}|_\conf \land \pl{g}_f|_\conf \land \pl{g}_{\mathsf{Lift}}|_\conf\} \prec \pl{g}|_\conf\\
        \equiv~ & P_{X}(\pl{x}|_\conf) \land P_{Y}(f^\uparrow(\pl{x})|_\conf) \land g_f \land (\forall \pl{x} \in \Var{x}, \forall \conf \in \llbracket \Phi \rrbracket, f^\uparrow(\pl{x})|_\conf = f(\pl{x}|_\conf)) \implies P(\pl{M}|_\conf) \tag{1}
    \end{align*}
    where $g_f = \forall x \in X, P_f(x,f(x))$, {\ntn{and $\equiv$ denotes logical equivalence} (by the definitions of $\prec$ and derivation of goals)}.
    {From the validity of $\template{\modlang}{D}$, we know that for any $x \in X$ and $M \in \modlang$ such that $C(x,M)$ holds, we have}
    \begin{align*}
        P_{X}(x) \land P_{Y}(f(x)) \land g_f \implies P(M) \tag{2}
    \end{align*}
    {From $\Inv{C}{\langle \pl{x},\pl{M}\rangle}{\phi}$ we have $C(\pl{x}|_\conf, \pl{M}|_\conf)$. We then reduce (1) to (2) by fixing $x= \pl{x}|_\conf$, $M = \pl{M}|_\conf$, and observing that from the correctness of the lifting of $f^\uparrow$, we have}
    \begin{align*}
        P_{Y}(f^\uparrow(\pl{x})|_\conf) = P_{Y}(f(\pl{x}|_\conf))
    \end{align*}
\end{proof}

\newText[analysisFormat1]{In the remainder of this section, we illustrate the lifting of analytic templates with respect to the model checking template (Example~\ref{ex:mc_template}) and querying template (Example~\ref{ex:queryTemplate}) defined in Sec.~\ref{sec:analyticTemplates}. Before we can do this, we need to specify how to model the variational inputs and outputs of their underlying analyses. Let us begin with their outputs. With ${\tt Result}$ denoting the set of possible results of the model checker, we define a \emph{variational model checking result} $\pl{r} \in \Var{\tt Result}$ to be an (explicit) product line $\{\langle r_i, \phi_i\rangle\}$, where $r_i$ is the result (either $ok$ or a counterexample) corresponding to configurations $\conf$ satisfying $\phi_i$. Conversely, we define a \emph{variational query result} $ \pl{S} \in \Var{\mathsf{QR}}$ {to be a set of pairs $\{\langle e_i, \phi_i\rangle \}$, with each $\phi_i$ describing the set of configurations in which element $e_i$ is present and satisfies the query}.}

\newText[analysisFormat2]{The choice of how to represent variational inputs to these analyses is more subtle. Taking model checking as an example, a ``natural'' choice may be to take the input as a \emph{product line of model checking problems} $\Var{(\token{LTS} \times \token{LTL})}$, i.e., an input is a pair $\langle \pl{M}, \pl{Spec} \rangle$, where $\pl{M}$ is a product line of LTS models (e.g., an FTS) and $\pl{Spec} = \{\langle \psi_i, \phi_i \rangle \}$, where $\psi_i$ is the LTL specification used for verification of configurations satisfying $\phi_i$. Indeed, there are some lifted verification tools which support analyzing a product line of models against variability-aware  specifications~\cite{classen2012model}. However, there is also a common design pattern in SPLE in which, when lifting an analysis whose input consists of both a model and specification (such as in model checking and querying), the lifted analysis \emph{only} analyzes the product line of models against a {single} (product-based) specification  ~\cite{MURPHY2025112280}. For instance, the lifted model query engine developed by Di Sandro et al.~\cite{di2023adding} follows this pattern. Our framework for lifted analytic templates can accommodate both kinds of lifting. However, for simplicity, our illustrations will assume that we are using the ``simpler'' form of lifting, in which we are analyzing a product line of models against a single specification. In particular, we will assume that for lifted model checking, we are verifying a product line of models against a single LTL specification (equivalently, a product line of specifications containing a single product). Likewise, for lifted querying, we will assume we are analyzing a product line of models against a single query (equivalently, a product line of queries containing a single product). }
\begin{example}\label{ex:mclift}
\begin{figure}[t]
    \centering
    \includegraphics[width=\linewidth]{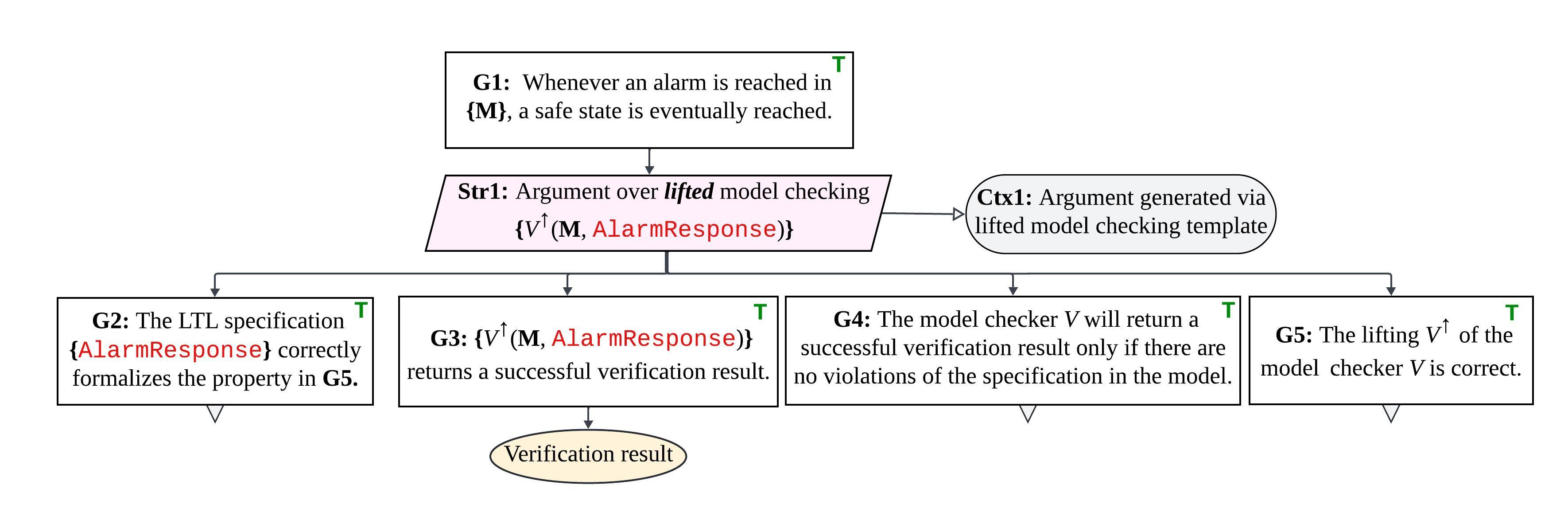}
    \caption{\bf Instantiation of the lifted model checking template on FTS $\pl{M}$ and LTL specification ${\ntn{\tt AlarmResponse}}$.}
    \label{fig:liftInst}
    \Description[]{}
\end{figure}
% \lm{REVISE text for new figure}
\nc{Recall the analytic template for model checking of LTSs formalized in Example~\ref{ex:mc_template}. Given a lifted LTL model \typofix{$V^\uparrow$ checker}{checker $V^\uparrow$} for FTSs (e.g., SNIP ~\cite{classen2012model}), we can lift this template and instantiate it as shown in Fig.~\ref{fig:liftInst}. In this example, the \change[instaClar]{instantiation is done with presence condition $\top$, so this argument pertains to all products of the FTS $\pl{M}$.}{parent goal being decomposed is annotated with $\top$, meaning that \textbf{G1} must be supported for all configurations. Suppose that, when instantiating the template, all configurations are verified by $V^\uparrow$. We are then able to instantiate the template across all configurations, hence all children of \textbf{Str1} are also annotated by $\top$. If some violations had been found, we could instantiate the template only for those configurations which were verified.}  Note \weakDelete{that goals $g_1$ and $g_{out}$ are abstracted over configurations of $\pl{M}$, and note} the new subgoal ${\ntn{\textbf{G5}}}$ asserting the correctness of the lift. We emphasize that  \weakChange{\typofix{$g_3$}{$g_{out}$}}{\textbf{G3}} is interpreted with respect to the \emph{lifted} model checking result. \weakChange{That is, when we instantiate the template, we have already verified all products of $\pl{M}$, and}{Consequentially,} deriving a product AC does not trigger additional product-based verification.} 
\end{example}

\subsection{Analytic Variational Argumentation}

Beyond their usefulness as a means of collecting variational evidence, there is another compelling reason to integrate lifted analyses into assurance arguments. Since lifted analyses -- by definition -- produce results which reflect variability in the product line, these results can be used as a basis for introducing variability in the assurance case. 

To illustrate this, let us \weakChange{begin by lifting}{lift} the enumeration template defined in Example~\ref{ex:enumDecomp}. We first define the notion of a \emph{variational singleton}, i.e., a singleton set $\{x\}$ annotated by presence condition $\phi$, denoted $\{x\}_\phi$. Given a configuration $\conf$, we have $\{x\}_\phi|_\conf = x$ if $\conf \vDash \phi$, and \change[undef]{$\{x\}_\phi|_\conf = \emptyset$ otherwise}{$\{x\}_\phi|_\conf = *$ otherwise, where $*$ represents an ``undefined'' value.}

\begin{example}[Lifted Enumeration Template]\label{ex:enumDecompLift}
    Fix a universe $U$ and a predicate $P$ over $U$. Given the enumeration template $\template{\mathbb{P}(U)}{\mathbb{P}(U)} = \langle \mathsf{Forall}_P, I, C\rangle$, the \emph{lifted} enumeration template is defined by the instantiation function $\ntn{I^\uparrow : \Var{\mathbb{P}(U)} \to \Var{\mathbb{P}(\Goal)}}$ (defined for finite inputs) as 
    $$\ntn{I^\uparrow(\{\langle x_1,\phi_1\rangle \ellipses  \langle x_n,\phi_n\rangle \}) = \{ \langle \{x_1\}_{\phi_1} , P, \phi \land \phi_1\rangle\ellipses  \langle \{x_n\}_{\phi_n}, P, \phi \land \phi_n\rangle\}}$$
\ntn{where $\phi$ is the presence condition of the parent goal being decomposed by the template.}
\end{example}

That is, $I^\uparrow$ simply maps each annotated element $\langle x_i, \phi_i\rangle$ to a variational subgoal which, upon derivation \weakNew{under configuration $\conf$}, asserts $P(x_i)$ only when $\phi_i$ is \change[junkexpl]{true}{satisfied by $\conf$. By construction, the ``undefined'' value $*$ will never be observed in any configuration; it merely serves to make derivation a total function.}

\begin{proposition}\label{prop:liftEnum}
    The instantiation function for the lifted enumeration template (Def.~\ref{ex:enumDecompLift}) is a correct lift of the instantiation function in Example~\ref{ex:enumDecomp}.
\end{proposition}
{\ntn{Finally, with $X =  {\tt LTS} ~\times~ {\tt Query}$ (the set of inputs to the query engine) and $Y = {\mathsf{QR}}$ (the set of possible query results), define the  \emph{lifted querying template} ${\ntn{\template{\token{LTS}}{X}^\uparrow}}$ induced by a lifted query engine $Q^\uparrow: \Var{X} \to \Var{Y}$, which is a lift of the product-based {\ntn{querying template $\template{\token{LTS}}{X}$}   (Example~\ref{ex:queryTemplate})}}}

\begin{example}\label{ex:analyticVar}          
 \begin{figure}[t]
     \centering
     \includegraphics[width=\linewidth]{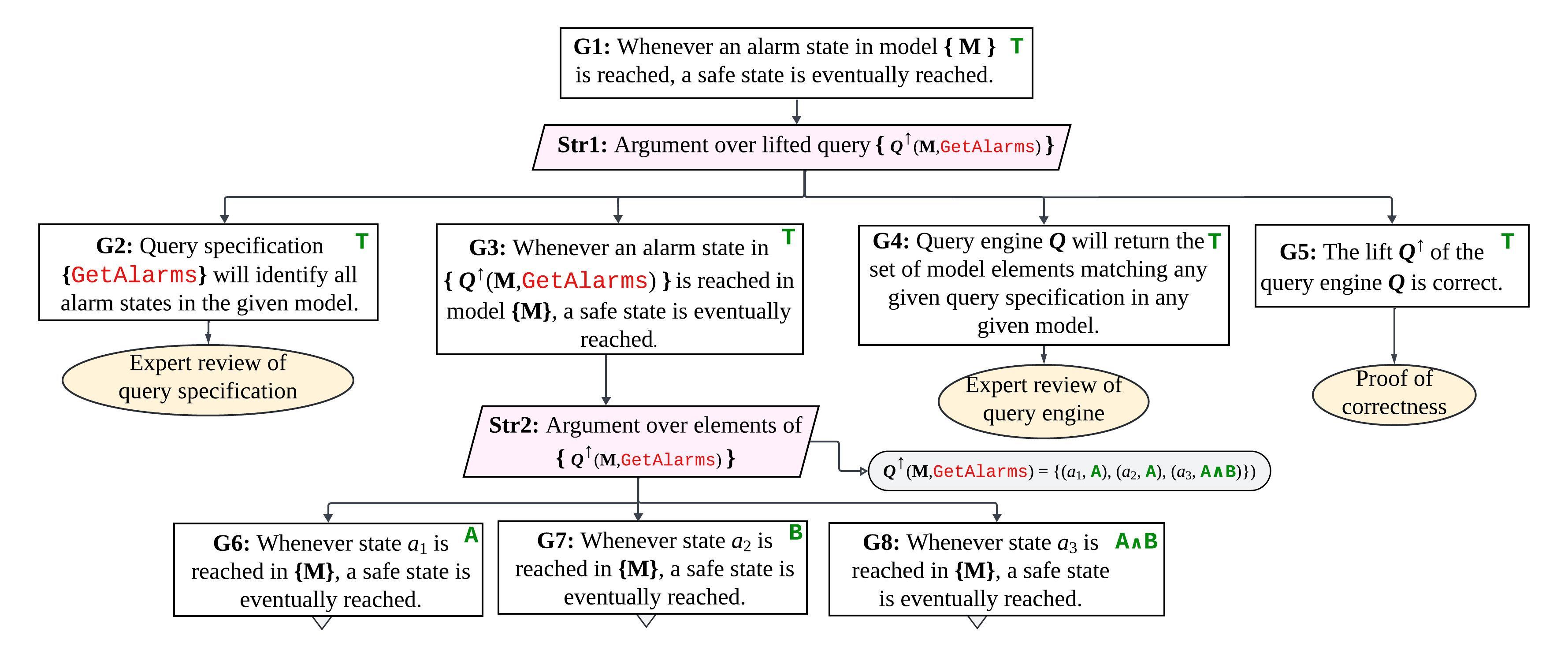}
     \caption{\bf Illustration of analytic variational argumentation achieved through a lifted query template and lifted enumeration template (Example~\ref{ex:analyticVar}). }
     \label{fig:lift_query_example}
     \Description[]{}
 \end{figure}  
\nc{The variational AC fragment in Fig.~\ref{fig:lift_query_example} shows how the lifted query template \weakDelete{$T_{Q^\uparrow}$} and lifted enumeration template \weakDelete{$T_{E^\uparrow}$} can be composed to systematically \weakChange{produce}{perform}  variational argumentation. We first apply the lifted query template to decompose root goal $\weakNew{\textbf{G1}}$, which involves running the lifted query on model $\pl{M}$ to produce a variational set $Q^\uparrow(\pl{M},{\tt GetAlarms})$. This variational set is then decomposed through the lifted enumeration template, such that the presence condition of each subgoal is taken from the annotated query result.}
\end{example}

\nc{What is noteworthy about this example is that the\change[auto]{entire process was completely automatic}{assurance-relevant variations between sets of products were identified automatically}, requiring no \weakDelete{design or} manual inspection of the product line\weakDelete{ by a human engineer}. Without lifted analyses, \weakChange{a human}{one} would need to manually inspect the product line model to determine configurations under which each alarm is relevant, and then manually annotate the AC accordingly, which is laborious and error-prone for complex models. Instead, \emph{analytic variational argumentation} allows us to automatically identify assurance-relevant variation points in product line models, and feed this information into lifted templates to produce variational assurance.}

\nc{In this section, we lifted the framework for formal AC development given in Sec.~\ref{sec:product_based_acs} to the product line level. Our formal language of variational ACs provides a suitable setting in which to study the lifting of AC templates. The soundness of the lifting \weakChange{process}{of arbitrary templates} is given by Thm.~\ref{thm:liftTemplateCorrect}. We also studied the lifting of analytic templates, and proposed analytic variational argumentation in which lifted analyses are used to identify assurance-relevant variation points to be reflected in the AC, while preserving argument validity (per Thm.~\ref{thm:liftAnalyticCorrect}).}

\section{A Semantic Approach to Assurance Case Regression}
\label{sec:product_regression}

In the preceding sections, we formalized and lifted AC development. The primary motivation for formalized ACs is the ability to verify their correctness. However, another benefit of formalization becomes apparent when we need to consider the effect of a \emph{change} on the AC. When ACs are completely informal or unstructured, determining which aspects of system assurance have been potentially compromised can be highly laborious and error-prone. By contrast, when ACs are developed and interpreted formally, the regression of assurance (or its preservation) \weakChange{becomes a problem which can be rigorously and systematically analyzed}{can be analyzed mechanically}. This is especially crucial in the context of ACs for product lines, where the complexity of analyzing assurance regression manually quickly becomes intractable.

In this section, we study the {regression} of (product) ACs when system models are evolved, which in turn will form the foundation of a lifted regression analysis in Sec.~\ref{sec:liftedRegression}. Proofs of selected theorems not given here are provided in Appendix~\ref{app:proofs:sec5}.

\subsection{Specification of Product-Based Regression Analysis}
\label{sec:regressionOutline}
\newText[abstr3]{In analyzing the regression of an AC, we are effectively analyzing the regression of assurance at three levels of abstraction: at the level of \emph{evidence}, at the level of \emph{arguments}, and at the level of the entire AC. In all cases, however, we use the same general notion of a \emph{regression value}, adapted from ~\cite{kokalySafetyCaseImpact2017}. } 

\begin{definition}[{Regression Value}]\label{def:regression}
     \change{Given some set $X$ and a predicate $P$ over $X$, a \emph{regression analysis} for $P$ is a function ${R_P} : X \to \{{\cmark}, {\xmark}, \textbf{?}\}$. The symbols \cmark, \xmark and \textbf{?} are referred to as \emph{regression values}. Given $x \in X$ (representing some ``old'' object),  and $x^\prime \in X$ (representing some ``new'' object), the functional specification of $R$ is as follows:}{Given a predicate $P$ over a set $X$, let $x \in X$ such that $P(x)$. Let $x' \in X$ be some modification of $x$. A \emph{regression value} for $x'$ with respect to $P(x)$ is one of the following:}
\begin{itemize}
    \item ``Reuse'' (denoted $\cmark$), indicating that $P(x')$ holds.
    \item ``Revise'' (denoted $\xmark$), indicating that $P(x')$ does \emph{not} hold.
    \item ``Recheck'' (denoted ?), indicating that we cannot determine whether $P(x')$ holds.
\end{itemize}
\end{definition}
\newText[regClar]{A regression analysis, broadly speaking, is any analysis that determines a regression value following some modification to some object (e.g., a program, a model, or an assurance case). Given a predicate $P$ over $X$, we may speak formally of a \emph{regression analysis} as a function $R_P : X \times X \to \{\cmark, \xmark, ?\}$ such that, given any ``old'' $x \in X$ satisfying $P(x)$ and some ``new' $x' \in X$, $R_P(x,x')$ computes a regression value for $x'$ with respect to $P(x)$. }
Examples of regression analyses include program equivalence checking (e.g., for regression of testing results~\cite{mora2018client}) and structural model analyses (e.g., for regression of model checking results~\cite{menghi2021torpedo}). Obviously, such analyses can be used to determine whether individual pieces of evidence can be reused through a system evolution. But we can also consider regression analysis of the assurance case itself, wherein the predicate we are analyzing is its \emph{support} (Def.~\ref{def:supp}). 

The procedure defined in this section serves exactly this purpose. It can be decoupled into two ``passes'' over \weakChange{$A \in \AC$}{an assurance case $A$} following an evolution of system models. The \emph{forward pass} proceeds in a \emph{top-down} fashion, such that goals referring to out-of-date models are updated, and strategies are checked for argument regression -- each strategy being annotated by a regression value depending on whether the argument remains sound. Following the forward pass, we obtain an ``updated'' AC $A'$. The \emph{backward pass} then aims to determine the support of $A^\prime$ given the support of $A$. This is done in a \emph{bottom-up} fashion; we begin by determining whether pieces of evidence can be reused to support the updated goals, and propagate these results up through the AC, annotating each internal goal with a regression value as we ascend. \weakNew{To avoid wasting resources, we \emph{only} analyze children of strategies if we can determine the regression of the strategy itself. If we cannot determine the regression of a strategy (i.e., is annotated as ``Recheck''), we immediately annotate all its descendants with ``Recheck''. This is done to indicate that the strategy itself must be inspected manually to determine which of its children are reusable following the evolution, and only then is it worthwhile to analyze their regression. }

\begin{example}
\begin{figure}
    \centering
    \includegraphics[width=0.6\linewidth]{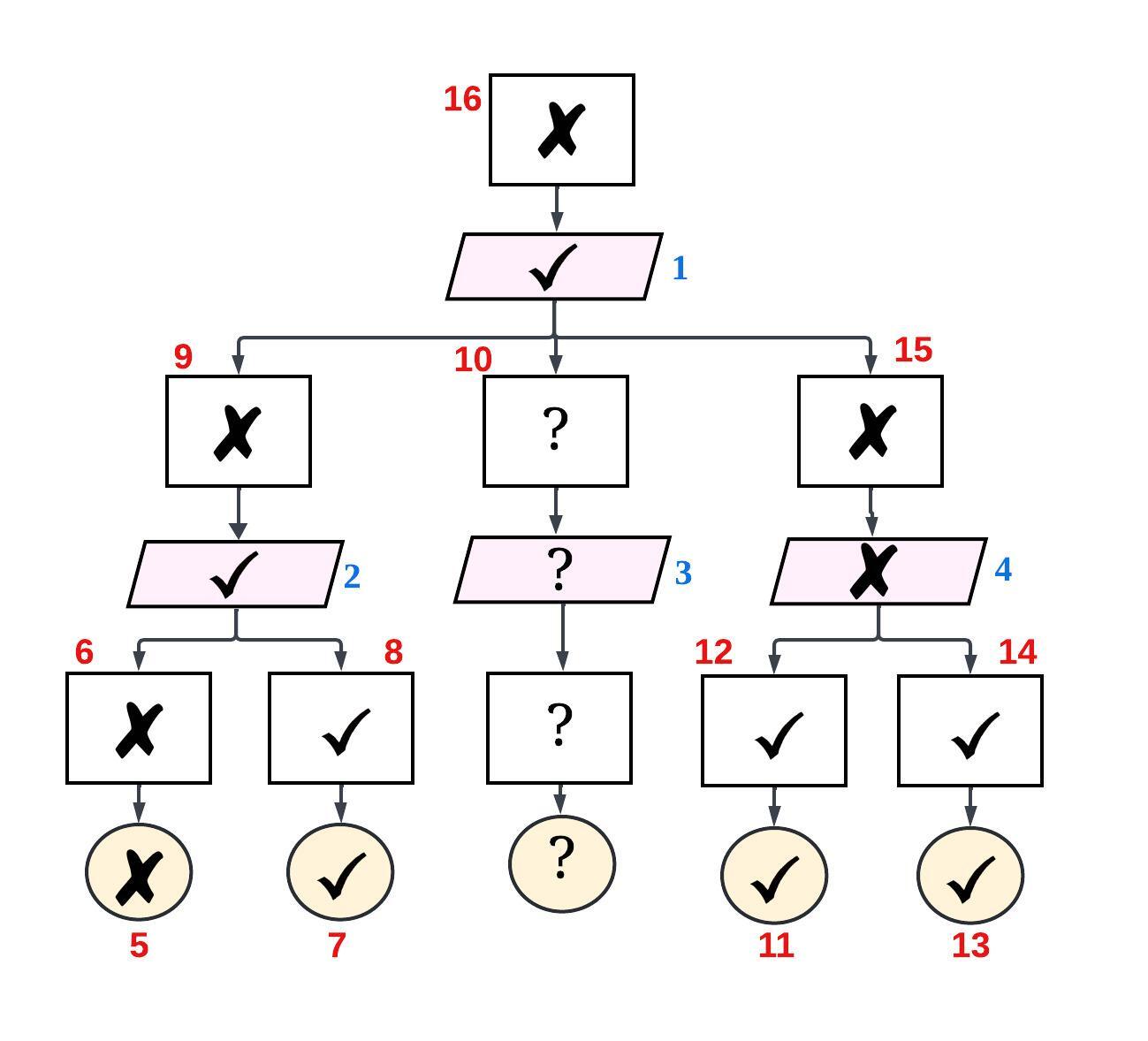}
    \caption{\bf An annotated AC produced through regression analysis, indicating the visitation order. Blue labels (1-4) indicate annotations made during the forward pass, red labels (5-16) indicate annotations made during the backward pass. Once the strategy labeled by (3) is annotated with ``Recheck'' (?) during the forward pass, all its descendants immediately inherit this regression value.} 
    \label{fig:product_regression_example}
    \Description[]{}
\end{figure}

Fig.~\ref{fig:product_regression_example} illustrates the potential annotations for an AC following an evolution, and their sequencing. Each goal, evidence artifact, and strategy is annotated by a regression value. The annotation of an evidence artifact indicates whether this artifact \typofix{}{is} still adequate for the goal to which it is assigned. The annotation of a strategy indicates whether the argument induced by this strategy remains sound. The annotation of a goal indicates whether the AC rooted at this goal is supported vis-à-vis Def.~\ref{def:supp}. Numbered labels indicate the sequence in which annotations are made, with those made during the forward pass shown in blue \weakNew{(1-4)}, and those made during the backward pass shown in red \weakNew{(5-16)}. \weakNew{The evidence-goal pair in the middle branch is marked as ``Recheck'' as a direct consequence of their parent strategy (labeled 3) being marked as ``Recheck'' during the forward pass. This regression value is then assigned to the parent goal of the strategy during the backward pass.}

\end{example}

\typofix{Let's}{We} fix a set $\mathcal{S} = \{M_1\ellipses  M_n\}$ of system models which are subjects of the AC. We define an \emph{evolution set} as a set $\Delta = \{(M_j,M_j^\prime), ... (M_\ell, M_\ell^\prime)\}$, where $1 \leq j,\ell \leq n$, and each $M_k^\prime$ is the ``evolved'' version of model $M_k$. Note that \change[subsetChange]{only a subset of system models may be modified in an evolution}{a given evolution may modify every system model appearing in the assurance case, or it may only modify a subset of these models}. Given a system model $M$, we use $M \in \Delta$ as shorthand for $\exists~M^\prime \in \mathcal{M}, (M,M^\prime) \in \Delta$.

\subsection{The Forward Pass: Regression of Arguments}
\label{sec:regression:forward}
As described above, the forward pass of the regression analysis is \weakDelete{effectively} tasked with updating and potentially revising the arguments of the AC. \weakChange{Proceeding recursively down the AC, we need to}{We} consider the regression of each strategy according to the means by which it was created: \weakChange{it may be an instantiation of a general template (Sec.~\ref{sec:acTemplates}), an analytic template (Sec.~\ref{sec:analyticTemplates}), or it may not be template-based at all}{in particular, whether it was created through the use of a formal template (in which case its regression can be formally analyzed), or whether it was created by hand (in which case we assume its regression will be analyzed manually)}. This is also where the distinction between argument-producing and evidence-producing analytic templates becomes important. \newText[newLazy1]{Recall from Sec.~\ref{sec:analyticTemplates} that an application of an analytic template $\template{\modlang}{X}$ for analysis $f: X \to Y$ is \emph{evidence}-producing if the subgoal $g_Y$ corresponding to the output of $f$ can be supported directly by using $f(x)$ as evidence (e.g., $f(x)$ is a model checking result and $g_Y$ asserts that the verification was successful). Conversely, an application of $\template{\modlang}{X}$ is \emph{argument-producing} if $g_Y$ is the parent goal of a decomposition, i.e., $g_Y$ is supported by further argumentation rather than by evidence. Note that we can automatically distinguish the two kinds of applications (i.e., the resulting strategies) simply based on the structure of the AC, and that a given analytic template may in principle have applications which are both evidence-producing and argument-producing.} 
{When we complete the forward pass through the AC, if we encounter an} \change[newLazy2]{analytic template $T_f$ which has been applied to model $M \in \Delta$, we need to ask whether it makes sense to re-apply $T_f$ on $M^\prime$, which in turn requires evaluating $f(M^\prime)$.}{application of an analytic template $\template{\modlang}{X}$ to some input $x \in X$, we need to decide whether or not to re-instantiate the template, which would involve re-executing the analysis $f$. We can adopt an \emph{eager} strategy, in which all analytic templates are re-instantiated, or a \emph{lazy} strategy, in which some analytic templates are \emph{not} re-instantiated.} 

\change[newLazy3]{In our regression analysis, we  adopt the convention that argument-producing  templates should be re-instantiated \emph{eagerly} during the forward pass, while evidence-producing templates should \emph{not}; instead, we use the information generated during the backward pass to decide whether the evidence-producing analyses {need} to be re-executed. }{In this work, we adopt a lazy strategy. In particular, we propose that if an application of an analytic template is {argument}-producing, it should be eagerly re-instantiated, whereas an application which is {evidence}-producing should not be re-instantiated. We justify our strategy with two observations. First, note that any nodes which descend from an argument-producing analytic strategy may depend on the results of the analysis. For example, considering the argumnet over querying shown in Example~\ref{ex:fullRE}, each of the subgoals corresponds to an alarm state identified by the query. If we do not re-execute the analysis by re-instantiating the template, we cannot determine whether the strategy enumerating over all the alarm states has regressed or not, since we do not know if there are any new states in the evolved model. In this circumstance, it makes sense to re-instantiate the template in order to analyze the completeness of the subsequent argument.}

\weakNew{However, in the case of evidence-producing instantiations of templates, the only AC nodes dependent on the analysis result are the output goal $g_Y$ and the evidence supporting it. To be sure, in some cases the only viable option may be to re-execute the analysis as a means of producing up-to-date evidence. In other cases, however, we may be able to determine whether the evidence can be reused without re-executing the analysis, and this may be less resource-intensive. For instance, referring once again to Example~\ref{ex:fullRE}, we can consider the (evidence-producing) application of the model checking template. Rather than re-running the model checker on the evolved model, we may wish apply a  structural regression analysis comparing the previous and modified versions of the model. For example, the structural regression analysis for model checking developed by Menghi et al.~\cite{menghi2021torpedo} has complexity which is linear in the size of the model, whereas LTL model checking is linear in the size of the model but exponential in the size of the LTL formula being verified~\cite{baier2008principles}.}

\weakNew{With this lazy evaluation strategy in mind, we begin by considering the general procedure for analyzing the regression of an argument built from a template (which may or may not be analytic)}.

\subsubsection{\ntn{Regression Analysis of Template-Based Arguments}} 

\label{sec:tempreg}
\begin{algorithm}[t]
\onehalfspacing
\caption{\ntn{Regression Analysis for Template-Based Arguments}}
\label{alg:templateRegression}
    \begin{algorithmic}[1]
            \Procedure{$\mathsf{TemplateRegression}$}{$M^\prime,st,\mathcal{A},\template{\modlang}{D},x$}
        \State \textbf{fix} $\template{\modlang}{D} = \langle P,I,C\rangle$
        % \If{$C(x,M^\prime)$}~ $\{g_1^\prime\ellipses g_n^\prime\} = I(x)$ 
        \If{some $x^\prime$ is available such that $C(x^\prime,M^\prime)$}
        \If{$st$ is analytic and evidence-producing}{~$\{g_1^\prime\ellipses g_n^\prime\} := \mathsf{Replace}(\mathcal{A},x,x')$}
        \Else $\;\{g_1^\prime\ellipses g_n^\prime\}  := {I}(x^\prime) $
        \EndIf
        \Else 
        \State \textbf{for} $A_i \in \mathcal{A}$ \textbf{do} $\mathsf{MarkAsRecheck}(A_i)$
            \State \Return $ ?$
        \EndIf
        \State $\langle\mathcal{A}_{\tt Obs},\mathcal{A}_{\tt Reuse}, \mathcal{A}_{\tt New}\rangle := {\mathsf{Match}}(\mathcal{A}, \{g_1^\prime\ellipses  g_n^\prime\})$
        \State \textbf{for} $A_i \in \mathcal{A}_{\tt Obs}$ \textbf{do} $\mathsf{MarkAsObsolete}(A_i)$
        \State $v_{st} := \IfThenElse{\mathcal{A}_{\tt New} = \emptyset}{\cmark}{\xmark}$
        \State \Return $\ v_{st}$
        \EndProcedure
    \end{algorithmic}
\end{algorithm}

$\;$
\weakChange{We fix a decomposition strategy ${\tt Decomp}(g,st,\mathcal{A})$, where $g = \langle M,P\rangle$ is a goal referring to model $M$ via predicate $P$, where $st$ is a strategy obtained by instantiating some template $\ntn{\template_{\modlang}{D}} = \langle P,I,C\rangle$ with input $x$, where $I$ is the template's instantiation function, and where $C$ is its correctness criterion. Suppose that $M \in \Delta$, and let $g^\prime$ be the updated goal over $M^\prime$.}{Consider a goal $\langle M, P\rangle$ which has been decomposed using strategy $st$ generated from template $\template{\modlang}{D} = \langle P, I, C\rangle$, by instantiation with $x \in D$.}
During the forward pass, we need to determine two things about this strategy: its regression value, and (if it needs to be revised), which of its subgoals are reused in the new argument. The latter is necessary to avoid analyzing the regression of \weakDelete{AC fragments underneath} goals which will\typofix{}{no} longer be relevant after argument revision.

Alg.~\ref{alg:templateRegression} defines a procedure which takes as input the evolved model $ M^\prime$, \weakNew{the strategy $st$}, the ``old'' child ACs $\mathcal{A}$, the template \weakChange{$T = \langle P,I,C \rangle$}{$\template{\modlang}{D}$}, and the \weakChange{input $x$ which was used to instantiate the template, and it}{instantiation input $x$. The algorithm} returns a regression value $v_{st}$ for the strategy. The first step is \weakChange{to check whether the input $x$ remains correct relative to $M^\prime$ given the correctness criterion. If so, we produce the updated subgoals $\{g_1^\prime\ellipses g_n^\prime\}$ for the template, which refer to $M^\prime$ rather than $M$. Otherwise, we can attempt to produce some new input $x^\prime$ which satisfies the correctness criterion, and use it for the instantiation (Line 4).}{to try to identify a new instantiation input $x' \in D$ which satisfies the correctness criterion $C$ of the template given the evolved model $M^\prime$.} We do not specify how such an $x^\prime$ should be synthesized \weakNew{in general} (though this introduces an interesting \emph{repair} problem for future work, \typofix{}{cf. Sec.~\ref{sec:conclusion}}). {If no such $x^\prime$ can be obtained, we have no way of determining what parts of the argument are reusable. As such, we \weakChange{annotate all descendants as \textbf{?} }{mark all descendants of the strategy as ``Recheck''} (Line {\ntn{7}}) before returning \weakChange{$v_{st} = \xmark$}{$v_{st} = ~?$} (Line 8).}{}

\weakChange{After either of the first two cases, we are left with a set of updated subgoals $\{g^\prime_1\ellipses  g_n^\prime\}$.}{If we are able to find a suitable input $x'$, we next need to consider whether to re-instantiate the template. As explained above, we adopt a lazy evaluation strategy in which applications of analytic templates which are evidence-producing are not re-instantiated. In such a case,  we only \emph{replace} each occurrence of $x$ with $x'$ in the subgoals of the argument (i.e., the roots of each $A \in \mathcal{A}$), which we do with an operator $\mathsf{Replace}$ (Line 4), producing a set of updated goals $\{g_1'\ellipses g_n'\}$. Otherwise, if the strategy was developed using a non-analytic template, or the strategy is analytic and argument-producing, we obtain the updated goals $\{g_1'\ellipses g_n'\}$ by re-instantiating the template with $x'$ (Line 5).} 

We now need to compare \weakChange{these}{the updated goals $\{g_1'\ellipses g_n'\}$} against the existing \weakChange{sub-ACs $\mathcal{A}$ to determine which sub-ACs}{subgoals given by the roots of each $A \in \mathcal{A}$  to determine which parts of the AC} ought to be further analyzed. We do this using an operator $\mathsf{Match}$ (Line {\ntn{9}}), which matches updated subgoals $g_i^\prime$ against the root goals of $\mathcal{A}$. \weakChange{When a match is found (meaning they refer to the same pair model (modulo its evolution) via the same predicate), we replace the root goal of the sub-AC with the updated goal, and add it to a set $\mathcal{A}_{\tt Reuse}$}{We consider a goal $g'_i$ to be \emph{matched} by some $A \in \mathcal{A}$ if the root goal of $A$ is either (1) identical to $g_i'$, (2) equivalent up to substitution with an evolved system model, or (3) equivalent up to substitution with an updated analysis result (in the case of an analytic template). Those children $A \in \mathcal{A}$ which can be matched to some $g_i'$ are added to the set $\mathcal{A}_{\tt Reuse}$, since they are reused in the updated AC.} Any \weakChange{sub-ACs}{children} in $\mathcal{A}$ which are not matched to any $g^\prime_i$ represent \emph{obsolete} children of the strategy (recorded in $\mathcal{A}_{\tt Obs}$), and are not analyzed further. Finally, any subgoals $g_i^\prime$ not matched by any existing \weakChange{subtree}{child $A \in \mathcal{A}$} are recorded as \emph{new} subgoals ($\mathcal{A}_{\tt New}$).\footnote{Given that we have identified $\mathcal{A}_{\tt New}$, it would be natural to include them as undeveloped subgoals in the AC right away, in which case we could return $v_{st} = \cmark$. For the sake of simplicity, we do not do any such ``argument repair'' here, and leave the structure of the AC unchanged.} Any obsolete \weakChange{sub-ACs}{children} are then annotated as such (Line 10). Finally, to determine a regression value for the strategy, we check whether $\mathcal{A}_{\tt New}$ is empty, \change[curPremMis]{(meaning that the current argument is missing some premises)}{meaning that re-instantiating the template (if it was done)  did not produce any \emph{new} subgoals}. \weakNew{If it is empty, we return  ``Reuse'' ($v_{st} = \cmark$); otherwise, we return ``Revise'' ($v_{st} = \xmark$).}

\begin{example} 
    Recall the annotated AC shown in Fig.~\ref{fig:product_regression_example}. Note the rightmost strategy, which is annotated by $\xmark$. Since the children of this strategy have apparently been verified against regression, it \emph{must} be the case that (1) we were able to produce a correct input to re-instantiate the template (otherwise {\ntn{the strategy and}} both subgoals would be annotated \weakChange{by}{with} $\textbf{?}$), (2) both subgoals were found to be reusable during the forward pass (otherwise they would have been marked obsolete), and (3) there is at least one missing subgoal (otherwise the strategy would be annotated by $\cmark$).
\end{example}

\begin{lemma}\label{lem:templateRegCorrect}
% \lm{REVISE}
    \weakNew{Let ${\mathtt{Decomp}}(g,st,\mathcal{A})$ be a decomposition of parent goal $g = \langle M,P\rangle$ into children $\mathcal{A}$, and let $M' \in \modlang$ be the evolved version of model $M \in \modlang$. Let $\template{\modlang}{D} = \langle P,I,C\rangle$ be the (valid) template used to produce this decomposition via instantiation with $x \in D$. Let $v_{st}$ be the regression value returned by $\mathsf{TemplateRegression}(M',st,\mathcal{A},\template{\modlang}{D},x)$. Then:}
    \begin{enumerate}
        \item \weakNew{$v_{st} = \cmark$ iff the decomposition remains sound after updating $g$ and $\mathcal{A}$ with evolved system model(s), and (if applicable) re-instantiating $\template{\modlang}{D}$ for $M^\prime$,
        \item $v_{st} = \xmark$ iff re-instantiating $\template{\modlang}{D}$ for $M^\prime$ produces at least one new subgoal not provided in $\mathcal{A}$ which is required to make the decomposition sound, and 
        \item $v_{st} = ~?$ iff the template could not be re-instantiated for $M'$.}
    \end{enumerate}

\end{lemma}

\subsubsection{The Forward Pass}
\begin{algorithm}[t]
    \caption{Forward Pass of the Regression Analysis}
    \label{alg:instantiation}
    \begin{algorithmic}[1] 
       \onehalfspacing
        \Procedure{$\mathsf{ForwardPass}$}{$A,\Delta$}
                \If{${\tt Rt}(A)$ is predicative with $M \in \Delta$}{~$\mathsf{Update}({\tt Rt}(A),M^\prime)$} \Comment{$(M, M') \in \Delta$}
                \EndIf
                \Switch{$A$}
                \Case{${\tt Und}(g)$ \textbf{or} ${\tt Evd}(g,e)$:}~\Return
                \EndCase 
                \Case{${\tt Decomp}(g, st, \mathcal{A})$:}
                    \If{$st$ instantiates template $\template{\modlang}{D}$ with input $x$}
                    \If{$g = \langle M, P\rangle$ with $M \notin \Delta$}
                          \State $\mathsf{Annotate}(st, \cmark)$ 
                          \State \textbf{for} $A' \in \mathcal{A}$ \textbf{do } $\mathsf{ForwardPass}(A',\Delta)$
                        \Else
                           \State $v_{st}  := \mathsf{TemplateRegression}(M^\prime,st,\mathcal{A},\template{\modlang}{D},x)$
                            \State $\mathsf{Annotate}(st,v_{st})$
                            \If{$v_{st}=$ ?}~\textbf{return}
                            \Else ~\textbf{for} $A' \in \mathcal{A}$ s.t. $A'$ is not marked obsolete \textbf{do} $\mathsf{ForwardPass}(A',\Delta)$ 
                             \EndIf 
                        \EndIf 
                    \ElsIf{$st$ is not template-based}
                        \State $\mathsf{Annotate}(\{st\},\textbf{?})$
                        \State \textbf{for} $A' \in \mathcal{A}$ \textbf{do } $\mathsf{MarkAsRecheck}(A')$
                    \EndIf
                \EndCase
            \EndSwitch
        \EndProcedure
    \end{algorithmic}
\end{algorithm}
The complete forward pass is formalized in Alg.~\ref{alg:instantiation}. We take as input $A$, the AC being updated, and $\Delta$, the evolution set. We first check to see if the root goal $g$ of $A$ refers to an ``old'' model $M$ which has evolved; if so, we replace $M$ by $M^\prime$ in $g$ (Line 2). We then \weakChange{pattern match}{perform case analysis} on $A$ \weakNew{(Line 3)}: if it is an undeveloped goal or goal supported by evidence, we return (Line 4). If $A$ is a decomposition over strategy $st$, we do a \weakNew{further} case analysis on how $st$ was created \weakNew{(Line 6)}. \weakChange{First, if $st$ instantiates a (non-analytic) template, we need to check whether the root goal refers to an old version of the model;}{If $st$ was not created using a template, we are unable to analyze its regression, so we annotate $st$ and its descendants as ``Recheck'' (Lines 16-17). Otherwise, if $st$ is the instantiation of a template, then we need to check whether the subject of the parent goal $g$ has evolved (Line 7);} if not, we can mark the strategy with $\cmark$ and proceed \weakNew{recursively} through the \weakChange{sub-ACs}{children} (Lines 8-9). Otherwise, we need to invoke the $\mathsf{TemplateRegression}$ procedure in Alg.~\ref{alg:templateRegression} \weakNew{(Line 11)}. Having obtained the strategy regression value $v_{st}$, we annotate $st$ with it \weakNew{(Line 12). If we found $v_{st} =~ ?$, then we could not re-instantiate the template, and that the children of this strategy will not be analyzed, so we return early (Line 13). Otherwise, we continue the analysis recursively on all children of the strategy which were not identified as obsolete by Alg.~\ref{alg:templateRegression} (Line 14).}

\weakDelete{If $st$ instead instantiates an analytic template, we can annotate it by $\cmark$ (as noted in Sec.~\ref{sec:analyticTemplates}, we consider all analytic templates to be correct-by-construction). We only need to re-execute the analysis if the model in the parent goal is updated \emph{and} the template is argument-building (Line 17). Either way, we proceed recursively (Line 18). Finally, if $st$ is not template-based, we annotate it and its descendants with \textbf{?} and halt (Lines 20-21).}

\subsubsection{Extracting the Reusable Core}
\newText[reusableCore]{Following the forward pass of the analysis, every strategy of the AC will be annotated by a regression value, and every goal will either not be annotated, annotated as ``Recheck'' (\textbf{?}), or marked as obsolete. At this point, we can immediately prune all subtrees of the AC which are rooted at an obsolete goal; we no longer need to support them in the updated AC, so we do not need to analyze their regression. The remaining goals of the AC are \emph{potentially reusable}, since they are either still needed to support their parent goals, or are descendants of strategies which could not be analyzed. We refer to the AC obtained after pruning obsolete subtrees from $A$ as the \emph{reusable core} of $A$, denoted $A_R$. Note that any goals in $A$ whose children were all deemed obsolete during the forward pass become undeveloped goals in $A_R$. To produce appropriate regression annotations for the remaining goals in $A_R$, we need to determine the regression of the evidence supporting its reusable goals, and propagate these regression values upwards through $A_R$.}

\subsection{The Backward Pass: Regression of Evidence}
% {Regression of Evidence}
\begin{algorithm}[t]
    \caption{Evidence Regression}
    \label{alg:productEvdRegression}
    \begin{algorithmic}[1] 
       \onehalfspacing
        \Procedure{$\mathsf{EvdRegression}$}{$\Delta, g^\prime,e$}
            \If{$g^\prime$ is propositional \textbf{or} $g^\prime$ refers to $M$ such that $M \notin \Delta$}{~$v = \cmark$}
            \Else
            \State \textbf{fix} $g^\prime = \langle M^\prime, P\rangle$ \Comment{$(M,M') \in \Delta$}
            \State $v := \:$\IfThenElse{\text{regression analysis} ${R_P}$\text{ is available}}{${R_P}(M,M^\prime)$}{\textbf{?}}
            \EndIf
            \State {{$\mathsf{Annotate}(\{e,g'\},v)$}}
            \State \Return $v$
        \EndProcedure
    \end{algorithmic}
\end{algorithm}
\weakNew{As described above, the purpose of the backward pass of the analysis is to determine which goals of the reusable core $A_R$ of can still be supported by evidence, and to propagate regression values from the leaves of the AC to the root.}
Before we consider the regression of evidence, we want to distinguish evidence of propositional goals (Rule [$\Supp$-1] in Def.~\ref{def:supp}) from evidence of predicative goals (Rule [$\Supp$-2]). In what follows, we assume that \emph{all} propositional goals do \emph{not} depend on system models (and therefore, their support never regresses). Turning to predicative goals, \typofix{let's}{we} fix $g = \langle M, P\rangle$, and evidence \weakChange{$e : P(M)$}{$e$ which is adequate to conclude $P(M)$}. Given a model evolution of $M$ into $M^\prime$, and $g$ \weakChange{has}{having} been updated as $g^\prime = \langle M^\prime, P\rangle$, we want to determine whether $e$ \weakChange{can support}{is adequate for} $g^\prime$. Taking the perspective of an assurance engineer, we can consider three cases: (i) a regression analysis $\ntn{R_P}$ for $P$ is already available; (ii) a regression analysis $\ntn{R_P}$ for $P$ is not available, and producing new evidence $e^\prime $ for $P(M^\prime)$ is expensive (e.g., it requires analyzing the system in a physical testing environment); (iii) a regression analysis $\ntn{R_P}$ for $P$ is not available, but checking $P(M^\prime)$ from scratch is inexpensive \weakDelete{(e.g., a static analysis can be repeated on a moderately sized model)}. In the backward pass, we do not distinguish between cases (i) and (iii), on the basis that checking $P$ from scratch is a ``worst-case'' form of regression analysis.

Thus, in \weakChange{computing the regression value $v$ between goals $g, g^\prime$, given that $e$ is evidence for $g$,}{order to determine whether $e$ remains adequate for the updated goal $g'$,} we follow the process defined in Alg.~\ref{alg:productEvdRegression}. \weakNew{In addition to $g'$ and $e$, the algorithm takes as input the evolution set $\Delta$.} First, we check whether $g^\prime$ is propositional, or if $g^\prime = g$ (meaning that the model referred to by $g$ was not evolved); if so, we set \weakNew{the regression value} $v = \cmark$ (Line 2). Otherwise, we \weakNew{fix $g' = \langle M', P\rangle$, and $M$ as the original version of the evolved model $M'$  (Line 4). We then} consider whether a regression analysis $\ntn{R_P}$ is available for the predicate in question (including, potentially, simply re-verifying $P$ from scratch). If such an analysis is available, we set $v$ to the regression value it computes; otherwise, we set $v$ to $\textbf{?}$ (Line 5). Before returning $v$, we annotate the evidence $e$ and goal $g^\prime$ with it ({\ntn{Line 6}}).

\paragraph{Composing Regression Values}
Having defined how regression can be determined for evidence, all that remains is to define the {composition} of regression from multiple subgoals in an argument. That is, given \weakChange{sub-ACs}{children} $\{A_1\ellipses A_n\}$ which decompose a goal $g$ via strategy $st$, we want to determine regression for $g$ given regression values for the root goals of each $A_i$ and for $st$. As per the intended semantics of the regression analysis, the regression value for $g$ should be the \emph{least} of \weakChange{the regression values for the subgoals and the strategy}{these regression values}, under the ordering $\xmark < \textbf{?} < \cmark$. That is, given that $\{A_1\ellipses  A_n\}$ have root goals respectively annotated by $\{v_1\ellipses  v_n\}$, and $st$ is annotated by $v_{st}$, we annotate $g$ with $\mathsf{Min}{(\{v_1\ellipses  v_n, v_{st}\})}$. Note that the AC shown in Fig.~\ref{fig:product_regression_example} conforms to this rule.

\subsubsection{The Backward Pass} 
$\;$
\weakChange{Given an AC $A$, and $A^\prime$ being the result of applying the forward pass to $A$, the backward pass returns a regression value for the root goal of ${A}^\prime$. Because the backward pass proceeds \emph{recursively}, it also annotates each internal goal of ${A}^\prime$ with its own regression value. We specify one additional preprocessing step prior to applying the backward pass: we prune all obsolete goals from $A^\prime$, and the corresponding goals from $A$, so that the two ACs remain isomorphic.}{Let $A$ be the original AC prior to the system evolution, and let $A_R$ be the reusable core of $A$ as identified by the forward pass of the analysis. The goal of the backward pass is to determine which goals in $A_R$ can still be supported by evidence. This stage of the analysis is specified in Alg.~\ref{alg:product_regression}, and takes as input the reusable core $A_R$ and the evolution set $\Delta$.}
\begin{algorithm}[t]
    \caption{Backward Pass of the Regression Analysis}
    \label{alg:product_regression}
    \begin{algorithmic}[1] 
       \onehalfspacing
        \Procedure{$\mathsf{BackwardPass}$}{$\Delta,{A_R}$}
            \Switch{$A_R$}
                \Case{${\tt Und}(g')$:}
                \State $\mathsf{Annotate}(g', \xmark)$
                \State \Return $\xmark$
                \EndCase
                \Case{${\tt Evd}(g^\prime,e)$:} 
                % \State \textbf{fix} $A = {\tt Evd}(g,e)$
                \State \Return ${\tt EvdRegression}(\Delta,g^\prime,e)$
                  % \EndIf
                \EndCase
                \Case{${\tt Decomp}(g^\prime, st, \{A_1^\prime\ellipses  A_n^\prime\})$:} \label{alg:product_regression:recursive}
                    \State $v_{st} := \mathsf{Annotation}(st)$
                    \If{$v_{st}$ = ?}
                    \State $\mathsf{Annotate}(g,v_{st})$
                    \State \textbf{return} $v_{st}$ \EndIf
                    \State $V := \{v_{st}\}$
                    \State \textbf{for each}  $A_i^\prime $ \textbf{do} $V := V \cup \{\mathsf{BackwardPass}(\Delta,A_i')\}$
                     \State $v := \mathsf{Min}(V)$ \label{alg:product_regression:min}
                    \State $\mathsf{Annotate}(g,v)$
                    \State \Return $v$
                \EndCase
            \EndSwitch
        \EndProcedure
    \end{algorithmic}
\end{algorithm}
\change[bpchange]{The backward pass is specified in Alg.~\ref{alg:product_regression}. The analysis proceeds recursively through $A'$, deconstructing $A$ accordingly when pattern matching (Lines 4 and 7).}{The analysis proceeds recursively through $A_R$; at each level of the recursion, we annotate some node(s) with regression values, and return the regression value of the root of the current subtree. We denote each of the goals in $A_R$ using the notation $g^\prime$ to emphasize that the subjects of these goals may have been updated during the forward pass.}  \weakNew{Recall that, assuming $A$ originally had no undeveloped goals, each undeveloped goal $g'$ in $A_R$ must be the parent of a strategy all of whose subgoals were deemed obsolete. In such a case, we annotate $g'$ as ``Revise'' ($\xmark$), since it is no longer supported by any subgoals, and return this regression value (Lines 4-5).} In the \weakChange{base case}{case of a goal supported by evidence ${\tt Evd}(g',e)$}, we invoke the evidence regression analysis (Alg.~\ref{alg:productEvdRegression}). In the recursive case, \weakNew{corresponding to a decomposition strategy, we first extract the regression value $v_{st}$ of the strategy $st$ and check whether it is ``Recheck'' (?); if so, we have reached a strategy which could not be analyzed, and all its descendants have already been marked as ``Recheck''. We thus mark $g$ as ``Recheck'' and return this same value (Lines 11-12).} \weakChange{we first check to see whether we have reached a strategy which is annotated as $\textbf{?}$. If so, we know that this is an informal strategy, and all its descendants have already been annotated by $\textbf{?}$, so we just return $\textbf{?}$ (Line 9). If we see that $v_{st} =\xmark$, we can check to see whether the descendants have already been annotated as \textbf{?} during template regression (Alg.~\ref{alg:instantiation}), in which case we return $\xmark$ (Line 10). Otherwise, we visit each of the subgoals}{Otherwise, we add $v_{st}$ to a set $V$ and recursively apply the analysis to each of the children of the strategy, adding the resulting regression values to $V$ (Lines 13-14)}. \weakChange{and annotate $g$ with the least of these and $v_{st}$ (Lines 12-14).}{Finally, we compute $v$ as the smallest element of $V$ under the ordering $\xmark < \textbf{?} < \cmark$, annotate $g$ with $v$, and return $v$ (Lines 15-17).}

\begin{example}
    \begin{figure}[t]
        \centering
        \includegraphics[width=\linewidth]{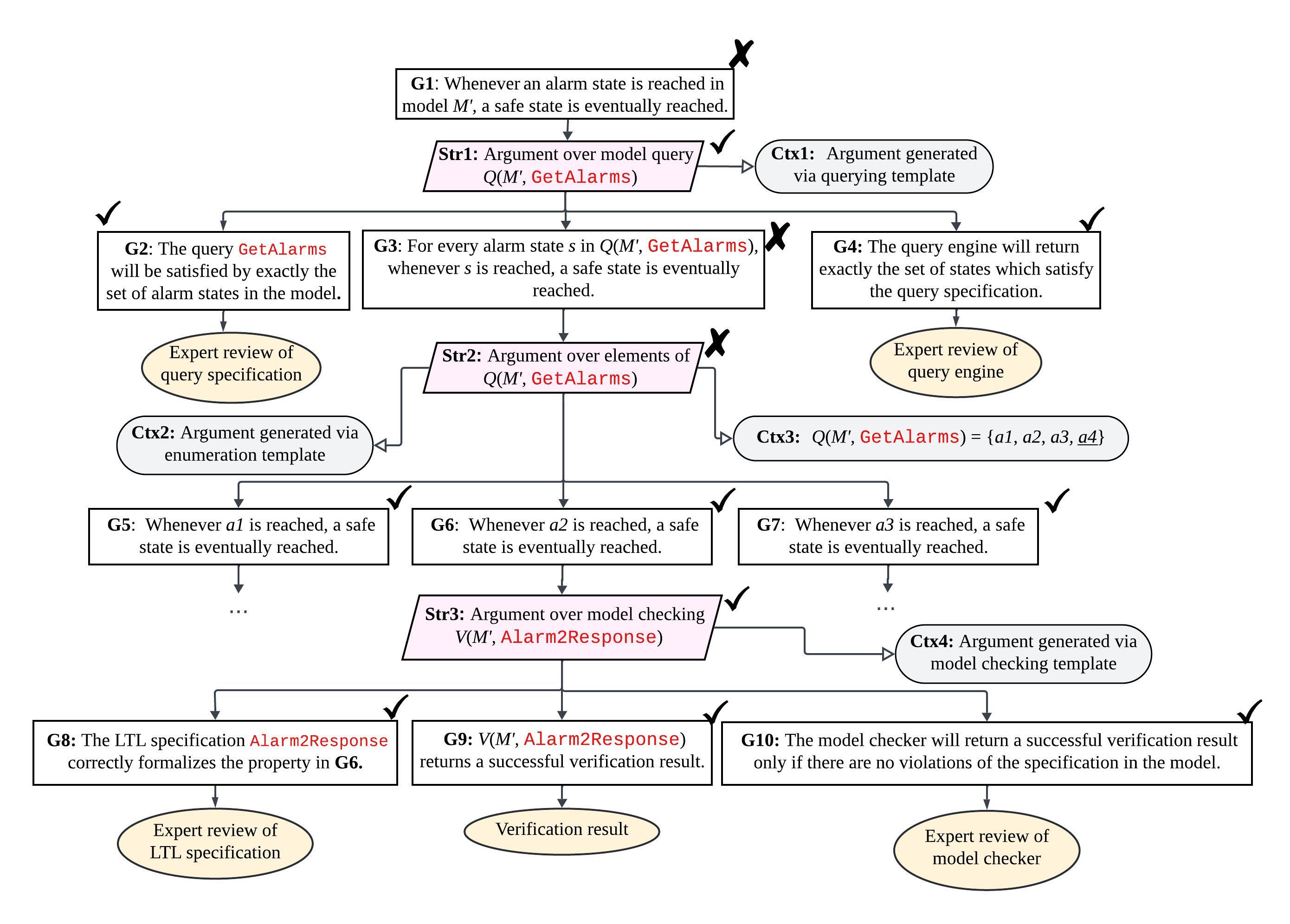}
        \caption{\bf Outcome of regression analysis for the running example AC. }
        \label{fig:regression_example_full}
        \Description[]{}
    \end{figure}

    Recall the product-based AC for the LTS $M$ shown in Fig.~\ref{fig:running_example_full_1}. Suppose that\weakNew{, after having produced evidence for all the undeveloped subgoals,} we evolve $M$ as $M^\prime$ in such a way that a new alarm state $a_4$ is introduced. The full regression \weakNew{analysis} for this evolution would produce the annotated AC shown in Fig.~\ref{fig:regression_example_full}. First, during the forward pass, all references to $M$ in goals are replaced by $M^\prime$. The strategy \weakNew{\textbf{Str1}} is an instance of an analytic template, \weakChange{and so it is labeled by $\cmark$. The strategy $st_1$}{and it} is also argument-producing, so the query engine is executed again on $M^\prime$, \weakNew{satisfying the correctness criterion (i.e., running the query on the same model referred to by the updated parent goal). After the re-instantiation, no additional subgoals are missing, so we mark \textbf{Str1} as ``Reuse'' ($\cmark$). The updated query result contains the new alarm state $a_4$}. When we proceed recursively into \weakNew{\textbf{Str2}}, we re-apply the enumerative instantiation template, and see that we now require four subgoals, whereas we only have three, so we annotate \weakNew{\textbf{Str2}} as $\xmark$. However, none of these three subgoals are obsolete, so we continue through all of them (here, we show only the process for \weakNew{\textbf{G6}} and its descendants). Once again reaching an instance of an analytic template in \weakNew{\textbf{Str3}}, we annotate it by $\cmark$. This strategy is evidence-producing, so we do not re-execute the model checker on $M^\prime$.

    We now move on to the backward pass. Beginning with the evidence-supported subgoals of $\weakNew{\textbf{Str1}}$, we see that neither $\weakNew{\textbf{G2}}$ nor $\weakNew{\textbf{G4}}$ are modified from the original AC \delete[gQ]{(moreover, $g_Q$ is propositional)}, so both these goals and their evidence are annotated by $\cmark$. \weakDelete{We continue recursively through the AC; there are no strategies annotated by $\textbf{?}$, so we can proceed straight to the leaves.}  Considering the \weakChange{sub-AC}{child AC} rooted at $\weakNew{\textbf{G6}}$, we can again certify \weakNew{that \textbf{G8}} and \weakNew{\textbf{G10}} \weakChange{and their}{remain supported by their} evidence. For $\textbf{G9}$, \weakChange{we could leverage }{suppose we apply}  the structural regression analysis for LTL model checking provided by Menghi et al.~\cite{menghi2021torpedo} \weakChange{. Suppose}{ and that}  it returns $\cmark$, indicating that the model checking result for $M$ can be reused for $M^\prime$. Then, following the rules for regression composition, we also annotate $\weakNew{\textbf{G5}}$ by $\cmark$, as all of its children are \weakNew{annotated by} $\cmark$. Suppose the same holds for $\weakNew{\textbf{G5}}$ and $\weakNew{\textbf{G7}}$. When we return to the querying subgoal $\weakNew{\textbf{G3}}$, it must be annotated by $\xmark$, since the strategy decomposing it is known to be \weakChange{unsound}{incomplete}. This $\xmark$ propagates through $\weakNew{\textbf{Str1}} $, resulting in $\weakNew{\textbf{G1}}$ being annotated by $\xmark$.
    \end{example}

\begin{theorem}\label{thm:instAlgCorrect}
   \weakNew{Let $A$ be an $AC$ such that $\Supp(A)$, and let $A_R$ be the reusable core extracted from $A$ following the forward pass of the analysis (Alg.~\ref{alg:instantiation}) given evolution set $\Delta$. Let $v$ be the regression value returned by $\mathsf{BackwardPass}(\Delta,A_R)$ (Alg.~\ref{alg:product_regression}). Then if $v = \cmark$, we know that $\Supp(A_R)$ holds;  if $v = \xmark$, we know that $\Supp(A_R)$ does not hold; and if $v = ~?$, we can make no determination about whether $\Supp(A_R)$ holds. }

\end{theorem}

In this section, we have defined a rigorous regression analysis targeting the support semantics of our formal assurance case language. The purpose \weakDelete{of this} was not merely to define this particular analysis, but to provide a foundation for the \emph{lifted} regression analysis defined in the next section.
\section{Lifted Assurance Case Regression for Evolving Product Lines}
\label{sec:liftedRegression}
In the preceding section, we defined a regression analysis for product ACs with respect to \weakDelete{the} assurance support. This section considers \emph{lifted assurance case regression}: given a variational AC developed for some model-based product line system, we want to analyze the impact of a system evolution on the AC in terms of variational support (Def.~\ref{def:vsupp}). \weakChange{As in the product setting, w}{W}e begin with a discussion of the semantics of variability-aware regression analysis  (Sec.~\ref{sec:liftedRegressionSpec}). We then lift both the forward pass (Sec.~\ref{sec:forwardLift}) and backward pass (Sec.~\ref{sec:backwardspassLift}) of our product-based regression analysis. Proofs of selected theorems not given here are provided in Appendix~\ref{app:proofs:sec6}.

\subsection{Variability-aware Regression Values}
\label{sec:liftedRegressionSpec}
\weakDelete{The goal of this section is to define \emph{lifted} assurance case regression, i.e., to define a modified version of  {Alg.~\ref{alg:product_regression}} which can be applied to variational ACs (Def.~\ref{def:placs}).} \weakDelete{Before we defined the analysis in Sec.~\ref{sec:product_regression}, we fixed a generic definition of regression analysis (Def.~\ref{def:regression}), with two motivations: first, we wanted to give a functional specification for the AC analysis; second, we wanted to \emph{use} existing individual analyses for regression of evidence (see Alg.~\ref{alg:productEvdRegression}). The same principle applies in the variational context. Thus,} We begin by formalizing  \emph{variability-aware regression values} as follows (adapted from ~\cite{shahin2021towards}): 

\begin{definition}[Variability-aware Regression Value]\label{def:varRegressionValues}
    Given a feature expression $\phi$ \newText[varregval]{over features $F$}, a \emph{variability-aware regression value} over $\phi$ is a tuple of feature expressions $\langle \phi_\cmark, \phi_\xmark, \phi_{\textbf{?}}\rangle$ \weakNew{over $F$} which partition $\llbracket \phi \rrbracket$; that is, for each $\conf \in \llbracket \phi\rrbracket$, $\conf$ belongs to exactly one of $\llbracket \phi_\cmark\rrbracket$, $\llbracket \phi_\xmark \rrbracket$, or $\llbracket \phi_{\textbf{?}}\rrbracket$. 
\end{definition}
\newText[clarsymbols]{ Def.~\ref{def:varRegressionValues} establishes that every variability-aware regression value is a tuple of 3 feature expressions, the first of which represents configurations which have \emph{not} regressed ( ``Reuse'', $\cmark$);  the second of which represents configurations which \emph{have} regressed (``Revise'', $\xmark$); and the third of which represents configurations for which regression cannot be determined (``Recheck'', $?$). We introduce additional notation to express a variability-aware regression value over $\phi$ in which all configurations of $\phi$ have the same regression value, namely, $\mathsf{REUSE}(\phi) = \langle \phi, \bot,\bot\rangle$, $\mathsf{REVISE}(\phi) = \langle \bot, \phi,\bot\rangle$, and $\mathsf{RECHECK}(\phi) = \langle \bot, \bot, \phi\rangle$}. 
\delete[varreg1]{We introduce the following shorthand for three special values: $\overline{\cmark} = \langle \top, \bot,\bot \rangle$, $\overline{\xmark} = \langle \bot, \top, \bot \rangle$, and $\overline{\textbf{?}} = \langle \bot, \bot, \top \rangle$. We also define an \emph{extension} operator ${\tt Extend}(v, s , \phi)$, where $v$ is a variability-aware regression value and $s \in \{\cmark, \xmark, \textbf{?}\}$, which adds $\phi$ as a disjunct to the $s$-component of $v$. For instance, given $v = \langle \phi_\cmark, \phi_\xmark, \phi_{\textbf{?}}\rangle$, we have $\mathsf{Extend}(v,\cmark,\phi) = \langle \phi_\cmark \lor \phi, \phi_\xmark, \phi_{\textbf{?}}\rangle$.} 
It should be clear that a variability-aware regression value encodes a \emph{product line} of regression values; the fact that every configuration can be mapped to exactly one of the three values $\{\xmark, \cmark, \textbf{?}\}$ provides the derivation function. Thus, we can let $\Var{\{\cmark, \xmark, \textbf{?}\}}$ denote the set of all possible regression values over $\phi$, where $\phi$ is inferred from the context. 

\newText[vartype0]{Before we precisely define what a variability-aware regression analysis should do, let us consider the different ways in which one product line may be evolved to form a new product line. In the simplest case, we may have an evolution in which some (or all) of the {products} in the product line have been modified, but the \emph{set} of configurations of the product line is otherwise the same. More precisely, this means that the set of features $F$ and feature model $\Phi$ are both  unchanged, only the products themselves being modified. Alternatively, we may consider an evolution of the product line in which some products are modified, the set of features $F$ is unchanged, but the \emph{feature model} has been changed. For instance, there may be two features $\token{A},\token{B}$ which were mutually exclusive in the original product line ($\neg (\token{A} \land \token{B})$), but are no longer mutually exclusive in the evolved product line, introducing a whole new set of configurations. Finally, there is the case in which some products are modified, the feature model is modified, \emph{and} the language of features is modified, e.g., by adding or removing a feature. All three kinds of evolutions can be analyzed for regression using the methods described in this section. With this in mind, we define the correctness specification of a variability-aware regression analysis as follows:}

\begin{definition}[Variability-aware Regression Analysis]\label{def:varRegression}
    \weakChange{L}{Let $X$ be a set and l}et $P$ be a predicate over $X$. A \emph{variability-aware regression analysis} for $P$ is a function \weakChange{$R^\uparrow : \Var{X} \times \Var{X} \times {\tt Prop}(F) \to \Var{\{\cmark, \xmark, \textbf{?}\}}$}{$R_P^\uparrow : \Var{X} \times \Var{X} \to \Var{\{\cmark, \xmark, ?\}}$} such that \weakChange{for any $\pl{x}, \pl{x}^\prime \in \Var{X}$ and $\phi \in {\tt Prop}(F)$,}{}{\ntn{for any product line $\pl{x} \in \Var{X}$ with feature model $\Phi$ over $F$ such that $\Inv{P}{\pl{x}}{\top}$ holds, and}}  {\ntn{given any $\pl{x}' \in \Var{X}$ with feature model $\Phi^\prime$} over $F'$,}
    $R_P^\uparrow(\pl{x},\pl{x}^\prime)$ produces a variability-aware regression value $\langle \phi_\cmark, \phi_\xmark, \phi_{\textbf{?}}\rangle$ over \weakChange{$\phi$, where}{$\Phi'$, such that for all valid configurations of $\Phi'$}:
\begin{itemize}
    \item \weakChange{for all $\conf \vDash \phi \land \phi_\cmark$}{if $\conf \vDash \phi_\cmark$,  then $\conf \in \llbracket \Phi\rrbracket$ and we can conclude $P(\pl{x}^\prime|_\conf)$}  
    \item \weakChange{for all $\conf \vDash \phi \land \phi_\xmark$, we know that $P(\pl{x}^\prime|_\conf)$ does \emph{not} hold}{if $\conf \vDash \phi_\xmark$,  it is not determined whether $P(\pl{x}'|_\conf)$ holds}
    \item \weakChange{for all $\conf \vDash \phi \land \phi_{\textbf{?}}$ it is unknown whether $P(\pl{x}|_\conf)$ is preserved or lost through the evolution.}{if $\conf \vDash \phi_{?}$, then $\conf \in \llbracket \Phi \rrbracket$ but we cannot determine whether $P(\pl{x}'|_\conf)$ follows from $P(\pl{x}|_\conf)$.}
\end{itemize}
\end{definition}
That is, a variability-aware regression value compares an ``old'' product line $\pl{x}$ against a ``new'' product line $\pl{x}^\prime$ and returns a variability-aware regression value indicating which configurations of $\pl{x}^\prime$ have regressed \weakNew{with respect to $P$} relative to their corresponding products in $\pl{x}$ {\ntn{(if they exist)}}.
\newText[vartype00]{Note that the configurations of the evolved feature model $\Phi^\prime$ can be partitioned into two sets: those which were also valid configurations of the original feature model $\Phi$, and those which are new configurations introduced by $\Phi^\prime$. Per Def.~\ref{def:varRegression}, any configurations of $\Phi^\prime$ which were \emph{not} included as part of the previous feature model \emph{must} be included within $\phi_\xmark$. In practical terms, when we are analyzing regression, we are really only concerned with configurations present under both feature models. In what follows, unless indicated otherwise, our regression analysis will assume the original and evolved product lines have the same feature model $\Phi$. This assumption that we have a unique feature model also allows us to speak of ``lifting'' a regression analysis, since lifting is only defined with respect to a single feature model. After we have defined our lifted regression analysis, we discuss its generalization to evolutions which do modify the feature model and/or language of features.
}

\delete[varlemma1]{Naturally, one can implement a variability-aware regression analysis by \emph{lifting} a product-based regression analysis (and this is precisely what we aim to do).}

    \delete[varlemma2]{\emph{LEMMA}. Let $P$ be a predicate over $X$ with regression analysis ${\ntn{R_P}}$. 
    \textcolor{red}{If ${{R_P}^\uparrow : \Var{X} \times \Var{X}} \to \Var{\{\cmark, \xmark, \textbf{?}\}}$} \textcolor{red}{is a lift of ${{R_P}}$, then ${{R_P}}^\uparrow$ is a variability-aware regression analysis (per Def.~\ref{def:varRegression}).}}

\weakDelete{Recall that in the product setting, our procedure annotates each node in the AC by a regression value indicating whether it has preserved its validity with respect to $\Supp$ (in the case of goals and evidence) or with respect to goal refinement (in the case of strategies). In the lifted setting, these nodes are annotated by variability-aware regression values. The final consideration before lifting the regression analysis is to define what we mean by an evolution of a product line model.}
\weakDelete{In general, a product line evolution can be purely \emph{structural}, in that it only consists of modifications to the model elements and/or their presence conditions, or it can include a \emph{variability} modification, e.g., adding or removing features or changing the feature model. In defining our lifted regression analysis, we begin by assuming that evolutions are purely structural, and at the end, generalize to include variability modifications.}

\weakChange{To describe structural evolutions precisely,}{As a final step before defining our lifted regression analysis,} we define a variational analog of the evolution set used in the \weakDelete{forward pass of the} product-based regression analysis (Sec.~\ref{sec:regression:forward}).
\begin{definition}[Variability-aware Evolution Set]\label{def:varEvoSet}
    Let $\mathcal{S} = \{\pl{M}_1\ellipses \pl{M}_n\}$ be the set of all product line models referred to in a variational AC, all over the same features $F$ and feature model $\Phi$. Let $\mathcal{S}^\prime = \{\pl{M}^\prime_1\ellipses  \pl{M}^\prime_n\}$ be the set of models following an evolution.  Then a \emph{variability-aware evolution set} over $\mathcal{S}$ is a function $\ntn{\widehat{\Delta}} : \mathcal{S} \to {\tt Prop}(F)$, such that for each $\pl{M}_i \in \mathcal{S}$, $\llbracket \ntn{\widehat{\Delta}}({\pl{M}_i})\rrbracket \subseteq \llbracket \Phi \rrbracket$, and for all configurations $\conf \in \llbracket \Phi \rrbracket$, $\conf \vDash \ntn{\widehat{\Delta}}(\pl{M}_i)$ if and only if in the evolved model $\pl{M}^\prime_i$, we have $\pl{M}_i|_\conf \neq \pl{M}^\prime_i|_\conf$. 
\end{definition}

In other words, the variability-aware \weakChange{impact}{evolution} set $\ntn{\widehat{\Delta}}$ describes which specific products have been modified for each product line model. 
\begin{example}
\begin{figure}[t]
    \centering

    \[\begin{tikzpicture}[shorten >=1pt,node distance=2.5cm,on grid,auto] 
   \node[state,initial] (q_0)   {$s_0$}; 
   \node[state] (q_1) [right=of q_0] {$s_1$}; 
   \node[state] (q_2) [right=of q_1] {$s_2$}; 
    \path[->] 
    (q_0) edge  node {$a \mid \textcolor{blue}{{\top}}$} (q_1)
    (q_2) edge [bend right=30, above] node {$c \mid \textcolor{blue}{{\tt A}}$} (q_0)
    (q_1) edge [below] node  {$d \mid \textcolor{blue}{{\tt A}}$}(q_2)
          edge [bend left=20, below] node {$b \mid \textcolor{blue}{{\tt B}}$} (q_0)
    (q_2) edge [loop right] node  {$e \mid \textcolor{blue}{{\tt A}}$} (q_2)
    ;
\end{tikzpicture}\]
\vspace{-3mm}
    \caption{\bf An evolution of the FTS $\pl{M}$ shown in Fig.~\ref{fig:FTS}, with $\Delta(\pl{M}) = \token{A}$.}
    \label{fig:change_example}
    % \vspace{-0.25in}
    \Description[]{}
\end{figure}
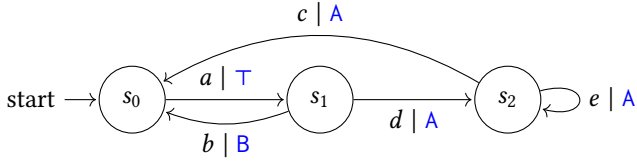
\nc{Recall the FTS shown in Fig.~\ref{fig:FTS}, which we refer to as $\pl{M}$. A modified version of $\pl{M}$ is shown in Fig.~\ref{fig:change_example}, in which a self-loop is added to state \typofix{$q_2$}{$s_2$}. This evolution of $\pl{M}$ only affects products with feature $\token{A}$.  Expressed using a variability-aware evolution set \weakNew{$\widehat{\Delta}$}, we have \weakNew{$\widehat{\Delta}(\pl{M}) = \token{A}$}.}
\end{example}
Given such a $\widehat{\Delta}$, we can denote the corresponding evolution set associated to configuration $\conf$ as $\widehat{\Delta}|_\conf$. We assume such a $\widehat{\Delta}$ is given as input to the lifted \weakChange{forward pass}{regression analysis}. Depending on the complexity of the model and the evolution, $\widehat{\Delta}$ could be derived manually, but in general it is preferable to produce it automatically, e.g., applying a variability-aware impact analysis~\cite{angerer2019change} on each of the product line models. 

\subsection{Lifting the Forward Pass}
\label{sec:forwardLift}

We now describe the lifting of the forward pass. The lifted forward pass follows the same overall procedure as the product-based version, and serves the same function: updating goals, re-instantiating templates where appropriate, and annotating strategies with regression values.  \weakDelete{As in the product setting, we begin by considering the regression of a (variational) non-analytic template.}

\subsubsection{Lifted Template Regression}
\begin{algorithm}[t]
\onehalfspacing
\caption{Regression of lifted templates}
\label{alg:templateRegressionLift}
    \begin{algorithmic}[1]
        \Procedure{$\mathsf{TemplateRegression}^\uparrow$}{$\pl{M}^\prime,\phi, st, \mathcal{A},\template{\modlang}{D}^\uparrow,\pl{x}, \phi_{\Delta}$}
        \State \textbf{fix} $\template{\modlang}{D}^\uparrow = \langle P,I^\uparrow,C\rangle$
        \If{some $\pl{x}^\prime$ is available such that $\Inv{C}{\langle \pl{x}^\prime, \pl{M}^\prime\rangle}{\phi}$
        }
        \If{$st$ is analytic and evidence-producing}{~$\{\pl{g}_1^\prime\ellipses \pl{g}_n^\prime\} := \mathsf{Replace}(\mathcal{A},\pl{x},\pl{x}')$}
        \Else $~\{\pl{g}_1^\prime\ellipses \pl{g}_n^\prime\}  := {I}(\pl{x}^\prime)$
        \EndIf
        \Else 
            \State $\pl{v}_{st} := \langle \phi \land \neg \phi_{\Delta},~~ \bot, ~~\phi \land \phi_\Delta\rangle$ \Comment{$\pl{v}_{st}$ partitions $\llbracket \phi \rrbracket$}
            \State \textbf{for} $\pl{A} \in \mathcal{A}$ s.t. ${\tt pc}(A) \implies \phi \land \phi_\Delta$ \textbf{do} $\mathsf{MarkAsRecheck}(\pl{A})$
            \State \Return ${\pl{v}_{st}}$
        \EndIf
       \State $\langle\mathcal{A}_{\tt Obs},\mathcal{A}_{\tt Reuse}, \mathcal{A}_{\tt New}\rangle := {\mathsf{Match}}^\uparrow(\mathcal{A},\{\pl{g}_1,\ellipses \pl{g}_n'\})$ 
       \State ${{\mathsf{UpdateChildren}(\mathcal{A}, \mathcal{A}_{\tt Obs}, \mathcal{A}_{\tt Reuse})}}$
       \State \textbf{for} $\pl{A} \in \mathcal{A}_{\tt Obs}$ \textbf{do} $\mathsf{MarkAsObsolete}(\pl{A})$
       \State $\phi_{\tt New} := \bigvee~\{\phi_i\mid \exists~ \pl{A}_i \in \mathcal{A}_{\tt New}, {~{\tt pc}(\pl{A}_i)} = \phi_i \}$
       \State $\pl{v}_{st} := \langle {{ \phi \land \neg \phi_{\tt New}}}, ~~{{ \phi \land \phi_{\tt New}}}, ~~\bot \rangle$ \Comment{$\pl{v}_{st}$ partitions $\llbracket \phi\rrbracket $}
       \State \Return $\pl{v}_{st}$ 
        \EndProcedure
    \end{algorithmic}
\end{algorithm}
\weakChange{The lifted template regression procedure, shown in Alg.~\ref{alg:templateRegressionLift}, is invoked when a parent goal referring to model $\pl{M}$ is decomposed by lifted template $T^\uparrow$, with $\pl{M}$ being modified for some configurations according to $\Delta$. We take as input the updated model $\pl{M}^\prime$, the lifted template $T^\uparrow = (P,I^\uparrow,C^\uparrow)$, the input $\pl{x}$ used to instantiate $T^\uparrow$, and the child ACs $\mathcal{A}$. We also take as input $\ntn{\phi}$, the presence condition of the template's parent goal, and $\phi_\Delta$, which represents the set of configurations which are modified in the evolution according to $\Delta$.}{Consider a variational goal $\langle \pl{M},P,\phi\rangle$ with $\pl{M} \in \Var{\modlang}$, which is decomposed via strategy $st$ to children $\mathcal{A}$, with $st$ being generated by the lifted template $\template{\modlang}{D}^\uparrow$ applied to input $\pl{x} \in \Var{D}.$ The \emph{lifted template regression} procedure, shown in Alg.~\ref{alg:templateRegressionLift}, takes these as input in addition to a feature expression $\phi_\Delta$, which represents the set of configurations in which the updated model $\pl{M}'$ has been modified relative to $\pl{M}$ (i.e., $\phi_\Delta = \widehat{\Delta}(\pl{M})$ ). The procedure mirrors the product-based regression analysis (Alg.~\ref{alg:templateRegression}). We first check whether we can produce any $\pl{x}' \in \Var{D}$ such that we satisfy the variational correctness criterion $\Inv{C}{\langle \pl{x}', \pl{M}'\rangle }{\phi}$ (Line 3).\footnote{Note that, in the case where not every configuration of $\phi$ has been modified, we may obtain such an $\pl{x}'$ beginning from the original input $\pl{x}$, which satisfied $\Inv{C}{x}{\phi}$, only ``repairing'' the configurations of $\pl{x}$ satisfying $\phi \land \phi_\Delta$ for $\pl{M}'$.} If no such $\pl{x}'$ can be found, then we cannot make any decision about the regression of the argument for configurations satisfying $\phi_{\Delta}$, but we can assume that the argument remains sound for configurations which are \emph{not} modified (i.e., satisfying $\neg \phi_\Delta$). We define the variability-aware regression value $\pl{v}_{st}$ accordingly (Line 7) so as to partition $\llbracket \phi \rrbracket$ into these two sets of configurations. We then annotate the entirety of each child AC which is present under configurations satisfying $\phi \land \phi_\Delta$ as ``Recheck'' (?), since we cannot know whether this child is still relevant (Line 8). We then return $\pl{v}_{st}$ (Line 9).}

\weakChange{The procedure begins by checking to see if $\pl{x}$ satisfies the lifted correctness criterion  with respect to $\pl{M}^\prime$ (Line 3), or if there is some other $\pl{x}^\prime$ which does (Line 4). In either case, we use the lifted instantiation function to produce the variational subgoals $\{\pl{g}_1^\prime\ellipses  \pl{g}_n^\prime\}$ for the set of configurations satisfying $\phi^\prime \land \phi_{\Delta}$. If no such input can be produced, we annotate all descendants of the strategy by $\overline{\textbf{?}}$ and return $\overline{\xmark}$ (Lines 6-7). Otherwise, we need to compare the new subgoals $\{\pl{g}_1^\prime\ellipses  \pl{g}_n^\prime\}$ against the existing sub-ACs $\mathcal{A}$. To this end, we use a lifted version of the $\mathsf{Match}$ operator (Line 8).} {Suppose instead that we \emph{were} able to find some input $\pl{x}'$ satisfying the correctness criterion (Line 3). Following our lazy evaluation procedure, if the strategy we are analyzing is analytic and evidence-producing, we do not re-instantiate the template, and only replace $\pl{x}$ with $\pl{x}'$ syntactically in the children of the strategy (Line 4). Otherwise, we re-instantiate the template with $\pl{x}'$. In both cases, we obtain a set of updated variational subgoals $\{\pl{g}_1'\ellipses \pl{g}_n'\}$. As before, the final step is to compare these updated subgoals against the children of the original strategy in $\mathcal{A}$. We do so using the \emph{lifted} matching operator $\mathsf{Match}^\uparrow$.}

Like the product-based $\mathsf{Match}$, the lifted operator attempts to unify the new subgoals $\pl{g}_i^\prime$ against the root goals of $\pl{A}_i \in \mathcal{A}$. The key difference is that we now need to account for potential variations in the presence conditions. Suppose that $\pl{g}^\prime = \langle \pl{M},P,\phi^\prime\rangle$ is a new subgoal produced by re-instantiation, and \weakNew{it is ``matched'' with some existing subgoal} $\pl{g} = \langle \pl{M},P,\phi \rangle$ \weakDelete{is the root subgoal of some existing sub-AC $\pl{A} \in \mathcal{A}$}. If \weakChange{$\phi \not\equiv \phi^\prime$}{$\phi$ and $\phi^\prime$ do not denote the same sets of configurations,}, we need to determine for which configurations $\pl{g}$ is obsolete, reusable, or new.
    \nc{We thus define three feature expressions: $\phi_{\tt O}$ (obsolete), $\phi_{\tt R}$ (reuse), and $\phi_{\tt N}$ (new), where $\phi_{\tt O} = \phi \land \neg \phi^\prime$, \weakChange{$\phi_{K} = \phi \land \phi^\prime$}{$\phi_{\tt R} = \phi \land \phi'$}, $\phi_{\tt N} = \phi^\prime \land \neg \phi$. When we update $\pl{A}$ with the new root goal $\pl{g}^\prime$, we make \weakNew{(up to)} three copies: one, annotated by $\phi_{\tt O}$, is added to $\mathcal{A}_{\tt Obs}$; another, annotated by $\phi_{\tt R}$, is added to $\mathcal{A}_{\tt Reuse}$; the third, annotated by $\phi_{\tt N}$, is added to $\mathcal{A}_{\tt New}$. \weakChange{In the special cases where $\phi^\prime \implies \phi$ or $\phi \implies \phi^\prime$, we have {$\phi_{\tt O} \equiv \bot$}{$\llbracket  \phi_{\tt O} \rrbracket = \emptyset$} or {$\phi_{\tt N}\equiv \bot$}{$\llbracket \phi_{\tt N} \rrbracket = \emptyset$}, respectively, so we only need to make two copies.}{Note that we can make fewer than three copies of $\pl{g}'$ if we observe that any of $\phi_{\tt O}$, $\phi_{\tt N}$, or $\phi_{R}$ denote the empty set of configurations.}}

\begin{example}
Suppose we match an updated goal $\pl{g}^\prime$, with presence condition ${\tt A}$, with an existing goal $\pl{g}$, whose presence condition is ${\tt A} \land {\tt B}$. Then we can create two copies of $\pl{g}^\prime$, $\{\pl{g}^\prime_{\tt B}, \pl{g}^\prime_{\tt \neg B}\}$, where $\pl{g}^\prime_{\tt B}$ has presence condition ${\tt A} \land {\tt B}$, while $\pl{g}^\prime_{\tt \neg B}$ has presence condition ${\tt A} \land \neg {\tt B}$. Only $\pl{g}^\prime_{\tt B}$ is added to  $\mathcal{A}_{\tt Reuse}$, and the undeveloped goal $\pl{g}^\prime_{\tt \neg B}$ is added to $\mathcal{A}_{\tt New}$.
\end{example}
 
Once we have determined \weakChange{which variational sub-ACs}{for which configurations each of the children of the strategy} are now obsolete, reused, or new, we \weakNew{ replace each child with its copies from $\mathcal{A}_{\tt Obs}$ and $\mathcal{A}_{\tt Reuse}$ (Line 11). We then} annotate the obsolete ones accordingly {\ntn{(Line 12)}}. Next, we compute a feature expression $\phi_{\tt New}$ describing which configurations are missing one or more subgoals; we do so by taking the disjunction of the presence conditions of all ACs in $\mathcal{A}_{\tt New}$ {\ntn{(Line 13)}}. The final regression value $\pl{v}_{st}$ for this strategy \weakDelete{(with respect to $\phi_\Delta$)} is then taken as {\ntn{``Reuse'' ($\cmark$)}} for configurations satisfying \weakChange{$\neg \phi_{\tt New}$}{$\phi \land \neg \phi_{\tt New}$}, and {\ntn{``Revise'' ($\xmark$)}} for those satisfying \weakChange{$\phi_{\tt New}$}{$\phi \land \phi_{\tt New}$} (Line 14). \weakNew{Note once again that $\pl{v}_{st}$ partitions $\llbracket \phi \rrbracket$.}

\begin{lemma}\label{lem:templateRegLift}

    \weakNew{Let $\pl{g} = \langle \pl{M}, P, \phi \rangle $ be a variational goal decomposed using the valid lifted template $\template{\modlang}{D}^\uparrow$, producing strategy $st$ and children $\mathcal{A}$ via instantiation with $\pl{x} \in \Var{D}$. Let $\pl{M}' \in \Var{\modlang}$ be the updated version of $\pl{M} \in \Var{\modlang}$ following a system evolution. Then for every $\conf \in \llbracket \phi \rrbracket$, we have} 
    $$\mathsf{TemplateRegression}^\uparrow(\pl{M}',\phi,st,\mathcal{A},\template{\modlang}{D}^\uparrow,\pl{x},\phi_{\Delta})|_\conf = {\mathsf{TemplateRegression}(\pl{M}'|_\conf, st, \mathcal{A}|_\conf, \template{\modlang}{D}, \pl{x}|_\conf)}$$
    \weakNew{That is, Alg.~\ref{alg:templateRegressionLift} correctly lifts Alg.~\ref{alg:templateRegression}}.
\end{lemma}

\subsubsection{The Forward Pass}
\begin{algorithm}[t]
    \caption{Forward Pass of the Lifted Regression Analysis}
    \label{alg:instantiationLift}
    \begin{algorithmic}[1] 
       \onehalfspacing
        \Procedure{$\mathsf{ForwardPass}^\uparrow$}{$\pl{A},{{\widehat{\Delta}}}$}
                \If{$\pl{g} = {\tt Rt}(\pl{A})$ refers to model $\pl{M}$}{~$\mathsf{Update}(\pl{g},\pl{M}^\prime)$} \Comment{$(\pl{M}, \pl{M}') \in \widehat{\Delta}$}
                \EndIf
                \Switch{$\pl{A}$}
                \Case{${\tt Und}(\pl{g})$ or ${\tt Evd}(\pl{g},e)$:}~\Return
                \EndCase 
                \Case{${\tt Decomp}(\pl{g}, st, \mathcal{A})$:}
                    \State \textbf{fix} $\pl{g} = \langle \pl{M}, P, \phi\rangle$
                    \If{$st$ instantiates template ${{\template{\modlang}{D}}^\uparrow}$ with input $\pl{x}$}
                    \State $\phi_\Delta := \widehat{\Delta}(\pl{M})$
                    \If{$\llbracket \phi \land \phi_\Delta\rrbracket = \emptyset$} 
                          \State $\mathsf{Annotate}(st, \mathsf{REUSE}(\phi))$
                          \State \textbf{for} $\pl{A}' \in \mathcal{A}$ \textbf{do} $\mathsf{ForwardPass}^\uparrow(\pl{A}',\widehat{\Delta})$
                        \Else
                           \State $\pl{v}_{st} := \mathsf{TemplateRegression}^\uparrow(\pl{M}^\prime,\phi,st,\mathcal{A},\template{\modlang}{D}^\uparrow,\pl{x},\phi_\Delta)$ \Comment{$(\pl{M},\pl{M}') \in \widehat{\Delta}$}
                            \State $\mathsf{Annotate}(st,\pl{v}_{st})$
            \State $\mathcal{A}_{\tt Obs} := \{\pl{A}_i \in \mathcal{A} \mid \pl{A}_i $ is not marked obsolete$\}$
            \State $\mathcal{A}_{?} := \{\pl{A}_i \in \mathcal{A} \mid \token{pc}(\pl{A}_i) \Rightarrow \phi_?\}$ \Comment{$\pl{v}_{st} = \langle \phi_\cmark, \phi_\xmark, \phi_?\rangle$}

            \State\textbf{for} $\pl{A}' \in \mathcal{A} \setminus (\mathcal{A}_{\tt Obs} \cup \mathcal{A}_?)$  \textbf{do} $\mathsf{ForwardPass}^\uparrow(\pl{A}',\widehat{\Delta})$
                        \EndIf 

                    \ElsIf{$st$ is not template-based}
                        \State $\mathsf{Annotate}(\{st\}, \mathsf{RECHECK}(\phi))$
                        \State \textbf{for} $\pl{A}' \in \mathcal{A}$ \textbf{do} $\mathsf{MarkAsRecheck(\pl{A}')}$
                    \EndIf
                \EndCase
            \EndSwitch
        \EndProcedure
    \end{algorithmic}
\end{algorithm}

The complete lifted forward pass is shown in Alg.~\ref{alg:instantiationLift}. As in the product \weakChange{version}{setting}, we proceed recursively through the AC, \weakChange{updating goals whenever we find one referring to a model which is modified under at least one configuration}{updating the subject of a predicative goal with the evolved product line model whenever such a goal is found} (Line 2).  In the recursive case, we once again perform case analysis on how \weakNew{the strategy} $st$ was created\weakChange{,with the most involved case being the regression of a non-analytic template as outlined above. In such a case, if the parent goal is not modified, we can annotate $st$ by $\overline{\cmark}$ and continue.}{. If $st$ was not generated via a formal template, we cannot make a decision about its regression, so we annotate $st$ and its descendants by $\mathsf{RECHECK}(\phi) = \langle \bot,\bot,\phi\rangle$ (Lines 19-20). Otherwise, if $st$ was obtained by instantiating a template with some input $\pl{x}$, we first compute the set of configurations in which the model of the parent goal has been modified, which we denote as $\phi_\Delta$ (Line 8). If there are no configurations for which the current strategy is present ($\phi$) \emph{and} which have been modified under this evolution ($\phi_\Delta$), we can annotate the strategy as $\mathsf{REUSE}(\phi) = \langle \phi, \bot, \bot\rangle$, i.e., the strategy can be reused across all relevant configurations, and we proceed recursively through the children (Lines 10-11). Otherwise, we need to invoke the template regression procedure (Alg.~\ref{alg:templateRegressionLift}) to obtain the regression value $\pl{v}_{st}$ for this strategy. In the product setting, we returned early if we found the strategy was annotated as ``Recheck'', and only recursively analyzed non-obsolete goals. In the lifted setting, there may be some configurations for which the strategy is reusable, and some for which the strategy is marked as ``Recheck'' (cf. Line 7 of Alg.~\ref{alg:templateRegressionLift}).   Mirroring the product-based analysis, we exclude from the recursion  children which are  marked as obsolete (Line 15) \emph{or} which are present under configurations for which the strategy is marked ``Recheck'' (Line 16).} 

\weakDelete{Otherwise, we invoke the lifted template regression procedure, which returns the regression value $\pl{v}_{st}$ for this strategy. Since this value only covers configurations in $\phi_\Delta$, we need to extend it with $\phi_{\neg \Delta}$, the configurations of which preserve the soundness of the strategy. After annotating the strategy with $\pl{v}_{st}$, we continue recursively on all subgoals which are not marked obsolete, nor with \textbf{?}. The remaining cases are relatively simple, reflecting almost exactly their analogues in Alg.~\ref{alg:instantiation}. If $st$ is an analytic template,  we annotate it by $\overline{\cmark}$, and re-execute the lifted template only if the product line is modified and it is argument-building. If $st$ is not template-based, we annotate it by $\overline{\textbf{?}}$ and halt.}

\begin{theorem}\label{thm:forwardsLift}
    \weakNew{Let $\pl{A} \in \Var{\AC}$ with $\token{pc}(\pl{A}) = \phi$. Then for every $\conf \in \llbracket \phi \rrbracket$, we have} 
    $$\mathsf{ForwardPass^\uparrow( \pl{A},\widehat{\Delta})|_\conf} = \mathsf{ForwardPass}(\pl{A}|_\conf,\widehat{\Delta}|_\conf)$$
    \weakNew{That is, Alg.~\ref{alg:instantiationLift} correctly lifts Alg.~\ref{alg:instantiation}.}
\end{theorem}

\weakNew{As in the product-based analysis, the forward pass results in an annotated variational AC in which every strategy is annotated by a regression value, and every goal will either be unannotated, annotated as ${\tt RECHCEK}$ (?) for all relevant configurations, or marked as obsolete.  As in the product-based analysis, we can extract the \emph{reusable core} $\pl{A}_R$ of the variational AC $\pl{A}$ following the forward pass of the analysis, by pruning all nodes which are  marked as obsolete. Note that the reusable core $\pl{A}_R$ of $\pl{A}$ remains a structurally valid (and well-formed) variational AC.}

\subsection{Lifting the Backward Pass}
\label{sec:backwardspassLift}

\weakChange{We now describe the lifting of the backward pass. As in the product setting, we first consider the regression of evidence, and then turn to the composition of regression values.}{Having completed the forward pass of the original variational AC $\pl{A}$, and extracted the reusable core $\pl{A}_R$, we now seek to determine for which configurations the goals of $\pl{A}_R$ can still be supported by evidence. Suppose we are considering an (updated) variational goal $\pl{g}'$, with existing evidence $e$ adequate for original goal $\pl{g}$. The process for computing the regression of support for $\pl{g}'$ is shown in Alg.~\ref{alg:varEvdRegression}. Mirroring the product-based analysis, we first consider whether $\pl{g}'$ is propositional (in which case we assume it has not regressed), or whether the model referred to by $\pl{g}'$ is not modified under any configurations in which $\pl{g}'$ is present (Line 2). In either case, we can determine that support for $\pl{g}'$ is preserved, so we set $\pl{v} = \mathsf{REUSE}(\phi) = \langle \phi, \bot,\bot\rangle$. Otherwise, we have a predicative goal $\pl{g}' = \langle \pl{M}', P, \phi \rangle$ such that for some configurations of $\phi$, the product line model $\pl{M}$ has been modified as $\pl{M}'$ (Line 6). Mirroring the product-based analysis, we then consider whether we have a (lifted) regression analysis $R_P^\uparrow$. If we do, then we can use it to compute the regression value with respect to $\phi$ (Line 8).\footnote{In practice, one can restrict the feature model further to $\phi \land \widehat{\Delta}(M)$, since we know \emph{a priori} that $P$ has been preserved for the unmodified configurations satisfying $\phi \land \neg \widehat{\Delta}$.} Otherwise, if no such analysis exists, we compute the set of configurations under which this product line model $\pl{M}$ has been modified (denoted $\phi_\Delta$) and define the regression value $\pl{v}$  to partition $\llbracket \phi \rrbracket$ into the set of configurations satisfying $\neg \phi_\Delta$, for which the evidence is reusable, and the set of configurations satisfying $\phi_\Delta$, for which it is unknown whether the evidence is reusable (Lines 10-11). In either case, we conclude by annotating both $e$ and $\pl{g}'$ with $\pl{v}$ and returning $\pl{v}$ (Lines 12-13). }

\begin{algorithm}[t]
    \caption{Variational Evidence Regression}
    \label{alg:varEvdRegression}
    \begin{algorithmic}[1] 
        \onehalfspacing
        \Procedure{$\mathsf{EvdRegression}^\uparrow$}{$\widehat{\Delta}, \pl{g}^\prime,e$}
        \State \textbf{fix} $\phi = \token{pc}(\pl{g}')$

            \If{$g'$ is propositional or $g'$ refers to $\pl{M}'$ s.t. $(\pl{M},\pl{M}') \in \widehat{\Delta}$ and $\llbracket \phi \wedge \widehat{\Delta}(\pl{M}) \rrbracket = \emptyset$ }
             \State $\pl{v} :=\mathsf{REUSE}(\phi)$
            \Else
            
           \State \textbf{fix} $\pl{g}' = \langle \pl{M}, P, \phi\rangle$ and  $\pl{M}'$ such that $(\pl{M}, \pl{M}') \in \widehat{\Delta}$

            \If{lifted regression analysis ${\ntn{R_P^\uparrow}}$ for $P$ is available}
            \State $\pl{v} := R_P^\uparrow(\pl{M}, \pl{M}')$ \Comment{Evaluate under restricted feature model $\phi$}
            \Else
            \State  $\phi_\Delta := \widehat{\Delta}(\pl{M})$
            \State {$\pl{v} := \langle \phi \land \neg \phi_\Delta,\bot,\phi \land \phi_\Delta\rangle$} \Comment{$\pl{v}$ partitions $\llbracket \phi\rrbracket $}
            \EndIf
        \EndIf
           \State $\mathsf{Annotate}(\{e, \pl{g}'\},\pl{v})$
           \State \Return $\pl{v}$
        \EndProcedure
    \end{algorithmic}
\end{algorithm}
\weakChange{Suppose we want to determine the variability-aware regression value $\pl{v}$ between variational (predicative) goals $\pl{g}, \pl{g}^\prime$, where $\pl{g}$ is supported by evidence $e$. We follow the procedure outlined in Alg.~\ref{alg:varEvdRegression}. In addition to the goals and evidence, the procedure takes $\phi$, the presence condition of $\pl{g}^\prime$ and feature expression $\phi_{\Delta}$ describing the set of configurations for which the model referred to in $g$ is modified in the evolution.}{}

\weakDelete{We fix $\pl{g} = \langle \pl{M},P,\phi\rangle$, $\pl{g}^\prime = \langle \pl{M}^\prime, P, \phi\rangle$, and $e : \Inv{\pl{M}}{P}{\phi}$. We need to determine the regression for (and \emph{only} for) the configurations satisfying $\phi_{\Delta} \land \phi^\prime$. As in the product setting, we do a case analysis on the existence of a (lifted) regression analysis $R^\uparrow$ for $P$. Lacking a lifted analysis, we can again consider simply re-verifying $P$ for this set of configurations in a {brute-force} manner as a worst-case form of regression analysis, but this is only feasible if $\llbracket \phi_\Delta \land \phi^\prime \rrbracket$ is not too large. In any case, such an analysis $R^\uparrow$ will return a variability-aware regression value which partitions $\phi_{\Delta} \land \phi^\prime$ (Line 4). We then need to extend it to cover all configurations of $\phi^\prime$; we do so by extending the $\cmark$-component to cover $\neg\phi_{\Delta}$, i.e.,  those configurations which did not need to be analyzed as they were not affected by the change (Line 5). If, by contrast, no such $R^\uparrow$ was available, and the set of relevant configurations was too large to verify in a brute-force manner, we can only split the configurations of $\phi^\prime$ into those satisfying $\phi_\Delta$, whose regression values are indeterminate, and those satisfying $\neg \phi_\Delta$, which are known not to be affected (Line 6). We conclude by annotating $\pl{g}$ and $e$ with $\pl{v}$ before returning $\pl{v}$.}

\begin{lemma}\label{lem:evdRegLiftCorrect}

    \weakNew{Let $\widehat{\Delta}$ be a variational evolution set, $\pl{g}'$ be a variational goal with presence condition $\phi$, and $e$ be adequate evidence for $\pl{g}'$. Then for all configurations $\conf \in \llbracket \phi \rrbracket$, we have }
    $$\mathsf{EvdRegression}^\uparrow(\widehat{\Delta},\pl{g}',e)|_\conf = \mathsf{EvdRegression}(\widehat{\Delta}|_\conf, \pl{g}'|_\conf, e)$$
    \weakNew{i.e., Alg.~\ref{alg:varEvdRegression} correctly lifts Alg.~\ref{alg:productEvdRegression}}.
\end{lemma}

\subsubsection{Composition of Variational Regression Values} {\;}
\weakChange{We next need to determine a suitable notion of \emph{composition} for lifted regression values, corresponding to $\mathsf{Min}$ in Alg.~\ref{alg:product_regression}. We  define $\mathsf{Min}^\uparrow$ as the $n$-ary operator as follows: given inputs $\{\pl{v}_1\ellipses  \pl{v}_n\}$ with $\pl{v}_i = \langle \phi^i_\cmark, \phi^i_\xmark, \phi^i_{\textbf{?}} \rangle$, we define $\mathsf{Min}^\uparrow(\{\pl{v}_1\ellipses  \pl{v}_n\})$ $ = \langle \phi_\cmark^\prime, \phi_\xmark^\prime, \phi_{\textbf{?}}^\prime \rangle$, where $\phi_\cmark^\prime = \bigwedge_i \phi_\cmark^i$, $\phi_\xmark^\prime = \bigvee_i \phi^i_\xmark$, and $\phi_{\textbf{?}}^\prime = \neg (\phi_\cmark^\prime \lor \phi_\xmark^\prime)$}{Having defined a procedure for analyzing the regression of variational evidence, all that remains is to consider the \emph{composition} of variability-aware regression values. That is, given a variational decomposition ${\tt Decomp}(\pl{g}, st, \{\pl{A}_1\ellipses \pl{A}_n\})$, given variability-aware regression values for $st$ and the root goals of each $\pl{A}_i$, we want to compute the appropriate regression value for $\pl{g}$. More precisely, we need to define a lifted composition operator $\mathsf{Min}^\uparrow(\{\pl{v}_1\ellipses \pl{v}_n, \pl{v}_{st}\})$, where each $\pl{v}_i$ is the regression value obtained from $\pl{A}_i$, and $\pl{v}_{st}$ is the regression value of $st$, which properly lifts the semantics of the $\mathsf{Min}$ operator defined in Sec.~\ref{sec:product_regression}. Note that the regression values obtained from each $\pl{A}_i$  only partition $\phi_i = \token{pc}(\pl{A}_i)$, and the regression value for $\pl{g}$ must provide a partitioning of $\phi = {\tt pc}(\pl{g})$.

We can begin by considering the simpler task of defining a binary operator $\pl{v}_1 \otimes \pl{v}_2$, where $\pl{v}_1$ is a regression value partitioning $\phi_1$, and $\pl{v}_2$ is a regression value partitioning $\phi_2$. Suppose that we want to compute a composite regression value $\pl{v} = \pl{v}_1 \otimes \pl{v}_2$ which partitions $\llbracket \phi_1 \vee 
\phi_2\rrbracket$, and consider an arbitrary $\conf \in \llbracket \phi_1 \vee \phi_2 \rrbracket$. We can reason as follows:
\begin{enumerate}
    \item If $\conf \in \llbracket \phi_1 \wedge \neg \phi_2 \rrbracket$, then $\pl{v}|_\conf$ reduces to $\pl{v}_1|_\conf$.
    \item If $\conf \in \llbracket \phi_2 \wedge \neg \phi_1 \rrbracket$, then $\pl{v}|_\conf$ reduces to $\pl{v}_2|_\conf$.
    \item Otherwise, if $\conf \in \llbracket \phi_1 \wedge \phi_2\rrbracket$, then, to respect the ordering $\xmark < ~? < \cmark$:
    \begin{enumerate}[topsep=0.2em]
        \item To have $\pl{v}|_\conf = \cmark$, we must have both $\pl{v}_1|_\conf = \cmark$ and $\pl{v}_2|_\conf = \cmark$
        \item To have $\pl{v}|_\conf = \xmark$, it suffices to have either $\pl{v}_1|_\conf = \xmark$ or $\pl{v}_2|_\conf = \xmark$
        \item to have $\pl{v}|_\conf = ?$, it suffices to have either $\pl{v}_1|_\conf = ?$ or $\pl{v}_2|_\conf = ?$, so long as the \emph{other} regression value does not become $\xmark$ under $\conf$.
    \end{enumerate}
\end{enumerate}

Thus, given the regression value $\pl{v}_1 = \langle \phi_\cmark, \phi_\xmark, \phi_?\rangle$ partitioning $\llbracket \phi\rrbracket$, and $\pl{v}_2 = \langle \psi_\cmark, \psi_\xmark, \psi_?\rangle$ partitioning $\llbracket \psi \rrbracket$, we can compute $\pl{v}_1 \otimes \pl{v}_2 = \langle \theta_\cmark, \theta_\xmark, \theta_? \rangle$, where 
\begin{align*}
    \theta_\cmark &= (\phi_\cmark \wedge \psi_\cmark) \vee (\phi_\cmark \wedge \neg \psi) \vee (\psi_\cmark \wedge \neg \phi) \\ 
    \theta_\xmark &= \phi_\xmark \vee \psi_\xmark\\
    \theta_? &= (\phi_? \wedge \neg \psi) \vee (\psi_? \wedge \neg \phi) \vee ((\phi_? \vee \psi_?) \wedge\neg \phi_\xmark \wedge \neg \psi_\xmark )
\end{align*}

\begin{lemma}\label{lem:compBinCorrect}
    Let $\pl{v}_1 = \langle \phi_\cmark, \phi_\xmark, \phi_?\rangle$ be a regression value partitioning $\llbracket \phi\rrbracket$, and $\pl{v}_2 = \langle \psi_\cmark, \psi_\xmark, \psi_?\rangle$ be a regression value partitioning $\llbracket \psi \rrbracket$. Then $\pl{v}_1 \otimes \pl{v}_2$ is a variability-aware regression value over $\llbracket \phi  \vee \psi \rrbracket$ such that for all $\conf \in \llbracket \phi \vee \psi \rrbracket$, we have $(\pl{v}_1 \otimes \pl{v}_2)|_\conf = \mathsf{Min}(\{\pl{v}_1,\pl{v}_2\}|_\conf)$.
\end{lemma}
\weakNew{Note that $\{\pl{v}_1, \pl{v}_2\}|_\conf$ evaluates to $\{\pl{v}_1|_\conf, \pl{v}_2|_\conf\}$ if $\conf \in \llbracket \phi \land \psi \rrbracket$, and to $\{\pl{v}_1|_\conf\}$ (resp. $\{\pl{v}_2|_\conf\}$) if $\conf \in \llbracket \phi \land \neg \psi \rrbracket$ (resp. $\llbracket \neg \phi \land \psi\rrbracket$). 
 It is straightforward to show that $\otimes$ is both associative and commutative. With these properties in mind, we can define the lifted operator $\mathsf{Min}^\uparrow(\{\pl{v}_1\ellipses \pl{v}_n\})$ as the $n$-ary fold}
$$\mathsf{Min}^\uparrow(\{\pl{v}_1\ellipses \pl{v}_n\}) = \pl{v}_1 \otimes \pl{v}_2 \otimes \ldots \otimes \pl{v}_n$$
}
\weakNew{A proof by induction on $n$ provides the more general result we require.}
\begin{lemma}\label{lem:minLiftCorrect}
    \weakNew{Let $\{\pl{v}_1\ellipses \pl{v}_n\}$ be a set of $n \geq 1$ variability-aware regression values, such that each $\pl{v}_i$ partitions  $\llbracket \phi_i\rrbracket$ for some feature expression $\phi_i$ (over the same alphabet of features). Then for every $\conf \in \llbracket \bigvee_i \phi_i \rrbracket$, we have} 
    $$\mathsf{Min}^\uparrow(\{\pl{v}_1\ellipses \pl{v}_n\})|_\conf = \mathsf{Min}(\{\pl{v}_1\ellipses \pl{v}_n\}|_\conf)$$
\end{lemma}

\subsubsection{The Backward Pass}

\begin{algorithm}[t]
    \caption{Backward Pass of the Lifted Regression Analysis}
    \label{alg:lifted_assurance_regression}
    \begin{algorithmic}[1] 
        \onehalfspacing
        \Procedure{$\mathsf{BackwardPass}^\uparrow$}{$\widehat{\Delta},\pl{A}_R$} 
            \Switch{$\pl{A}_R$}
                \Case{${\tt Und}(\pl{g}')$:}
                    \State 
                    $\mathsf{Annotate}(\pl{g}', \mathsf{REVISE}(\phi))$  \Comment{$\phi = {\token{pc}}(\pl{g}')$} 
                    \State \textbf{return} ${\tt REVISE}(\phi)$
                \EndCase
                \Case{${\tt Evd}(\pl{g}^\prime,e)$:} 
                    \State \textbf{return} $\mathsf{EvdRegression}^\uparrow(\widehat{\Delta},\pl{g}^\prime, e)$
                    % \EndIf
                \EndCase
                \Case{${\tt Decomp}(\pl{g}^\prime, st, \{\pl{A}_1^\prime\ellipses  \pl{A}_n^\prime\})$:} 
                    \State $\pl{v}_{st} := \mathsf{Annotation}(st)$
                    \If{$\pl{v}_{st} = \mathsf{RECHECK}(\phi)$} \Comment{$\phi = {\token{pc}}(\pl{g}')$} 
                        \State $\mathsf{Annotate}(\pl{g}',\pl{v}_{st})$
                        \State \textbf{return} $\pl{v}_{st}$
                    \EndIf

                    \State  $V:= \{\pl{v}_{st}\}$
                    \State \textbf{for each} $\pl{A}_i$ \textbf{do} $V := V \cup \{\mathsf{BackwardPass}^\uparrow(\widehat{\Delta},\pl{A}_i')\}$
                        \State $\pl{v} := \mathsf{Min}^\uparrow(V)$ 
                        \State $\mathsf{Annotate}(\pl{g}',\pl{v})$
                    \State \Return $\pl{v}$
                \EndCase
            \EndSwitch
        \EndProcedure
    \end{algorithmic}
\end{algorithm}
\label{sec:backwardsPassLift}
We now have all we need to define the lifted backward pass\weakChange{. The procedure is}{,} specified in Alg.~\ref{alg:lifted_assurance_regression}. \weakChange{As in the product setting, w}{W}e assume that the forward pass has been applied to a variational AC $\pl{A}$\weakChange{ to yield $\pl{A}^\prime$, and that all obsolete branches have been pruned from both ACs, maintaining symmetry}{, and that we have subsequently extracted the reusable core $\pl{A}_R$.} \weakChange{Furthermore, we restrict the presence conditions on both ACs to only reflect the configurations for which matched goals are reusable, as per the description of $\mathsf{Match}^\uparrow$ in Sec.~\ref{sec:forwardLift}. That is, the ACs have equivalent presence conditions on all goals. {Given} both ACs as input, alongside $\phi^\prime$, the presence condition on the root goal of $\pl{A}^\prime$, and the variability-aware evolution set $\Delta$.In the base case, we first check whether we can skip evidence regression if $\pl{g}$ is either propositional or its goal model is unmodified (Line 5). There is also a third case, in which the model referred to in the goal is modified, but this modification does not apply to any of the configurations which are relevant to this goal (Line 8). In all three cases, we annotate the goal and evidence by $\overline{\cmark}$ and return this value. Otherwise, we invoke the lifted evidence regression procedure, which will annotate the goal and evidence accordingly and return a regression value.}{We proceed recursively on the structure of $\pl{A}_R$. We denote the goals of the AC using the notation $\pl{g}'$ to emphasize that they may have been updated during the forward pass. Mirroring the product-based analysis, if we find an undeveloped goal $\pl{g}'$ with presence condition $\phi$, it must be the case that all of its subgoals have become obsolete and been removed (assuming $\pl{A}$ was originally supported). We therefore annotate $\pl{g}'$ with $\mathsf{REVISE}(\phi)$ and return this value (Lines 4-6). In the second base case, where we find a variational goal $\pl{g}'$ supported by evidence $e$, we invoke the lifted evidence regression procedure (Alg.~\ref{alg:varEvdRegression}). }  

\weakChange{In the recursive case, we check the annotation which was made to strategy $st^\prime$ during the forward pass. If it is $\overline{\textbf{?}}$, we have encountered an informal strategy whose descendants have already been annotated by $\overline{\textbf{?}}$, so we return this value (Line 14). Similarly, if the strategy is annotated by $\xmark$, but each child is already annotated by $\overline{\textbf{?}}$, this means we could not find a reusable argument fragment during the forward pass, so we proceed no further, returning $\overline{\xmark}$. Otherwise, we recursively compute regression values for each child, and determine the regression value for the parent goal by taking their (lifted) minimum}{In the recursive case, in which goal $\pl{g}'$ is decomposed by strategy $st$ into $\{\pl{A}_1'\ellipses \pl{A}_n'\}$, we extract the regression value $\pl{v}_{st}$ assigned to $st$ during the forward pass (Line 9), and check whether it is $\mathsf{RECHECK}(\phi)$. If so, we know that all descendants of this strategy have been marked as ``Recheck'' for all relevant configurations, so we can annotate $\pl{g}'$ with this same value and return it (Lines 11-12). Otherwise, we recursively compute regression values for each of the children and add them to set $V$ along with $\pl{v}_{st}$ (Lines 13-14). We then apply the lifted operator $\mathsf{Min}^\uparrow$ to these $n+1$ regression values to determine the regression  of parent goal $\pl{g}$, which we then return (Lines 15-17). Since the annotation of $\pl{v}_{st}$ partitions $\llbracket \phi\rrbracket$, the resulting annotation for parent goal $\pl{g}'$ will also partition $\llbracket \phi \rrbracket$, since as the decomposition is well-formed (Def.~\ref{def:wellformed}).}

\begin{example}\label{ex:lif_reg_example}
\begin{figure}[t]
    \centering
    \includegraphics[width=\linewidth]{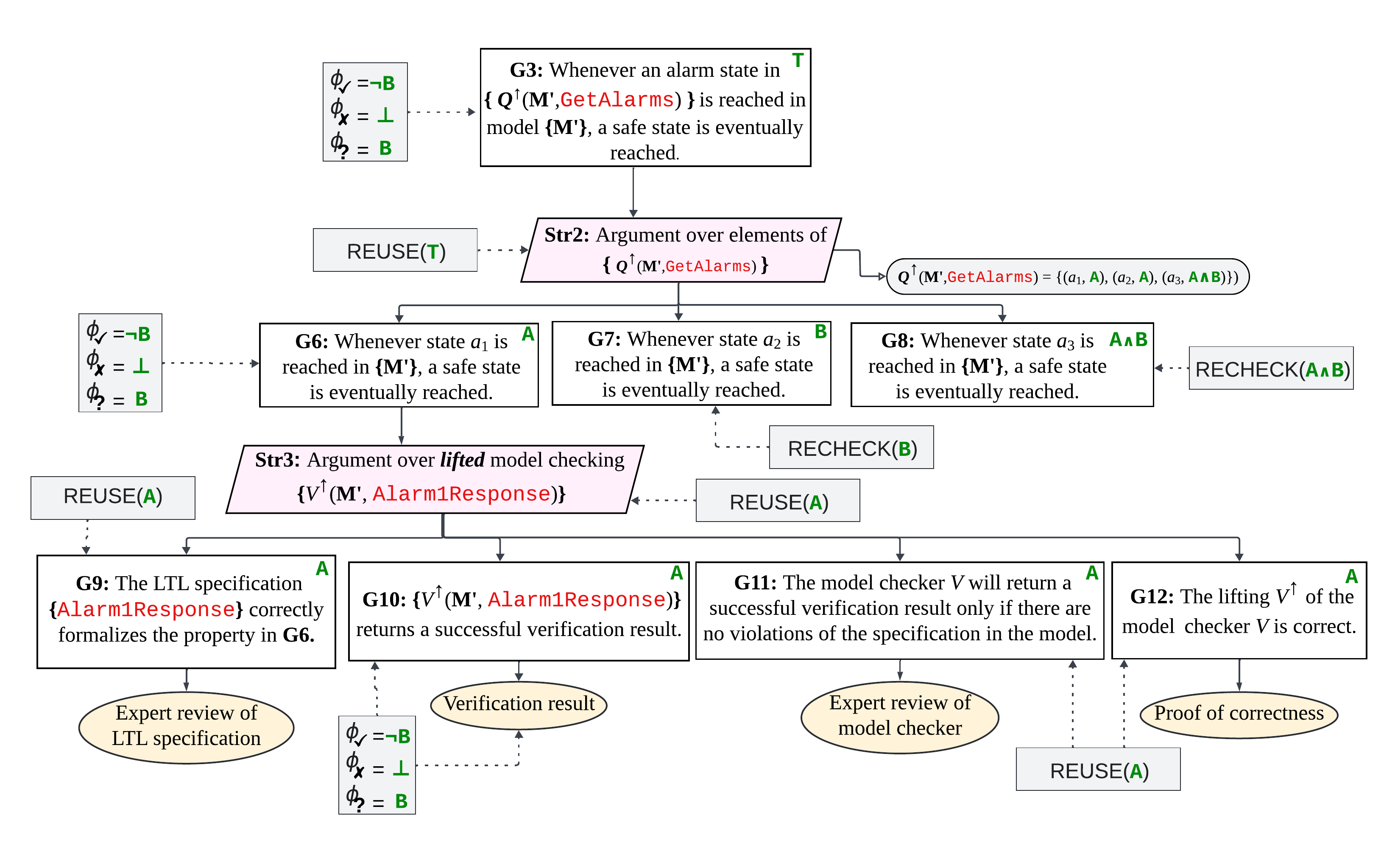}
    \caption{Fragment of the variational AC for the running example after applying lifted regression analysis (Example~\ref{ex:lif_reg_example}). Variability-aware regression labels shown in grey boxes.}
    \label{fig:lift_regression_example}
    \Description[]{}
\end{figure} 
\nc{Recall the variational AC produced in Example~\ref{ex:analyticVar} through the lifted query and lifted enumeration templates, as illustrated in Fig.~\ref{fig:lift_query_example}. Suppose that we supported each of its undeveloped goals through the lifted model checking template shown in Fig.~\ref{fig:model_check_ac}\weakNew{, and that all undeveloped goals (asserting the validity of specifications, analyses, and lifting) have been fully supported}. Now, suppose that we modified the product line model $\pl{M}$ in such a way that $\Delta(\pl{M}) = \token{B}$, \typofix{i.e,}{i.e.,} only configurations with feature $\token{B}$ were modified. An example outcome of applying the lifted regression analysis procedure is shown in Fig.~\ref{fig:lift_regression_example}. For brevity, we only show the fragment of the AC {beginning from the enumeration template over elements of the lifted query}{immediately following the re-instantiation of the lifted query template, which produced goal \textbf{(G3)}}. In this example, the modification of feature $\token{B}$ does not introduce any additional alarms \weakNew{in the evolved product line model $\pl{M}'$}, so argument \weakNew{\textbf{(Str2)}} does not regress for any configuration, and is  \weakChange{marked by $\overline{\cmark}$}{annotated with $\mathsf{REUSE}(\top)$ accordingly}. Continuing through the branch pertaining to the alarm  $\weakNew{a_1}$, the analytic template for the model checking \weakChange{is correct-by-construction, so it is also annotated by $\cmark$}{is \emph{not} re-instantiated, per the lazy evaluation strategy. However, its subgoals are still updated to refer to $\pl{M}'$ (to ensure the correctness criterion holds), and we can verify that there are no missing subgoals under $\pl{M}'$, so \textbf{Str3} is annotated by $\mathsf{REUSE}(\token{A})$}. Considering now the evidence for the model checking template, goals \textbf{G9}, \textbf{G11}, \textbf{G12} are not modified and remain supported for all configurations \weakNew{with feature $\token{A}$}. For the evidence from the lifted model checker, suppose in this case that we do \emph{not} have a corresponding lifted regression analysis. Then, as per Alg.~\ref{alg:varEvdRegression}, we mark the goal \weakNew{\textbf{G10}} with $\phi_{\textbf{?}} = \token{B}$, indicating that we cannot determine whether the modified configurations still satisfy the property, and $\phi_\cmark = \neg {\tt B}$, since the unmodified configurations can still use this evidence. \weakChange{This annotation is propagated as-is to goal $\textbf{G6}$, since all the other subgoals are $\overline{\cmark}$}{When we compose this regression value with $\mathsf{REUSE}(A)$, we obtain the same regression value assigned to \textbf{G10}, and assign it to \textbf{G6}}. \weakChange{Assuming that goals $\pl{g}_1$ and $\pl{g}_3$ were assigned the same regression value (since they apply the same template), we again propagate this annotation as-is back to the parent goal $\pl{g}$.}{When analyzing regression of \textbf{G7} and \textbf{G8}, both of these goals pertain to configurations contained within $\llbracket \widehat{\Delta}(\pl{M})\rrbracket = \llbracket \token{B}\rrbracket$, so (assuming we follow the same procedure as for \textbf{G6}), these goals are both annotated as $\mathsf{RECHECK}$ for their respective sets of configurations.} \weakNew{Finally, when we compose the regression values from \textbf{G6}, \textbf{G7}, \textbf{G8} and \textbf{Str2}, we are once again left with the same regression value $\langle \neg {\tt B}, \bot, {\tt B}\rangle$, which we assign to \textbf{G3}.} {Thus, the overall outcome of the analysis is that assurance has been preserved for configurations without feature $\token{B}$, and needs to be rechecked for those with feature $\token{B}$}. But more usefully, we can pinpoint this need for potential revision specifically to the model checking evidence; we know, through the semantics of the templates used to generate the AC, that both strategies \weakNew{\textbf{Str2}} and \weakNew{\textbf{Str3}} remain sound \weakNew{across their respective configurations} \typofix{through}{throughout} the evolution.}
\end{example}

\begin{theorem}\label{thm:backwardsLift}
    \weakNew{Let $\pl{A} \in \Var{\AC}$ be a variational AC, and let $\pl{A}_R$ be the reusable core extracted from $\pl{A}$ following the (lifted) forward pass of the analysis under variational evolution set $\widehat{\Delta}$. Let $\phi = \token{pc}(\pl{A}_R)$. Then for every $\conf \in \llbracket \phi \rrbracket$, we have} 
    $$\mathsf{BackwardPass}^\uparrow(\widehat{\Delta},\pl{A}_R)|_\conf = \mathsf{BackwardPass}(\widehat{\Delta}|_\conf,\pl{A}_R|_\conf)$$
    \weakNew{That is, Alg.~\ref{alg:lifted_assurance_regression} correctly lifts Alg.~\ref{alg:product_regression}}.
  
\end{theorem}

 {Theorems~\ref{thm:forwardsLift} and~\ref{thm:backwardsLift} give us our \weakChange{desired}{final} result: that the lifted forward pass (Alg.~\ref{alg:instantiationLift}) and lifted backward pass (Alg.~\ref{alg:lifted_assurance_regression}) provide a sound variability-aware regression analysis. The correctness of the lifted procedures allows us to ``lift'' the correctness of the product-based analysis (Thm.~\ref{thm:instAlgCorrect}).}

\begin{corollary}
    \label{thm:final}
    \weakNew{Let $\pl{A} \in \Var{\AC}$ and let $\widehat{\Delta}$ be a variational evolution set for the models referred to in $\pl{A}$, and assume $\Supp^\uparrow(\pl{A})$. Let $\pl{A}_R$ be the reusable core extracted from $\pl{A}$ following the lifted forward pass (Alg.~\ref{alg:instantiationLift}) under $\widehat{\Delta}$, and let $\pl{v} = \langle \phi_\cmark, \phi_\xmark, \phi_?\rangle$ be the variability-aware regression value over $\phi = \token{pc}(\pl{A})$ returned by running the lifted backward pass (Alg.~\ref{alg:lifted_assurance_regression}) on $\pl{A}_R$. Then for each $\conf \in \llbracket \phi \rrbracket$:}
    \begin{itemize}
        \item \weakNew{If $\conf \vDash \phi_\cmark$, then $\Supp(\pl{A}_R|_\conf)$ holds.}
        \item \weakNew{If $\conf \vDash \phi_\xmark$, then $\Supp(\pl{A}_R|_\conf)$ does not hold.
        \item If $\conf \vDash \phi_?$, then we make no determination about whether $\Supp(\pl{A}_R|_\conf)$ holds.}
    \end{itemize}
\end{corollary}
\begin{proof}
\weakNew{Let $\conf \in \llbracket \phi \rrbracket$. Since $\Supp^\uparrow(\pl{A})$ holds by assumption, we have $\Supp(\pl{A}|_\conf)$ in particular. Since composition of lifted functions produces lifted functions, we know that composing the lifted forward pass (Thm.~\ref{thm:forwardsLift}) and lifted backward pass (Thm.~\ref{thm:backwardsLift}) produces a lift of the entire product-based regression analysis described in Sec.~\ref{sec:product_regression}. The lifted regression analysis thus inherits the correctness of the product-based regression analysis (Thm.~\ref{thm:instAlgCorrect}).}
\end{proof}

\subsubsection{Handling Variability Evolution}\label{sec:varEvo}
\change[vartype2]{In the preceding discussion, we have operated under the assumption that the product line evolution was purely structural, i.e., the old and new product models were all defined with respect to the same features $F$ and feature model $\Phi$. It is thus natural to ask what changes in our regression analysis when features are removed or added, or when the feature model is modified. Our regression analysis is designed to analyze the assurance of configurations which existed in the product line \emph{prior} to the evolution -- any configurations which are only deemed valid after an evolution cannot be analyzed for regression, since they were never assured in the first place. Nevertheless, we want the interpretation of our regression analysis -- and more specifically, the regression annotations -- to be done relative to the variability model of the evolved product line, not the old one.
To this end, given an evolution which modifies the variability model, we want to begin by computing two sets of configurations:  $\Phi_{\tt Reuse}$, denoting configurations which were valid in the old product line and \emph{remain} valid in the new product line, and $\Phi_{\tt New}$, which are configurations only valid in the new product line. The precise computation of these sets depends on the details of the evolution; if we only modify the feature model $\Phi \mapsto \Phi^\prime$ without changing the alphabet of features, we can take $\Phi_{\tt Reuse} = \Phi \land \Phi^\prime$, and $\Phi_{\tt New} = \neg \Phi \land \Phi^\prime$. If we also add a new feature $f$ to $F$, then we need to ``exclude'' this feature accordingly, i.e., $\Phi_{\tt Reuse} = \Phi \land \neg f \land \Phi^\prime$, $\Phi_{\tt New} = (\neg \Phi \lor f) \land \Phi^\prime$. 

Having computed $\Phi_{\tt Reuse}$ and $\Phi_{\tt New}$, when we perform the lifted regression analysis, we can \emph{restrict} our product line models (and the variational AC) to only those configurations in $\Phi_{\tt Reuse}$, since we do not care about the regression of configurations which are no longer valid. We then need to consider what we can say about configurations in $\Phi_{\tt New}$. As we have seen in the regression analysis, there are some assurance entities -- such as propositional goals, and analytic templates -- which never regress due to a change in a model (see Alg.~\ref{alg:lifted_assurance_regression}, Line 7, and Alg.~\ref{alg:instantiationLift}, Line 18). These entities remain supported for all configurations, even those in $\Phi_{\tt New}$, so they can remain annotated by $\overline{\cmark}$. Conversely, entities which are \emph{contingent} on product line models -- such as verification evidence -- have no support for configurations in $\Phi_{\tt New}$. Thus, we can extend each of their annotations' $\xmark$-component with $\Phi_{\tt New}$, reflecting the fact that none of these configurations have any assurance, and restrict  their $\cmark$-component and $\textbf{?}$-component with $\neg \Phi_{\tt New}$ to reflect that \emph{only} the  configurations which were already valid in the previous feature model have (or potentially have) assurance.
}{At the start of this section, we noted that an evolution of a product line model may also modify its feature model, or even the language of features over which the feature model is expressed. The lifted regression analysis we have just defined only considers the regression between two (sets of) product line models over the same features and feature model. Let us begin generalizing our results to the more general kinds of product line evolutions. For the sake of simplicity, let us assume that there is only a single product line model $\pl{M}$ in the AC for which regression is being analyzed.}

\newText[vartype3]{First, let us consider the case where $\pl{M}$ is a product line with feature model $\Phi$, and it is evolved into $\pl{M}^\prime$ with feature model $\Phi^\prime$. Suppose we have a variational AC $\pl{A}$ for $\pl{M}$ for which we have established support across the entire product line, i.e., we have shown ${\Supp^\uparrow(\pl{A})}$. We now want to compute a regression value $\pl{v}$ partitioning $\Phi^\prime$ which tells us for which configurations of the new product line model $\pl{A}$ remains supported. We can do this as follows. First, we partition the new feature model $\Phi^\prime$ into two disjoint sets of configurations: $\Phi_{\tt New} = \Phi^\prime \land \neg \Phi$, representing those configurations which are present \emph{only} under the new feature model, and $\Phi_{\tt Reuse} = \Phi^\prime \land \Phi$, representing those which are ``kept'' from the previous feature model. We then apply the lifted regression analysis defined above to $\pl{A}$, under the \emph{restricted} feature model $\Phi_{\tt Reuse}$. The result is a variability-aware regression value $\pl{v} = \langle \phi_\cmark, \phi_\xmark, \phi_?\rangle$ partitioning ${\Phi}_{\tt Reuse}$. To then extend this regression value to cover all of $\Phi_{\tt New}$, we modify $\pl{v}$ as $\pl{v}^\prime = \langle \phi_\cmark \wedge \neg \Phi_{\tt New}, \phi_{\xmark} \wedge \vee \Phi_{\tt New}, \phi_{?} \wedge \neg \Phi_{\tt New} \rangle $. Note that this is equivalent to taking $\pl{v} \otimes \langle \bot, \Phi_{\tt New}, \bot \rangle $, and simply reflects the fact that assurance of any configurations in $\Phi_{\tt New}$ is completely unsupported in $\pl{A}$.

Finally, let us consider what happens when, in addition to modifying the feature model, we modify the language of features (e.g., by adding or removing a feature). The case removing a feature is not very interesting, since we are strictly reducing the set of configurations covered by the feature model; let us focus instead on the case in which we add a feature $f$ as part of the new feature model $\Phi^\prime$. The only tangible difference with the preceding the construction is the definition of the feature models $\Phi_{\tt Reuse}$ and $\Phi_{\tt New}$. In particular, when we are analyzing regression of the existing AC $\pl{A}$, we need to ``exclude'' feature $f$ from consideration, so we define ${\Phi_{\tt Reuse}} = \Phi \land \Phi^\prime \land \neg f$. Conversely, when we define $\Phi_{\tt New}$, we must explicitly \emph{include} any configurations with $f$, i.e., $\Phi_{\tt New} = (\neg \Phi \vee f) \wedge \Phi^\prime$. }
Note that $\Phi_{\tt New}$ and $\Phi_{\tt Reuse}$ always partition $\Phi^\prime$.

% \weakDelete{}
\begin{example} \label{ex:prod_reg}
   \nc{Recall the annotated variational AC produced through the regression analysis described in Example~\ref{ex:lif_reg_example} and shown in Fig.~\ref{fig:lift_regression_example}. Suppose that, in addition to modifying feature $\token{B}$, we introduced a new feature $\token{C}$, with its own states and transitions. For simplicity, suppose we do not otherwise modify the feature model. We then have $\Phi_{\tt Reuse} = \Phi \land \neg \token{C}$, $\Phi_{\tt New} = \Phi \land \token{C}$. 
   To account for this, each annotation of a non-propositional goal or non-analytic template in the AC would be \weakChange{extended in its $\xmark$-component with feature expression $\token{C}$, and we would accordingly restrict its $\cmark$-component with $\neg C$}{composed with $\langle \bot, \token{C}, \bot \rangle$, interpreted with respect to the set of configurations for each each goal is relevant (e.g., $\llbracket \token{A} \rrbracket$ in the case of \textbf{G10})}. \weakChange{For instance,}{In this case,} the annotations on \weakChange{leaf goal}{goals} \weakChange{$\pl{g}_{out}$}{\textbf{G10}}, \weakChange{$\pl{g}_3$}{\textbf{G6}} and root goal \weakChange{$\pl{g}$}{\textbf{G3}} \weakNew{each} become $\langle \token{\neg B} \land \neg \token{C}, \token{C}, \token{B} \land \neg \token{C} \rangle$. \weakNew{This reflects the fact that (i) assurance has not regressed (i.e., can be reused) for configurations without the {modified} feature $\token{B}$ \emph{or} the \emph{new} feature $\token{C}$; (ii) there is \emph{zero} assurance for any configurations with the new feature $\token{C}$; (ii) and configurations with the modified feature $\token{B}$ but without the new feature $\token{C}$ need to be rechecked. Note, however, that we can avoid re-annotating \emph{propositional} goals, e.g., \textbf{G11} and \textbf{G12}, since they are not affected by any changes in the product line model. That is, \textbf{G11} and \textbf{G12} still remain annotated by $\mathsf{REUSE}(\token{A})$, including for those configurations which have the new feature $\token{C}$. \weakDelete{and goal $\pl{g}_1$ becomes annotated by $\langle \neg \token{C}, \token{C}, \bot\rangle$.  However, the non-contingent nodes, such as the analytic template over model checking and propositional goals, remain  \weakChange{$\overline{\cmark}$}{$\mathsf{REUSE}(\Phi_{\tt Reuse})$}}}}
\end{example}

\nc{In this section, we defined a variability-aware regression analysis for product line ACs by lifting the product-based regression analysis defined in Sec.~\ref{sec:product_regression}. By using this lifted analysis, engineers are able to determine not only which parts of the AC need to be revised following an evolution of the product line, but also for which specific configurations this revision needs to occur. This is particularly useful in contexts where the deployment of evolved systems is time-sensitive; based on the outcome of the lifted regression analysis, the assurance team can approve the deployment of the evolution for configurations which have not lost their assurance (represented by $\phi_\cmark$), and postpone deployment of the evolution for the remaining configurations until their assurance has been re-verified.}

\section{MMINTA-PL: An Assurance Case Framework for Evolving Product Line Models}
\label{sec:tool}
{Any integration of formal methods as part of AC development should include extensive tool support, as many AC developers are not formal methods experts. {To this end, we have developed tool support for lifted AC development and regression as part of an Eclipse-based model management framework. Our tool aims \typofix{provides}{to provide}  functionalities modeled on the framework outlined in Secs.~\ref{sec:lifting} and ~\ref{sec:liftedRegression}. {In this section, we  provide a brief overview of our tool (Sec.~\ref{sec:toolDesc}), and demonstrate its capabilities on a small case study (Sec.~\ref{sec:casestudy}).}}}

% \vspace{-0.15in}

\subsection{A Model Management Tool for Product Line ACs}
\label{sec:toolDesc}
% \vspace{-0.15in}
{\emph{MMINT}\footnote{\url{https://github.com/adisandro/MMINT}} is an Eclipse-based model management framework developed at the University of Toronto. It is a generic framework which can be extended with plugins for specific modeling tasks. One of its extensions, \toolName~\cite{minta}, is used for model-driven AC development, supporting GSN modeling and model-based analyses. Another extension, \toolNamePL, supports product line modeling and lifted model-based analyses~\cite{di2023adding}.}

{Several functionalities of \toolName~ and \toolNamePL~ were reused directly to support lifted AC development. We combine \toolName's GSN metamodel~\cite{minta}, and
\toolNamePL's generic variational metamodel (GVM)~\cite{di2023adding}) to define a metamodel for product lines of GSN ACs. We also needed to implement some new functionalities: (i) we extended \toolName's GSN template module to recognize product line models, such that instantiation of GSN templates can be done with either product-level or lifted versions; (ii) users can define product-level analytic templates as formalized in Sec.~\ref{sec:analyticTemplates}, allowing the results of a specified analysis to be weaved into an AC as part of template instantiation. When these templates are instantiated on variational models, if the analyses associated with the template have been lifted, the lifted analyses are executed, and the lifted template is instantiated instead; (iii)\typofix{:}{} \toolNamePL's GVM is unable to provide appropriate visualizations for arbitrary product line models, so we created a custom visualization module to facilitate the inspection of PL ACs.}

\nc{To support lifted AC regression, we first implemented an extensible product-based regression framework for GSN ACs, following the components and procedures described in Sec.~\ref{sec:product_regression}. We extended \toolName's GSN template module to accommodate template-based regression, together with a baseline regression analysis for informal strategies. We then implemented the lifted regression procedure on top of \toolNamePL's GVM, as described in Sec.~\ref{sec:liftedRegression}. Users can specify how to use a lifted analysis within a GSN template to react to system evolution, i.e., whether a template should be re-instantiated eagerly upon a system evolution. The implementation follows the overall procedure described in Sec.~\ref{sec:liftedRegression}, with one exception: when analyzing an evolution which includes a modification of the feature model, whenever we re-execute a template during the regression analysis, we also apply it to the configurations which are introduced by the evolution. This is not necessary for regression analysis, strictly speaking, but is a convenient way of determining how the AC will need to be repaired to accommodate new products. }
% \vspace{-0.1in}
\subsection{Case Study}
\label{sec:case_study}

% \vspace{-0.1in}

% \vspace{-0.15in}
\label{sec:casestudy}
% \vspace{-0.15in}
{To demonstrate the feasibility of the lifted AC development and the features of our AC development tool, we followed the steps outlined in Fig.~\ref{fig:workflows} to create a partial AC for a product line of medical infusion pumps.\footnote{Artifacts at \url{https://github.com/adisandro/MMINT//\#using-mmint}} The structure and semantics of this AC correspond closely to the running example used in Secs.~\ref{sec:lifting}-~\ref{sec:liftedRegression}.}
% \vspace{-0.2in}
\paragraph{System Details.}
{Infusion pumps are devices used to administer medication or other fluids to patients. As different patients may have different medical needs, it is natural to model an SPL of infusion pumps with different optional features.}
{We began from an existing Extended Finite State Machine (EFSM) model of an infusion pump created as part of a multi-institute research project~\cite{alur2004formal}. While this model was not originally defined as an SPL, it was designed to model features and hazards for an infusion pump \emph{in general}, the authors noting that in general ``no single device [...] has all of the design features''~\cite{zhang2010hazard}. We extended the EFSM to an SPL by mapping {five} optional features to their associated states and transitions and annotating these elements with presence conditions. For example,}
{\token{CHECK\_INFUSION\_RATE} is an optional feature that allows a pump to monitor the current rate of delivery of a drug.} 
{{The resulting SPL encompasses a family of {20} valid product configurations.}  In our assurance scenario, the top-level assurance obligation (i.e., the root node of the AC) is to show that when an alarm is triggered (e.g., due to a dosage limit violation), the system will \typofix{not} immediately halt drug administration (i.e., will not be administering a dose immediately following the alarm).}

\paragraph{Step 1: Defining the Assurance Process.}
{Before beginning AC development, we need to determine which types of analyses and argumentation are applicable for our assurance task. For this case study, we considered two kinds of analyses: querying of models (via the Viatra Query Language~\cite{vql}) and model checking~\cite{baier2008principles}. As part of the assurance plan, we require that the use of these analyses be accompanied by sufficient assurance that the models and specifications used for analyses have been validated.}
% \vspace{-3mm}

\paragraph{Step 2: Formalization.}
{Based on the assurance plan defined in Step 1, we can now formalize the associated argument structures. The assurance obligations associated with using queries and model checking can be formalized as analytic templates, as described in Sec.~\ref{sec:analyticTemplates}.  We also employ the enumeration template formalized as described in Example~\ref{ex:enumDecomp}.}

% \vspace{-3mm}
\paragraph{Step 3: Lifting.}
{We now lift the analyses and templates formalized in Step 2 so that they can be applied directly to SPLs. For lifted model queries, we reused a lifted query engine developed as part of previous work~\cite{di2023adding}. For lifted model checking, we used the tool FTS4VMC~\cite{ter2022fts4vmc} which verifies FTSs against specifications written in the action-based branching time logic v-ACTL\newText[actlref]{~\cite{ter2016modelling}. The language v-ACTL is based on ACTL~\cite{de1990action}, itself an extension of Computation Tree Logic~\cite{baier2008principles} with actions.} The templates formalized for both analyses can be lifted via the construction given in Def.~\ref{def:analyticTemplateLift}. We can also employ the lifted enumeration template formalized in Example~\ref{ex:enumDecompLift}. }

\paragraph{Step 4: Lifted AC Development}

\begin{figure}[t]
    \centering
\includegraphics[width=0.8\textwidth]{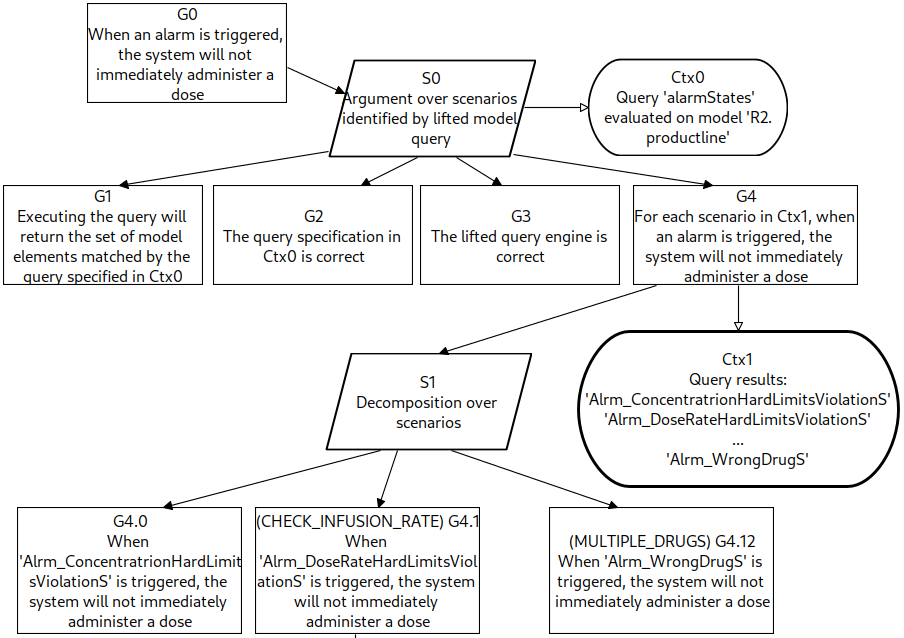}
% \vspace{-0.1in}
    \caption{Instantiations of the lifted query and scenario decomposition templates.}
    \label{fig:query_template}
    % \vspace{-0.25in}
    \Description[]{}
\end{figure}
{Having defined, formalized and lifted the analyses and templates for our task, we can perform lifted AC development. We begin with the root goal asserting the safety property we want to assure (not administering drug doses immediately after an alarm is raised). Under normal circumstances, we could proceed by running a model query which  returns all alarm states in the model, and then assure the property for each alarm scenario. Thanks to the lifting of the query engine, we can apply the same rationale using our lifted analyses and argument templates, even though different products may be associated with different sets of alarm states. Fig.~\ref{fig:query_template} shows the results of instantiating the lifted query template followed by a lifted enumeration. When we instantiate the generic model query template, the tool automatically detects that the model is an SPL, and executes the lifted query engine to return a variability-aware set of query results (\token{Ctx0}). These results are then woven into the AC using a lifted version of the analytic template. This lifted analytic argument instantiation is sound as per Thm.~\ref{thm:liftAnalyticCorrect}.  {Returning to the AC, we proceed from goal \token{G4} using a variational domain decomposition, assigning each alarm state to its own goal, such that the presence conditions identified by the lifted query result (e.g., \token{CHECK\_INFUSION\_RATE}) are used to annotate each subgoal}. {The absence of a presence condition (e.g., \token{G4.0}) means that the alarm is present in every product.} This decomposition corresponds to the lifted enumeration template and is thus also sound by Prop.~\ref{prop:liftEnum}.}

{We can then continue to produce assurance for each identified alarm scenario in a lifted fashion. We focus on goal \token{G4.1}, which effectively asserts that every product with feature \token{CHECK\_INFUSION\_RATE} satisfies the given safety property in the context of alarms due to dose rate violations. This can be verified using (lifted) software model checking. As with queries, we can instantiate the model checking template formalized in Step 2 and lifted in Step 3, using the lifted model checker FTS4VMC for SPL-level verification. We formalize the property in \token{G4.1} in v-ACTL as 
$$\textbf{AG} [\token{Alrm\_DoseRateHardLimitsViolationS}] \textbf{AX}(\neg \token{Infusion\_NormalOperationS})$$
\change[vatlexpl]{and run the}{i.e., whenever $(\textbf{AG})$ a transition occurs corresponding to an alarm being raised due to a dose rate limit violation, across all possible successor states (\textbf{AX}), none of them are a state in which a drug is being infused.} lifted model checker on the infusion pump FTS model. In this case, the model checker does not reveal any violations, meaning that every product with \token{CHECK\_INFUSION\_RATE} satisfies the given property. This can then be incorporated as variational evidence for this family of products, as shown in Fig.~\ref{fig:model_check_ac}. {Evidence for \token{G6}, that the formalization is correct, still needs to be produced. Note that evidence for \token{G5} and \token{G8} needs to be produced {only once} and can be reused in subsequent applications of the template.} Once we have provided the required evidence for this argument (and the analogous evidence for the lifted query), we can then repeat this verification process for each alarm scenario until all assurance obligations have been satisfied.}

\begin{figure}[t]
    \centering
\includegraphics[width=0.8\textwidth]{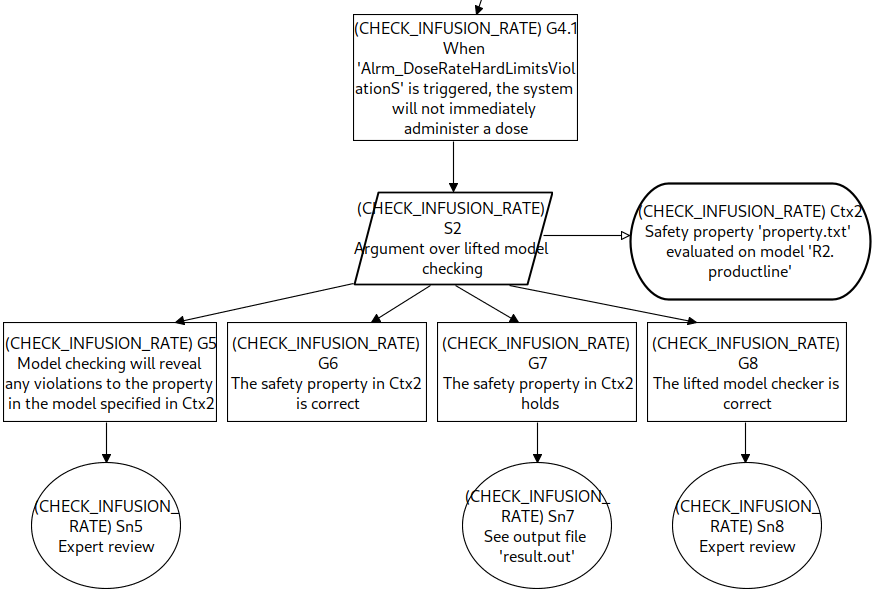}
% \vspace{-0.1in}
    \caption{Tool-generated lifted analytic model checking template instantiation for claim \token{G4.1}.}
    \label{fig:model_check_ac}
    % \vspace{-0.25in}
    \Description[]{}
\end{figure}

\subsubsection{Evolving the Product Line}
\label{sec:evolutionSpec}

\nc{We now consider a scenario in which the infusion pump product line evolves through the addition of a new feature and the modification of an existing feature. We add the new feature \token{PROGRAMMABLE\_INFUSION} to the product line, which introduces the functionality for users to customize the settings of an infusion. Specifically, this feature gives users the ability to modify the rate at which fluids are delivered during an infusion and to set the total volume of fluid to be infused. When a new rate is set, the system checks whether the new values fall within a predetermined \textit{acceptably safe} threshold. If an unsafe value is detected, the system enters a new alarm state, \token{Alrm\_UnsafeNewRateS}. This alarm can be cleared by setting new safe parameters or stopping the infusion. Adding this feature also adds five new states (one of which is an alarm state) and 35 new transitions annotated with the \token{PROGRAMMABLE\_INFUSION} feature to the model.} 

\nc{We also modify the \token{VISUAL\_DISPLAY} feature by adding the functionality for patients and doctors to review the history of past infusions administered by the pump and the settings used in each. This becomes a baseline functionality for all infusion pumps with a visual display, leading to the addition of three new states and six new transitions annotated with the \token{VISUAL\_DISPLAY} feature. This modification is restricted to configurations of the product line with this feature enabled. }

\nc{The \token{PROGRAMMABLE\_INFUSION} feature is only allowed to be used in infusion pumps that are able to check the settings of an infusion and display them to the user (i.e., those with the \token{CHECK\_INFUSION\_RATE} feature and the \token{VISUAL\_DISPLAY} feature). This is to ensure that users are able to confirm that they set their infusion settings correctly. Consequently, the feature model for our evolved product line is given as follows: }
\begin{align*}
\token{HW\_MONITORING \land (MULTIPLE\_DRUGS \implies (CHECK\_DRUG\_TYPE \land VISUAL\_DISPLAY))} \\
\land ~\token{(PROGRAMMABLE\_INFUSION \implies (CHECK\_INFUSION\_RATE \land VISUAL\_DISPLAY))}
\end{align*}
\nc{There are 25 configurations that satisfy this feature model.}
\nc{In summary, the product line evolved by adding the new \token{PROGRAMMABLE\_INFUSION} feature and modifying the existing \token{VISUAL\_DISPLAY} feature. This resulted in the addition of seven new non-alarm states, one new alarm state, 41 new transitions and a modified feature model.}

\subsubsection{Lifted Assurance Case Regression}

\begin{figure}[t]
    \centering
    \includegraphics[width=\linewidth]{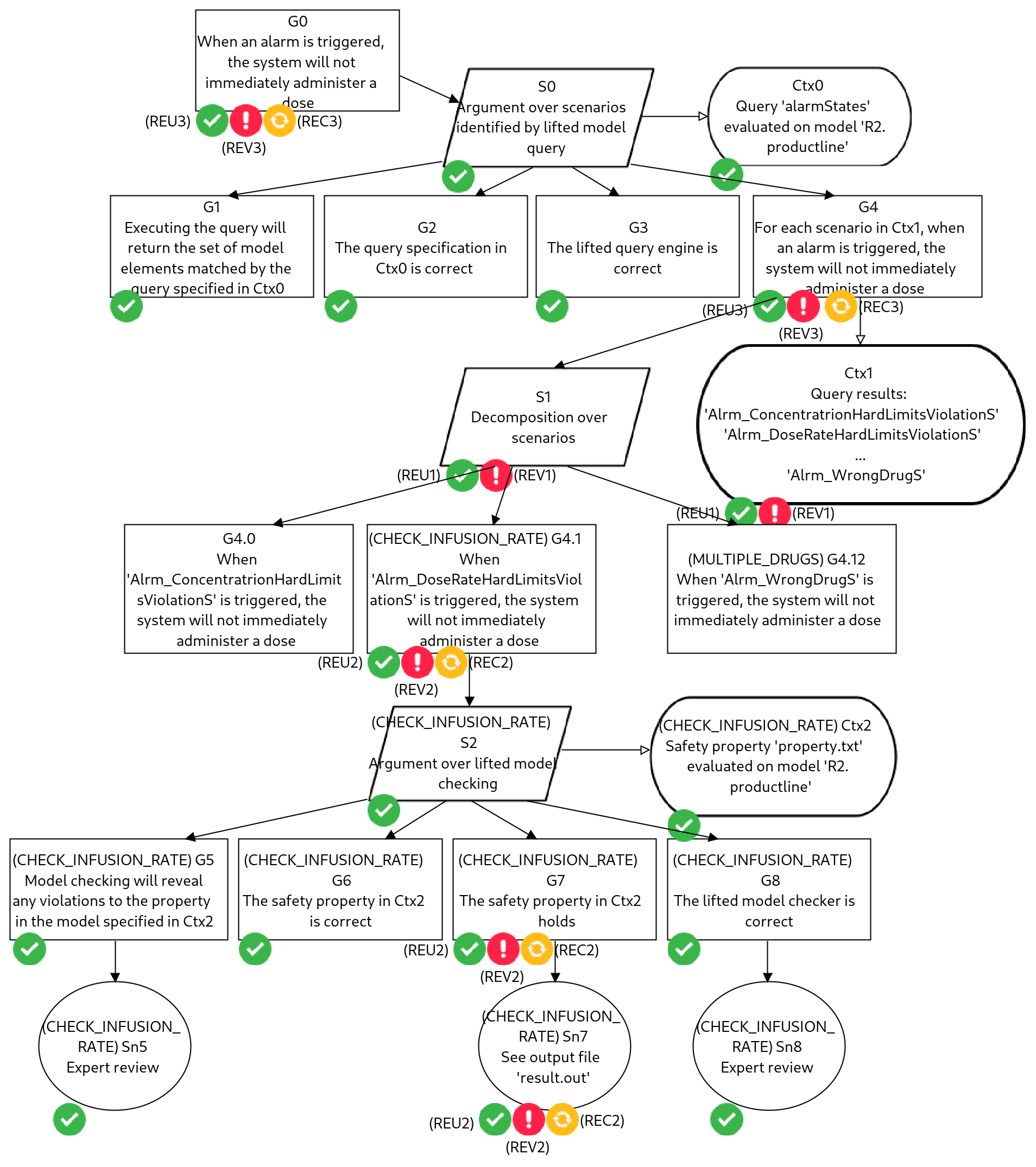}
    \caption{Outcome of lifted regression analysis for the infusion pump AC.}
    \label{fig:tool_regression}
    \Description[]{}
\end{figure}

\nc{Given the above evolution scenario, we execute the automated regression analysis that annotates the PL AC with variability-aware regression values, as shown in Fig.~\ref{fig:tool_regression}. The green regression symbol corresponds to $\phi_\cmark$, the yellow to $\phi_{\textbf{?}}$, the red to $\phi_\xmark$. We use the annotation identifier prefixes ${\tt REU}$ (for \emph{reuse}, i.e., $\phi_\cmark$), ${\tt REV}$ (for \emph{revise}, i.e., {$\phi_\xmark$}), and ${\tt REC}$ (for \emph{recheck}, i.e., $\phi_{\textbf{?}}$).}

\nc{Before we begin the regression analysis, given the evolution as described in Sec.~\ref{sec:evolutionSpec}, we compute variational impact set ${\ntn{\widehat{\Delta}}}$ for the infusion pump model $\pl{M}_{IP}$ as ${\ntn{\widehat{\Delta}}}(\pl{M}_{IP}) = \token{VISUAL\_DISPLAY}$, i.e., among \emph{existing} configurations, only those with this feature have been modified. We also compute the set of \emph{new} configurations, induced by the introduction of the feature ${\tt PROGRAMMABLE\_INFUSION}$ and associated revision of the feature model, as}
\begin{align*}
    \Phi_{\tt New} = ~&\token{CHECK\_INFUSION\_RATE}~\land~\token{HW\_MONITORING}~\land~\token{PROGRAMMABLE\_INFUSION}\\&~\land~\token{VISUAL\_DISPLAY}~\land~(\token{MULTIPLE\_DRUGS} \implies \token{CHECK\_DRUG\_TYPE}).
\end{align*}
Beginning with the forward pass of the regression procedure, we start with the top level goal \token{G0}, reaching the lifted analytic template for querying. We run the same lifted query on the evolved system model, returning an updated variability-aware set of alarm states. We find that{\ntn{, with respect to the configurations $\Phi_{\tt Old}$,}} all alarms in the AC are still present in the evolved model, {\ntn{and no new states have been introduced for configurations in $\Phi_{\tt Old}$}.} \change[]{so we mark the enumeration strategy as}{Accordingly, if we were \emph{only} concerned with configurations of $\Phi_{\tt Old}$, we could annotate the enumeration strategy as} \change[]{$\overline{\cmark}$}{$\mathsf{Reuse}(\Phi_{\tt Old}) = \langle \Phi_{\tt Old}, \bot,\bot \rangle $} and continue the forward pass into the AC subtrees. However, we also find that there is a new alarm introduced through the $\token{Programmable\_Infusion}$ {{\ntn{feature}}}, which is not currently assured -- this is recorded in the AC to facilitate reuse of the argument following the regression analysis. Thus, the only problem with the strategy is that it is as of yet incomplete for the \typofix{configruations}{configurations} in $\Phi_{\tt New}$, so the annotation for \token{S1} and \token{Ctx1} is computed as $\langle \token{REU1}=\neg \Phi_{\tt New}, \token{REV1}=\Phi_{\tt New}, \token{REC1}=\bot\rangle$. {\ntn{Note that this annotation can be derived from ${\mathsf{REUSE}(\Phi_{\tt Old})} \otimes \langle  \bot, \Phi_{\tt New}, \bot \rangle = \langle \Phi_{\tt Old} \wedge \neg \Phi_{\tt New}, \Phi_{\tt New}, \bot\rangle$, and simplifying with respect to $\Phi$ given that $\Phi_{\tt Old} \implies \Phi$.}}

\nc{We continue the forward pass into the sub-ACs for each type of alarm scenario. For the sake of brevity, we show the full detail only for goal \token{G4.1} for the alarm \token{Alrm\_DoseRateHardLimitsViolationS}. The subtree is populated by the template for lifted model checking. We start the backward pass by checking whether the lifted model checker has a regression analysis available. It does not, so we know that, among \emph{existing} configurations, the evidence is reusable for configurations satisfying $\neg \token{VISUAL\_DISPLAY}$, and requires re-checking for those containing $\token{VISUAL\_DISPLAY}$, per Line 6 of Alg.~\ref{alg:varEvdRegression}. When we also take into account the configurations which are introduced by \typofix{teh}{the} new feature model, we define the annotations of \token{Sn7}, as $\langle \token{REU2}=\neg \token{VISUAL\_DISPLAY} \land \neg \Phi_{\tt New}, \token{REV2}=\Phi_{\tt New}, \token{REC2}=\token{VISUAL\_DISPLAY} \land \neg \Phi_{\tt New}\rangle$. All the other evidence in this subtree is \typofix{annotaetd}{annotated} by {\ntn{$\mathsf{REUSE}(\token{CHECK\_INFUSION\_RATE})$}}, since they do not depend on the model, as per the semantics of the model checking template.}

\nc{As we propagate the annotation for $\token{Sn7}$ back up the AC, it remains unchanged until we reach $\token{S1}$. We now need to compute the regression value for $\token{G4}$ per Alg.~\ref{alg:lifted_assurance_regression}, which, after simplifying, gives us $\typofix{}{\langle}\token{REU3} = \neg \token{VISUAL\_DISPLAY}\land \neg \Phi_{\tt New}, \token{REV3} = \Phi_{\tt New}, \token{REC3} =\token{VISUAL\_DISPLAY} \land \neg \Phi_{\tt New}\rangle$. That is, the annotation from $\token{Sn7}$ subsumes the annotation from $\token{S1}$. This value is then propagated as-is to the root goal, since no other\typofix{sublings}{siblings} of $\token{G4}$ regress. Thus, to summarize, we can know that configurations which do not have feature ${\tt VISUAL\_DISPLAY}$, and which were already assured in the original AC, preserve all their assurance. More importantly, we can pinpoint exactly what evidence needs to be rechecked for the remaining configurations (and what those configurations are), and we know what is needed to repair strategy $S1$ to account for the new configurations induced by $\token{ProgrammableInfusion}$.}

\subsection{Discussion}

In developing the above AC fragment for the infusion pump SPL, we demonstrated the feasibility of lifted AC development to support multi-layered AC development in a lifted fashion. We emphasize two specific points. (1) {By following the workflow outlined in Sec.~\ref{sec:lifting:overview}, {the {operational process} of AC development becomes essentially the same as in product-level AC development}}. That is to say, it suffices to know what one would do to assure a single product, and the lifted analyses and templates can correctly generalize this knowledge to the SPL-level. This is due to the particular semantics we have given to variational ACs, as outlined in Sec.~\ref{sec:lifting}. (2) Through analytic variational argumentation, we are able to \emph{systematically identify {variation points} in the SPL which are relevant to the design of the AC}. For example, in our case study, the variation points identified by the lifted query led systematically to \emph{structural variability} in the AC (i.e., the subgoals of strategy \token{S1} in Fig.~\ref{fig:query_template}).

One of the limitations of the tooling is that modeling languages and model-based analyses must be defined natively in order to leverage the GVM. {For instance, the lifted model checker is not a natively defined model-based analysis, and its verification results are not automatically interpreted as variational data (this contrasts with querying, which is a native model-based analysis)}. {There is currently no automated support for template verification at the time of instantiation; we assume that all argument templates are verified manually ahead of time. We also have not implemented a general-purpose template regression procedure in the spirit of Alg.~\ref{alg:templateRegression}, instead relying on specialized regression procedures for individual templates, as we did for the combined query/enumeration template used to generate Fig.~\ref{fig:query_template}.}

{With respect to the validity of our observations, we note that our partial AC is relatively narrow in scope and was designed by the authors for the purpose of demonstrating lifted AC development and regression.
While we believe we have successfully demonstrated the feasibility of our lifted AC development and maintenance, further empirical validation is required. Ideally, this evaluation can be done as part of a collaboration with industrial assurance engineers, since real-world assurance processes may be more difficult to formalize and lift than those shown here.}
% % \vspace{-3mm}
\section{Related Work}
\label{sec:related}
\paragraph{Analysis of Software Product Lines.}
{Implementing scalable analyses of SPLs is one of the central problems in SPLE. Thum et al.~\cite{thum2014classification} divide SPL-level analysis strategies as either \emph{product-based} (e.g., brute-force or sampling-based techniques), \emph{feature-based} (i.e., analyzing feature modules independently and composing the results), or \emph{family-based} (i.e., operating at the level of sets of products). \newText[mcrefs]{The lifted model checker FTS4VMC and lifted Viatra query engine described in Sec.~\ref{sec:case_study} are examples of family-based analysis tools. Other family-based model checkers include SNIP~\cite{classen2014formal}, mentioned in Sec.~\ref{sec:liftAnalytic}, and ProFeat, which is a variability-aware extension of the probabilistic model checker PRISM. A technique for family-based model checking based on $\mu$-calculus has also been developed for the model checker mCRL2~\cite{ter2017family}.} Lifted analyses are a specific form of SPL-level analyses, which can be implemented using any of these three strategies. The regression analysis procedure defined in Sec.~\ref{sec:liftedRegression} is an instance of family-based analysis. A wide variety of analyses have been lifted in the SPLE literature, }e.g.~\cite{classen2014formal,thum2012family,kastner2012type,salay2014lifting,shahin2019lifting}.  {We have previously proposed a catalogue of lifted analyses, distinguishing various approaches to their design and implementation~\cite{MURPHY2025112280}. Our proposal for lifted AC development and regression analysis presupposes that lifting of assurance-relevant analyses is feasible. As we have demonstrated in our case study, existing lifted analyses, such as lifted model checking, can be used in lifted AC development; however, some assurance scenarios may require novel lifted analyses.}

\paragraph{Assurance Cases for Product Lines.}
{{{Assurance cases for SPLs were studied by Habli and Kelly~\cite{habli2009model,habli2010safety}, who argued that safety-relevant variation points in the SPL should be reflected explicitly in the product line AC. Habli and Kelly propose using GSN patterns and the modular GSN extension~\cite{habli2010safety} to represent product line ACs. Habli~\cite{habli2009model} also provides an SPL safety metamodel which allows for variation points in the SPL to be traced explicitly to a PL AC. This approach was further refined in \typofix{de}{De} Oliveira et al.~\cite{de2015supporting}, introducing tool-supported generation of modular PL ACs from feature-based system models and safety analyses. However, this approach does not provide a distinction between product-level and SPL-level semantics or analysis, implicitly assuming that analysis and reasoning are performed at the \typofix{product}{level of products}. As an alternative to product-based AC development, Habli~\cite{habli2009model} also considers a primarily {feature-based} AC development method, in which AC modules are developed independently for each feature, and then the  PL AC is obtained by composing these modules. However, this method requires the assurance engineer to either (a) identify and mitigate all potential feature interactions, which can \typofix{devolves into to the brute-forced}{devolve into brute-force} assurance engineering, or (b) tolerate an incomplete assessment of potential feature interactions. By contrast, lifted AC development supports analysis and reasoning over {all} valid configurations of the SPL, without resorting to product-level work. To the best of our knowledge, the only existing variability-aware AC development process was proposed by Nešić et al.~\cite{nevsic2021product}, which lifts contract-based templates to PLs of component-based systems~\cite{nevsic2021product}. Our proposal for lifted AC development is effectively an attempt to generalize the approach of Nešić et al. to {arbitrary} templates and analyses.}}}

\paragraph{Formal Methods for Assurance Case Development.}
{{The most obvious use of formal methods for AC development is for the production of evidence; well-established verification techniques such as model checking~\cite{baier2008principles} and deductive verification~\cite{leino2010dafny} provide invaluable evidence for formally specified requirements. As we have mentioned above, the integration of formal methods with AC development requires extensive tool support. The AC development tool AdvoCATE~\cite{Denney:2018} uses the AutoCert inference engine~\cite{denney2014automating} to check system software implementations against formal specifications, which are then grafted into the AC following a predefined template. The Evidential Tool Bus (ETB)~\cite{Cruanes:2013} gathers verification evidence from various external tools, and then creates ACs from the bottom up using compositional rules written in a variant of Datalog. Formal methods can also be used to ensure that instantiations of argument templates actually yield sound arguments.}
{As ACs and proofs are closely related, proof assistants have been used on several occasions to study rigorous AC development. Rushby~\cite{rushby2009formalism} demonstrated a proof-of-concept of how an AC could be modelled using the proof assistant PVS. The AC editor D-Case was extended with a translation to the Agda programming language, such that an AC could be specified as an Agda program to check for well-formedness~\cite{takeyama2011brief}. An extension to Isabelle was developed to embed the Structured Assurance Case Metamodel (SACM) as part of its documentation layer~\cite{foster2021integration}. Finally, Viger et al.~\cite{viger2023foremost} used Lean to study the correctness of decomposition templates for model-based ACs.}}

\paragraph{Change Impact Assessment and Regression for Assurance Cases.}
A number of different approaches for analyzing the impact of changes have been proposed, including syntactic pattern detection \cite{carlanEnhancingStateoftheartSafety2020}, the use of machine learning models trained on AC development data \cite{menghiAssuranceCaseDevelopment2023}, the use of natural language processing \cite{muramATTESTAutomatingReview2023}, and property-specific safety contracts over modular ACs \cite{wargContinuousDeploymentDependable2019} to identify impacted portions of ACs.
In terms of specifically model-based approaches,
Annable et al. \cite{annable2024comprehensive} and Cârlan et al. \cite{carlanAutomatingSafetyArgument2022} rely on existing traceability links between system and AC to provide both semantic and syntactic change impact assessment at the product level; however, the procedure defined in ~\cite{carlanAutomatingSafetyArgument2022} targets a specific AC template for machine learning components, while ~\cite{annable2024comprehensive} relies on manual inspection of the semantics of the AC and system to eliminate false positives. In comparison, our product-based assurance regression framework is generic for arbitrary formal templates, and defines the semantics of the assurance directly as part of the template.
Kokaly et al. \cite{kokalySafetyCaseImpact2017} also rely on a model-based approach to syntactically determine product-level AC elements to be marked as \textit{reuse}, \textit{recheck} or \textit{revise}.
While their proposed approach is complete, it is not sound (i.e., false positives may occur), although techniques to improve precision are discussed at a conceptual level.
At the level of product line assurance, 
Shahin et al. \cite{shahin2021towards} lifted the work from  \cite{kokalySafetyCaseImpact2017} to the product-line level. To do this, they model variational ACs as sets of AC elements annotated with presence conditions, without any internal structure or semantics. By contrast, our lifted regression analysis is based on a variational AC language with precisely defined structure and semantics, allowing us to provide a more granular form of analysis which is nonetheless amenable to verification.

% \vspace{-4mm}
\section{Conclusion}
\label{sec:conclusion}
{In this work, we formalized lifted AC development, an assurance engineering methodology  for software product lines. By formalizing a language of variability-aware ACs, we extend existing formal approaches to AC development by lifting formal argument templates and software analyses to the SPL-level. We demonstrated the feasibility and usefulness of our methodology by developing a variational AC over a product line of medical devices.}

\nc{We identify several avenues for future work. One is to further extend our model-based development framework with support for additional lifted templates and analyses, and to integrate the modeling layer with theorem proving support, which would enable template instantiations to be verified automatically during AC development~\cite{viger2023foremost}. Another important line of work is the \emph{repair} of argument templates over model evolutions. Our template regression procedures allow for engineers to specify new instantiations for templates given a change to a model (Algs.~\ref{alg:templateRegression},\typofix{and} ~\ref{alg:templateRegressionLift}), but does not specify \emph{how} these instantiations should be produced. This introduces an opportunity for argument repair and synthesis, which might be approached deductively (based on the semantics of the template being repaired), or could potentially be approached through generative AI.}

\bibliographystyle{splncs04}\bibliography{main,aren}
\newpage
\appendix 
\section{Proofs of Selected Theorems}
\label{app:proofs}
\subsection{Proofs for Section 3}
\label{app:proofs:sec3}
\subsubsection*{Proof of Lemma~\ref{def:suppCorrect}}

\begin{lemma*}
Let $A$ be an assurance case \weakDelete{for system $S$} such that\; ${\tt Rt}(A) = g$. If $A$ is supported, then we can infer $g$.
\end{lemma*}
\begin{proof}
    {Assume $\Supp(A)$, let ${\tt Rt}(A)=g$, and proceed by induction on $A$. It cannot be the case that $A$ is an undeveloped goal, since there is no inference rule \weakChange{for $A = {\tt Und}(g)$}{from which $\Supp({{\tt Und}(g)})$ can be derived}. If $A = {\tt Evd}(g,e)$, we can conclude $g$ from the evidence $e$, per [$\Supp$-1] and [$\Supp$-2]. Otherwise, $A = {\tt Decomp}(g,st,\{A_1\ellipses A_n\})$ such that for each $A_i$, we have $\Supp(A_i)$. By the inductive hypothesis, we can conclude \weakChange{$\forall i \leq n,g_i$, where $g_i = {\tt Rt}(A_i)$}{that each goal $g_i= {\tt Rt}(A_i)$ holds}. By $[\Supp$-3], we have $\{A_1\ellipses  A_n\} \prec g$, which thus allows us to conclude $g$.}
\end{proof}

\subsection{Proofs for Section 4}
\label{app:proofs:sec4}
\subsubsection*{Proof of Thm.~\ref{thm:supplift}}
\begin{theorem*}
    If {\rm $\pl{A}$} is well-formed, \weakNew{then} $\Supp^\uparrow(${\rm $\pl{A}$}) \weakChange{$\equiv$}{if and only if}  $\;\mathsf{Inv}(\Supp,$ {\rm $\pl{A}$}$,\phi$), where $\phi = \token{pc}(${\rm \pl{A}}$)$.
\end{theorem*}

\begin{proof}
    {Suppose that $\pl{A} \in \Var{\AC}$ is well-formed. Let $\phi = {\tt pc}(\pl{A})$.}

    \textbf{Forward direction}. {Assume $\Supp^\uparrow(\pl{A})$, and proceed by induction on $A$. 
    
    If $\pl{A} = {\tt Evd}(g,e)$ with goal $g$ being propositional, then this proposition does not vary across configurations, so (by [$\Supp^\uparrow$-1]) $e$ is adequate for $P$ across all configurations satisfying $\phi$. If $g$ is a predicative goal, then by [$\Supp^\uparrow$-2], we have $e$ as adequate evidence of $\pl{g}|_\conf$ for each $\conf \in \llbracket \phi \rrbracket$. \weakNew{In either case, we have } ${{\Supp}}(\pl{A}|_\conf)$ holds for each $\conf \in \llbracket \phi \rrbracket$. 
    
    Otherwise, let $A = {\tt Decomp}(\pl{g},st,\mathcal{A})$. \weakNew{By $[\Supp^\uparrow$-3], we have 
    \begin{enumerate}
        \item $\Supp^\uparrow(\pl{A}_i)$ for each $\pl{A}_i \in \mathcal{A}$, and
        \item $\mathcal{A}|_\conf \prec \pl{g}|_\conf$ for each $\conf \in \llbracket \phi\rrbracket$. 
    \end{enumerate} Fix $\conf \in \llbracket \phi \rrbracket$, for which we need to show $\Supp(\pl{A}|_\conf)$.} By the inductive hypothesis, for each $\pl{A}_i$ with $\phi_i = {\tt Rt}(\pl{A}_i)$, \weakNew{from (1)} we have $\Inv{\Supp}{\pl{A}_i}{\phi_i}$. Per $[\Supp$-3], to show $\Supp(\pl{A}|_\conf)$, we need to show $\Supp(\pl{A}_i|_\conf)$ for each $\pl{A}_i|_\conf \in \mathcal{A}|_\conf$, and show $\mathcal{A}|_\conf \prec \pl{g}|_\conf$. The latter is given by (2). The former comes from the inductive hypothesis.}
    
    \textbf{Backwards direction}. {Assume $\Inv{\pl{A}}{\Supp}{\phi}$, and proceed by induction on $\pl{A}$. The base cases, for both propositional and predicative goals, are identical to the forward {\ntn{case}}. In the inductive case $\pl{A} = {\tt Decomp}(\pl{g},st,\mathcal{A})$, we need to \weakNew{apply [$\Supp^\uparrow$-3], and therefore need to} show  
    \begin{enumerate}
        \item $\forall \pl{A}_i \in \mathcal{A}, \Supp^\uparrow(\pl{A})$, and 
        \item $\Inv{\prec}{\langle\mathcal{A},\pl{g}\rangle}{\phi}$.
    \end{enumerate}
    By the inductive hypothesis, to show (1) for any given $\pl{A}_i$, it suffices to show that $\Inv{\Supp}{\pl{A}_i}{\phi_i}$ holds, where $\phi_i = {\tt Rt}(\pl{A}_i)$. Fix such an $\pl{A}_i$ and $\phi_i$. Fix $\conf \in \llbracket \phi_i \rrbracket$; since $\pl{A}$ is well-formed, we also have $\conf \in \llbracket \typofix{\conf}{\phi}\rrbracket$. Since we know $\Inv{\Supp}{\pl{A}}{\phi}$, we know in particular $\Supp(\pl{A}|_\conf)$, which in turn implies $\Supp(\pl{A}_i|_\conf)$ (per [$\Supp$-3]). By generalization of $\conf$ and \weakNew{$\pl{A}_i$}, and the \typofix{IH}{inductive hypothesis}, we have (1). To obtain (2), let $\conf \in \llbracket \conf \rrbracket$ be arbitrary, consider the refinement $\mathcal{A}|_\conf \prec \pl{g}|_\conf$, and see that this is provided by our assumption $\Supp(\pl{A}|_\conf)$ and $[\Supp$-3].}
\end{proof}

\subsubsection*{Proof of Prop.~\ref{prop:enumLiftCorrect}}
\begin{proposition*} The instantiation function of the variational domain decomposition (Example~\ref{ex:vdomdecomp}) correctly lifts the instantiation function defined in Example~\ref{ex:domainDecomp2}.
\end{proposition*}
\begin{proof}
    \nc{Let ${{\mathcal{F}}} = \{\pl{X}_1,...,\pl{X}_n\}$ be the family of variational sets used to produce subgoals $\{\pl{g}_1,...\pl{g}|_n\}$ through $I^\uparrow$ \weakChange{with presence condition $\phi$}{to support a parent with presence condition $\phi$}. Given any $\conf \in \llbracket \phi \rrbracket$, we need to show that $\{\pl{g}_1|_\conf,...\pl{g}_n|_\conf\} = \{g_1,...,g_n\} = I(\{\pl{X}_1|_\conf,...,\pl{X}_n|_\conf\})$. Note that there is no structural variability in the argument, as each goal is over the same presence condition $\phi$. So fix some $\pl{g}_i$, and see that $\pl{g}_i|_\conf = \langle \pl{X}_i, {\ntn{\mathsf{Forall}_P}},\phi\rangle|_\conf = \langle \pl{X}_i|_\conf,{\ntn{\mathsf{Forall}_P}}\rangle = g_i$.}
\end{proof}

\subsection{Proofs for Section 5}
\label{app:proofs:sec5}

\subsubsection*{Proof of Lemma~\ref{lem:templateRegCorrect}}
\begin{lemma*}
% \lm{REVISE}
    Let ${\mathtt{Decomp}}(g,st,\mathcal{A})$ be a decomposition of parent goal $g = \langle M,P\rangle$, and let $M' \in \modlang$ be the evolved version of model $M \in \modlang$. Let $\template{\modlang}{D} = \langle P,I,C\rangle$ be the (valid) template used to produce this decomposition via instantiation with $x \in D$. Let $v_{st}$ be the regression value returned by $\mathsf{TemplateRegression}(M',st,\mathcal{A},\template{\modlang}{D},x)$. Then:
    \begin{enumerate}
        \item $v_{st} = \cmark$ iff the decomposition remains sound after updating $g$ and $\mathcal{A}$, and (if applicable) re-instantiating $\template{\modlang}{D}$ for $M^\prime$,
        \item $v_{st} = \xmark$ iff re-instantiating $\template{\modlang}{D}$ for $M^\prime$ produces at least one new subgoal not provided in $\mathcal{A}$ which is required to make the decomposition sound, and 
        \item $v_{st} = ~?$ iff the template could not be re-instantiated for $M'$.
    \end{enumerate}

\end{lemma*}
\begin{proof}
    \weakNew{Let $g, M', st, \mathcal{A}, \template{\modlang}{D}$ and $x$ be as described in the statement of the lemma and let $v_{st}$ be the regression value returned by Alg.~\ref{alg:templateRegression}. Proceed by case analysis:}
    \begin{itemize}
        \item \weakNew{If no $x'$ was available to satisfy the correctness criterion for the template, we returned ? as needed on Line 8, and this is the only manner by which we can obtain  $v_{st} = ~?$.}
        \item \weakNew{Suppose $x'$ was available which satisfied $C$ for $M'$. There are two sub-cases.} 
        \begin{itemize}
            \item \weakNew{If $st$ is analytic with respect to $f : X \to Y$ and evidence-producing, we merely replace each occurrence of $x$ with $x'$ in the subgoals $\{g_X, g_Y\}$ of the template. Since analytic templates always produce the same set of subgoals (modulo the updated input to the analysis), no additional subgoals will be required. Since the template is valid and $C(x',M')$ holds, the strategy remains sound. Thus, the $\mathsf{Match}$ operator on Line 9 will return $\mathcal{A}_{\tt New} = \emptyset$ and $v_{st}$ will be assigned $\cmark$ on Line 11 as needed.}
            \item \weakNew{Alternatively, suppose $st$ is either analytic and argument producing or non-analytic. If it is analytic and argument-producing, the argument for the correctness of $v_{st}$ is the same as the previous case. If $st$ is non-analytic, then the re-instantiation on Line 5 may or may not introduce some new subgoals relative to $\mathcal{A}$. If it does not, then we know the strategy remains sound since $\template{\modlang}{D}$ is valid and we have $C(x',M')$. Otherwise, we know that proving the soundness of the strategy requires at least one additional subgoal not present in $\mathcal{A}$, as identified by $\mathsf{Match}$ on Line 9. Hence we will have $\mathcal{A}_{\tt New} \neq \emptyset$ and $v_{st}$ will be assigned $\xmark$ on Line 11 as needed.}
        \end{itemize} 
    \end{itemize}
\end{proof}
\subsubsection*{Proof of Thm.~\ref{thm:instAlgCorrect}}
\begin{theorem*}
   Let $A$ be an $AC$ such that $\Supp(A)$, and let $A_R$ be the reusable core extracted from $A$ following the forward pass of the analysis (Alg.~\ref{alg:instantiation}) given evolution set $\Delta$. Let $v$ be the regression value returned by $\mathsf{BackwardPass}(\Delta,A_R)$ (Alg.~\ref{alg:product_regression}). Then if $v = \cmark$, we know that $\Supp(A_R)$ holds;  if $v = \xmark$, we know that $\Supp(A_R)$ does not hold; and if $v = ~?$, we can make no determination about whether $\Supp(A_R)$ holds. 
   
\end{theorem*}
In the following proof, given any ${A} \in {\AC}$, let $Height(A) \in \mathbb{Z}^+$ be the number of goals in the longest branch from the root goal of ${A}$ to any leaf. Formally, $Height({\tt 
 Evd}(g,e)) = Height({\tt Und}(g)) = 1$ and $Height({\tt Decomp}(g,st,\mathcal{A})) = 1 + \max \{Height({A}_i) \mid A_i \in \mathcal{A}\}$.
\begin{proof} 
\weakNew{Let $A$ and $A_r$ be as given in the statement of the theorem, and let $v$ be the regression value returned by Alg.~\ref{alg:product_regression}. Proceed by strong induction on $n = Height(A_R)$:}
\begin{itemize}
    \item \weakNew{If $n=1$, we must have either $A_R = {\tt Und}(g')$ or $A_R = {\tt Evd}(g',e)$. Suppose $A_R$ is an undeveloped goal. By construction of $A_R$, and since $\Supp(A)$ holds, we know that this undeveloped goal is is the parent of a strategy all of whose children were identified as obsolete, in which case $g'$ is known to be unsupported. Thus, Lines 4-5 annotate $g'$ as ``Revise'' and return $\xmark$ as needed.

    If instead $A_R$ is a goal $g'$ supported by evidence $e$, we invoke Alg.~\ref{alg:productEvdRegression} on $g'$ and $e$. We can proceed by case analysis on the regression value returned by the evidence regression procedure.}
    \begin{itemize}
        \item \weakNew{The evidence regression procedure returns $\cmark$ only in one of three cases: either (i) $g'$ is propositional, in which case $e$ remains adequate by assumption; or (ii)  goal $g'$ is predicative and was not affected by the evolution;  or (iii) $g'$ was modified, but a regression analysis was employed on Line 5 which returned $\cmark$ for $M'$. Since we know that $e$ was adequate for the original goal $g$, this result allows us to conclude $P(M')$ from $e$, and hence $\Supp({\tt Evd}(g',e))$. Thus, returning $\cmark$ on Line 7 of Alg.~\ref{alg:product_regression} is correct.}
        \item \weakNew{If the evidence regression procedure returns $\xmark$, then $g'$ must be predicative and a regression analysis $R_P$ must have been applied to $M'$ which returned $\xmark$. Thus, we know (by the definition of $R_P$) that $P(M')$ cannot be inferred, so $g'$ is no longer supported. Thus, returning $\xmark$ on Line 7 of Alg.~\ref{alg:product_regression} is correct. 
        \item Otherwise, if the evidence regression procedure returns $?$, then either no regression analysis $R_P$ was available, or one was applied and returned $?$. In either case, returning ? on Line 6 of Alg.~\ref{alg:product_regression} is correct.}
    \end{itemize}   
    \item \weakNew{In the inductive case, assume $n > 1$, thus $A = {\tt Decomp}(g',st,\mathcal{A})$ with goal $g'$, strategy $st$, and  children $\mathcal{A}$ such that for all $A^\prime_i \in \mathcal{A}$, we have $Height(A_i^\prime) < n$.}

    \weakNew{Let $v_{st}$ be the annotation of the strategy produced during the forward pass. If $v_{st} = ~?$, then this strategy could not be analyzed for regression (either as a result of Lines 16-17 of Alg.~\ref{alg:instantiation} or by Lines 11-12 of Alg. ~\ref{alg:instantiation} and Lemma~\ref{lem:templateRegCorrect}). Moreover, all of its descendants are also annotated with ``Recheck''. Thus, we cannot determine whether parent goal $g$ is supported. Lines 11 and 12 thus annotate $g$ and return regression value $?$ as needed.}

    \weakNew{Otherwise, suppose the conditional on Line 10 fails. After Line 13, $V = \{v_{st}\}$, where $v_{st}$ correctly reflects the regression of $st$ by Lemma~\ref{lem:templateRegCorrect}. After Line 14, we add to $V$ the regression values computed for each child $A_i'$. By the inductive hypothesis, each of these values $v_i$ correctly reflects the regression of support for each $A_i'$. We then take $v = \mathsf{Min}(\{v_1'\ellipses v_n', v_{st}\})$. Proceed by case analysis on $v$.}

    \begin{itemize}
        \item \weakNew{If $v = \cmark$, then we must have $v_{st} = \cmark$ and each $v_i = \cmark$, by the ordering $\xmark < ~? <\cmark$. That is, the strategy remains sound, and each child AC $A_i'$ remains supported, so we have $\Supp({\tt Decomp}(g',st,\mathcal{A}))$ by definition of $\Supp$. Thus we correctly return $v = \cmark$ on Line 17.}

        \item \weakNew{If $v = \xmark$, then either $v_{st} = \xmark$ or there is some $A_i' \in \mathcal{A}$ such that $v_i = \xmark$. In the first case, the strategy $st$ is missing some subgoals necessary to prove the goal refinement, so we do not have $\Supp({\tt Decomp}(g',st,\mathcal{A}))$. In the second case, there is some child AC which is not supported, so we also do not have $\Supp({\tt Decomp}(g',st,\mathcal{A}))$. In either case, we correctly return $\xmark$ on Line 17.}

        \item \weakNew{If $v = ~?$, then there must be some child $A_i' \in \mathcal{A}$ such that $v_i = ~?$, but no child $A_j'$ such that $v_j = \xmark$, and neither do we have $v_{st} = \xmark$. Thus, there is at least one child AC for which support cannot be determined, but we have no way to conclude that  $\Supp({\tt Decomp}(g',st,\mathcal{A}))$ does certainly \emph{not} hold. Thus we cannot determine or falsify $\Supp({\tt Decomp}(g',st,\mathcal{A}))$, so we correctly return ? on Line 17.}
    \end{itemize}
\end{itemize}

\end{proof}

\subsection{Proofs for Section 6}
\label{app:proofs:sec6}

\subsubsection*{Proof of Lemma~\ref{lem:templateRegLift}}

\begin{lemma*}

    Let $\pl{g} = \langle \pl{M}, P, \phi \rangle $ be a variational goal decomposed using the valid lifted template $\template{\modlang}{D}^\uparrow$, producing strategy $st$ and children $\mathcal{A}$ via instantiation with $\pl{x} \in \Var{D}$. Let $\pl{M}' \in \Var{\modlang}$ be the updated version of $\pl{M} \in \Var{\modlang}$ following a system evolution. Then for every $\conf \in \llbracket \phi \rrbracket$, we have 
    $$\mathsf{TemplateRegression}^\uparrow(\pl{M}',\phi,st,\mathcal{A},\template{\modlang}{D}^\uparrow,\pl{x},\phi_{\Delta})|_\conf = {\mathsf{TemplateRegression}(\pl{M}'|_\conf, st, \mathcal{A}|_\conf, \template{\modlang}{D}, \pl{x}|_\conf)}$$
    That is, Alg.~\ref{alg:templateRegressionLift} correctly lifts Alg.~\ref{alg:templateRegression}.
\end{lemma*}
\weakNew{Note that, for the sake of simplicity, in the proof of the lemma, we assume that identifying a new variational input $\pl{x}'$ with which to re-instantiate the template is ``all-or-nothing''. That is, if we fail to find such an input for the entire set of configurations, we assume we do not an analogous product-based input  for \emph{any} of the individual configurations.}

\begin{proof}
    \weakNew{Let $\pl{g}$, $\template{\modlang}{D}^\uparrow$, $st$, $\mathcal{A}, \pl{x}$ and $\pl{M}^\prime$ be as described in the statement of the lemma. Let $\conf \in \llbracket \phi  \rrbracket$ and let $\pl{v}_{st}$ be the regression value returned by the lifted template regression analysis. We proceed by case analysis on $\pl{v}_{st}|_\conf$.}

    \begin{itemize}
        \item \weakNew{Suppose $\pl{v}_{st}|_\conf = \cmark$. Note that we can either have $\conf \vDash \phi_{\Delta}$ or $\conf \nvDash \phi_{\Delta}$.

        If $\conf \nvDash  \phi_{\Delta}$, then applying the product-based template regression analysis under configuration $\conf$ succeeds trivially, since we can simply re-instantiate the template with the same input $\pl{x}'|_\conf = \pl{x}|_\conf$, which will result in a completely identical set of subgoals. Thus the product-based analysis will return $\cmark$ as needed.
        
        Suppose instead that we have $\conf \vDash \phi_\Delta$, so we must have returned on Line 15, with $\conf \vDash \phi \land \neg \phi_{\tt New}$. By the definition of $\phi_{\tt New}$, there cannot be any $\pl{A}_i \in \mathcal{A}_{\tt New}$ such that $\conf \vDash {\tt pc}(\pl{A}_i)$. By the construction of $\mathcal{A}_{\tt New}$, this in turn implies that instantiating $I^\uparrow(\pl{x}')$ did not introduce any new variational subgoals present under $\conf$. Since $I^\uparrow$ is a lift of $I$, applying $I$ on $\pl{x}'|_\conf$ also would not have introduced any new subgoals. Hence, applying the product-based template regression for $\pl{M}|_\conf$ using $\pl{x}|_\conf'$ would have returned $v_{st} = \cmark$, as needed.}
        \item \weakNew{Suppose $\pl{v}_{st}|_\conf = \xmark$. Then we must have returned $\pl{v}_{st}$ on Line 15, with $\conf \vDash \phi \land \phi_{\tt New}$. By analogous reasoning to the above case, applying the product-based template regression for the strategy derived under $\conf$ \emph{would} have identified at least one missing subgoal, and thus the product-based analysis would have returned $v_{st} = \xmark$ as needed.}

        \item \weakNew{Finally, if $\pl{v}_{st} = ?$, then we must have returned on Line 9, with $\conf \vDash \phi \land \phi_\Delta$, meaning we could not produce a new input for instantiating the lifted template. By the ``all-or-nothing'' assumption noted above, we assume no analogous product-based input $\pl{x}|_\conf$ can be produced for configuration $\conf$. Thus, when applying the product-based regression analysis under configuration $\conf$, we would similarly return $v_{st} = ~?$ as needed.}
    \end{itemize}
\end{proof}

\subsubsection*{Proof of Thm.~\ref{thm:forwardsLift}}
\begin{theorem*}
    Let $\pl{A} \in \Var{\AC}$ with $\token{pc}(\pl{A}) = \phi$. Then for every $\conf \in \llbracket \phi \rrbracket$, we have 
    $$\mathsf{ForwardPass^\uparrow( \pl{A},\widehat{\Delta})|_\conf} = \mathsf{ForwardPass}(\pl{A}|_\conf,\widehat{\Delta}|_\conf)$$
    That is, Alg.~\ref{alg:instantiationLift} correctly lifts Alg.~\ref{alg:instantiation}.
\end{theorem*}
In what follows, let $Height : \Var{\AC} \to \mathbb{Z}^+$ be defined exactly as for product-based ACs, i.e., $Height({\tt 
 Evd}(\pl{g},e)) = Height({\tt Und}(\pl{g})) = 1$ and $Height({\tt Decomp}(\pl{g},st,\mathcal{A})) = 1 + \max \{Height(\pl{A}_i) \mid A_i \in \mathcal{A}\}$. \weakNew{We also assume, for simplicity, that all formal templates are  lifted; that is, there is no argument in the variational AC which was created \emph{manually}, but for which there is an equivalent product-level template.}
\begin{proof}
    \weakNew{Let $\pl{A}$ be the variational AC and ${\ntn{\widehat{\Delta}}}$ the variability-aware evolution set provided as input to the forward pass. We proceed by strong induction on $n = Height(\pl{A})$. Let $\phi = {\tt pc}(\pl{A})$.

    If $n = 1$, then we know either $\pl{A} = {\tt Evd}(\pl{g},e)$ or $\pl{A} = {\tt Und}(g)$. In either case, all we do is update $\pl{g}$ if the model is updated, which we mirrors the product-based algorithm whenever we are in a configuration $\conf$ where $M|_\conf \in {\ntn{\widehat{\Delta}}}|_\conf$.}

    \weakNew{Otherwise, $n > 1$ and we  have $\pl{A} = \token{Decomp}(\pl{g},st,\mathcal{A})$ for some variational goal $\pl{g}$, strategy $st$, and children $\mathcal{A}$, each of which has height less than $n$. By the inductive hypothesis, the lift is correct for each $\pl{A}_i \in \mathcal{A}$.

    If $st$ is not template-based, then we annotate all descendants of $\pl{g}$ as ``Recheck'' (Lines 19-20). By assumption, the corresponding product-level strategy under any $\conf \in \llbracket \phi\rrbracket$ cannot be template-based, and so all descendants of $\pl{g}$ present under such a $\conf$ would have been marked as ``Recheck'' by the product-level algorithm.

    Otherwise, $st$ is template-based. Suppose that Line 9 of the lifted procedure succeeds, i.e., $\phi \wedge \phi_\Delta$ is not satisfied by any configuration. Then for any $\conf \in \llbracket \phi \rrbracket$, the corresponding execution of the product-based analysis would also have succeeded (cf. Line 7 of Alg.~\ref{alg:instantiation}), since $\pl{M}|_\conf \notin \widehat{\Delta}|_\conf$. In such a case, the product AC derived under each such $\conf$ would have $st$ annotated by $\cmark$ before proceeding recursively; thus, Line 10 correctly annotates $st$ with ${\mathsf{REUSE}}(\phi)$ and the conclusion follows form the inductive hypothesis.}

     \weakNew{Finally, if $\phi \wedge \phi_\Delta$ is satisfiable, we invoke the lifted template regression analysis. The resulting regression value $\pl{v}_{st}$ is correct ``lifted'' result with respect to the product-based template analysis (Lemma~\ref{lem:templateRegLift}). In the product setting, we do not analyze any children of the strategy if the strategy is marked as ``Recheck'', or if the children are obsolete; this is reflected by Lines 15 and 16, and the conclusion follows from the inductive hypothesis.}

\end{proof}

\subsubsection*{Proof of Lemma~\ref{lem:evdRegLiftCorrect}}
\begin{lemma*}

    Let $\widehat{\Delta}$ be a variational evolution set, $\pl{g}'$ be a variational goal with presence condition $\phi$, and $e$ be adequate evidence for $\pl{g}'$. Then for all configurations $\conf \in \llbracket \phi \rrbracket$, we have 
    $$\mathsf{EvdRegression}^\uparrow(\widehat{\Delta},\pl{g}',e)|_\conf = \mathsf{EvdRegression}(\widehat{\Delta}|_\conf, \pl{g}|_\conf, e)$$
    i.e., Alg.~\ref{alg:varEvdRegression} correctly lifts Alg.~\ref{alg:productEvdRegression}.
\end{lemma*}
\weakNew{In the following proof, we assume for simplicity a one-to-one correspondence between lifted and product-based regression analyses; i.e., if we do not have access to a lifted regression analysis for some evidence, then we also do not have a corresponding product-based regression analysis.}

\begin{proof}
    \weakNew{Let $\widehat{\Delta}, \pl{g}',\phi,$ and $e$ be as described in the statement of the lemma. First, if $\pl{g}'$ is propositional, or refers to a model which is not modified under any configurations satisfying $\phi$, then the procedure returns ``Reuse'' for all configurations satisfying $\phi$. The product-based analysis would therefore return $\cmark$ on Line 2 of Alg.~\ref{alg:productEvdRegression} for each of these configurations, as needed.}

    \weakNew{Otherwise, we either do or do not apply a lifted regression analysis $R_P^\uparrow$; if not, we return the regression value $\langle \phi \wedge \neg \phi_\Delta, \bot, \phi \wedge \phi_\Delta \rangle$, where $\phi_\Delta$ describes the set of configurations modified under $ \widehat{\Delta}$. For any $\conf \in \llbracket \phi\rrbracket$, if $\conf \vDash \phi_\Delta$, the product-based analysis in Alg.~\ref{alg:productEvdRegression} would have returned $\cmark$ on Line 2, as needed, or returned $?$ on Line 5 (under the assumption that we also do not have a corresponding product-based regression analysis for this evidence). As such, we return the correct values under both sets of configurations.

    Finally, if we do apply a lifted regression analysis $R_P^\uparrow$, the correctness of the result is provided by the correctness of $R_P^\uparrow$.}
\end{proof}

\subsubsection*{Proof of Lemma~\ref{lem:compBinCorrect}}

\begin{lemma*}
    Let $\pl{v}_1 = \langle \phi_\cmark, \phi_\xmark, \phi_?\rangle$ be a regression value partitioning $\llbracket \phi\rrbracket$, and $\pl{v}_2 = \langle \psi_\cmark, \psi_\xmark, \psi_?\rangle$ be a regression value partitioning $\llbracket \psi \rrbracket$. Then $\pl{v}_1 \otimes \pl{v}_2$ is a variability-aware regression value over $\llbracket \phi  \vee \psi \rrbracket$ such that for all $\conf \in \llbracket \phi \vee \psi \rrbracket$, we have $(\pl{v}_1 \otimes \pl{v}_2)|_\conf = \mathsf{Min}(\{\pl{v}_1,\pl{v}_2\}|_\conf)$.
\end{lemma*}

\begin{proof}
    
    \weakNew{Let $\pl{v}_1,\pl{v}_2,\pl{v}$ be as described in the statement of the lemma.     
    Let $\pl{v} = (\pl{v}_1 \otimes \pl{v}_2)$.
    
    Given that $\pl{v}_1$ partitions $\llbracket \phi \rrbracket$ and $\pl{v}_2$ partitions $\llbracket \psi\rrbracket$, it is clear that the $\pl{v}$ partitions the union of these sets. It remains to be shown that the result derived under any $\conf$ in this set is the correct minimum. Proceed by case analysis on $\pl{v}|_\conf$:}
    \begin{itemize}
        \item \weakNew{If $\pl{v}|_\conf = \cmark$, then either (i) $\conf$ is in exactly one of $\llbracket \phi \rrbracket$ or $\llbracket \psi\rrbracket$, in which case $\mathsf{Min}^\uparrow$ only receives one input $\pl{v}_i$, which must have $\pl{v}_i|_\conf = \cmark$; or (ii)  $c \in \llbracket \phi \land \psi \rrbracket$, such that both $\pl{v}_1|_\conf = \cmark$ and $\pl{v}_2|_\cmark = \cmark$, in which case we are comparing it against $\mathsf{Min}(\cmark, \cmark)$, which is indeed $\cmark$}.

        \item \weakNew{The case for $\pl{v}|_\conf = \xmark$ is analogous, with the only distinction being that $\mathsf{Min}$ will return $\xmark$ so long as one of its inputs is $\xmark$.}
        \item \weakNew{If $\pl{v}|_\conf = ~?$, then we follow a similar line of reasoning as in case (a). The only need for care is in the case where $\conf \in \llbracket \phi \land \psi \rrbracket$, since in addition to checking that at least one of $\pl{v}_1|_\conf = ~?$ or $\pl{v}_1|_\conf = ~?$, we must ensure that the other does not reduce to $\xmark$ under $\conf$, as then $\mathsf{Min}$ would return $\xmark$. The third clause in the formula for the ``Recheck'' expression precludes such a situation.}
    \end{itemize}

\end{proof}

\subsubsection*{Proof of Thm.~\ref{thm:backwardsLift}}

\begin{theorem*}
    Let $\pl{A} \in \Var{\AC}$ be a variational AC, and let $\pl{A}_R$ be the reusable core extracted from $\pl{A}$ following the (lifted) forward pass of the analysis under variational evolution set $\widehat{\Delta}$. Let $\phi = \token{pc}(\pl{A}_R)$. Then for every $\conf \in \llbracket \phi \rrbracket$, we have 
    $$\mathsf{BackwardPass}^\uparrow(\widehat{\Delta},\pl{A}_R)|_\conf = \mathsf{BackwardPass}(\widehat{\Delta}|_\conf,\pl{A}_R|_\conf)$$
    That is, Alg.~\ref{alg:lifted_assurance_regression} correctly lifts Alg.~\ref{alg:product_regression}.
\end{theorem*}
\begin{proof}
    \weakNew{Let $\pl{A}_R, \widehat{\Delta},$ and $\phi$ be as given in the theorem statement. We proceed by strong induction on $n = Height(\pl{A}_R)$.}

    \weakNew{If $n = 1$, then we either have an undeveloped goal or a goal supported by evidence. In the first case, we annotate the goal by $\mathsf{REVISE}(\phi)$, mirroring the fact that the undeveloped goal would be annotated with $\xmark$ under every configuration of $\phi$ in the product-based analysis (cf. Lines 4-5 of Alg.~\ref{alg:product_regression}). In the second case, the conclusion follows directly from Lemma~\ref{lem:evdRegLiftCorrect}.}

    \weakNew{If $n > 1$, then we have $\pl{A}_R = \token{Decomp}(\pl{g}', st, \mathcal{A})$ for some variational goal $\pl{g}$, strategy $st$, and children $\mathcal{A}$, such that each $\pl{A}'_i \in \mathcal{A}$ has height less than $n$. By the correctness of the lifted forward pass, the annotation we receive for $v_{st}$ encodes the strategy annotations we would obtain under each configuration of $\phi$. If we determine that $\pl{v}_{st} = \mathsf{RECHECK}(\phi)$ on Line 10, we annotate $\pl{g}'$ with this value and return accordingly, mirroring the analogous branch of the product-based analysis for all configurations satisfying $\phi$. Otherwise, we add $\pl{v}_{st}$ to $V$  and obtain the regression values $\pl{v}_i$ for each $\pl{A}'_i \in \mathcal{A}$, which correctly encode the regression values for each of their respective sets of configurations by the inductive hypothesis.  The conclusion then follows from the correctness of $\mathsf{Min}^\uparrow$ (Lemma~\ref{lem:minLiftCorrect}).}
\end{proof}

\end{document}